\pdfoutput=1
\documentclass[preprint,12pt,authoryear]{elsarticle}

\newif\ifanon \anonfalse

\usepackage[T1]{fontenc}
\usepackage{lmodern}
\usepackage{amssymb}
\usepackage{amsmath}
\usepackage{amsthm}
\usepackage{mathtools}
\usepackage{bm}
\usepackage{booktabs}
\usepackage{enumitem}
\usepackage{framed}
\usepackage{graphicx}
\usepackage{adjustbox}
\usepackage{orcidlink}
\usepackage{microtype}
\newtheorem{theorem}{Theorem}
\newtheorem{lemma}{Lemma}
\newtheorem{corollary}{Corollary}
\newtheorem{proposition}{Proposition}
\theoremstyle{definition}
\newtheorem{assumption}{Assumption}
\newtheorem{definition}{Definition}
\theoremstyle{remark}
\newtheorem{remark}{Remark}

\DeclareMathOperator{\E}{\mathbb E}
\DeclareMathOperator{\Var}{Var}
\DeclareMathOperator{\Cov}{Cov}

\DeclareMathOperator{\tr}{tr}
\DeclareMathOperator{\rank}{rank}
\newcommand{\R}{\mathbb R}
\newcommand{\Pp}{\mathbb P}
\newcommand{\G}{\mathcal G}
\newcommand{\F}{\mathcal F}
\newcommand{\tX}{\widetilde X}
\newcommand{\hbeta}{\widehat\beta}
\newcommand{\hsigma}{\widehat\sigma}
\newcommand{\hQ}{\widehat Q}
\newcommand{\hrho}{\widehat\rho}

\newcommand{\hlambda}{\widehat\lambda}
\newcommand{\heta}{\widehat\eta}
\newcommand{\hpsi}{\widehat\psi}

\journal{Journal of Econometrics}

\begin{document}

\begin{frontmatter}

\title{Breakdown Reliability for Saturated Fixed-Effect Inference}

\ifanon
\else
\author{Stanis\l{}aw M. S. Halkiewicz~\orcidlink{0009-0000-7344-7522}\corref{cor1}}
\cortext[cor1]{Corresponding author.}
\ead{stashal@o2.pl}
\affiliation{organization={Group of Machine Learning Research (GMUM), Jagiellonian University},
            city={Krak\'ow},
            country={Poland}}
\fi

\begin{abstract}
Fixed-effect saturation does not itself distort conventional inference, but classical measurement error does. Under a local noise drift $\sigma_\nu^2=c^2/n$, the FE-OLS $t$-statistic converges to a non-central normal; saturation contributes a common $\sqrt{1-\rho}$ scaling rather than preferentially destroying signal or noise. Inverting the size distortion gives a Stock--Yogo-style critical value. Self-consistency of the within-reliability-corrected pilot yields a breakdown reliability $\lambda^{\dagger}=|t|/(|t|+\eta^{\dagger})$ --- the minimum within reliability at which conventional inference retains nominal size within the chosen tolerance --- computable from the reported $t$-statistic alone and algebraically $\rho$-free conditional on it; $\eta^{\dagger}\approx0.65$ at $5\%$ size and a 5-point tolerance. Replacing $|t|$ by $|t|+z_{1-\gamma_\beta}$ gives a certified breakdown reliability; with a lower-reliability bound whose coverage error is $\gamma_\lambda$, false certification is at most $\gamma_\beta+\gamma_\lambda$. Under a checkable projection-compatibility condition, a cluster-level score CLT and consistency of the Arellano variance estimator in the many-fixed-effect regime justify applying the same map to the reported cluster-robust $t$-statistic; clustering can reverse a verdict. In a saturated democracy--growth panel, aggregate V-Dem polyarchy is certified at $\gamma_\beta=0.05$, conditional on the supplied measurement model, while its judicial-constraints sub-index is flagged under i.i.d.\ and clustered standard errors. In a twin-pair wage design, the specification is flagged under both independent and correlated reporting-error models, although implied coverage of the nominal-$95\%$ interval ranges from $8\%$ to $68\%$. The diagnostic covers classical error in a continuous regressor, not binary-treatment misclassification.
\end{abstract}

\end{frontmatter}

\section{Introduction}\label{sec-intro}

The weak-instruments literature has formalized a familiar empirical concern: when the first-stage relationship between instruments and an endogenous regressor is weak, two-stage least squares is biased, and the bias does not vanish at the conventional $\sqrt n$ rate. \citet{SS} showed that the bias is governed by the concentration parameter $\mu^2$; \citet{SY} tabulated critical values for the first-stage $F$-statistic that bound this bias below pre-specified tolerances; \citet{LMMP} and \citet{ASS} refined the practical recommendations.

This paper poses an analogous question for linear regression with high-dimensional fixed effects. Empirical work routinely reports specifications with progressively richer FE structure: worker--firm models \citep{AKM}, gravity equations with three-way fixed effects, event studies with unit-and-time-interacted controls. \citet{Mittag2019}, for example, computes exact OLS in worker--firm and worker--firm-match models with multiway-clustered standard errors; omitting match effects nearly doubles the estimated unexplained gender wage gap in his application. As FE saturate, the residual variation of the treatment shrinks:
\[
Q_K = \E[\Var(X \mid \G_K)] \to 0,
\]
where $\G_K$ is the $\sigma$-algebra generated by the $K$th FE specification. A natural conjecture is that small $Q_K$, like small $\mu^2$, produces a weak-identification problem requiring its own critical values.

The conjecture turns out to be false in the baseline model. Under strict exogeneity, FE-residualized OLS is unbiased regardless of $Q_K$. With conditionally Gaussian homoskedastic errors, the degrees-of-freedom-corrected $t$-statistic in the Cattaneo--Jansson--Newey many-covariate regime (with residual variance formed after projecting out the treatment) has an exact $t_{n - d_K - 1}$ distribution; its size distortion relative to $N(0,1)$ is $O(1/n)$ and does not depend on $\tau^2 = nQ_K$. The IV analogy breaks at its foundation, because TSLS bias arises from the first-stage residual appearing in both the numerator and denominator of the second-stage estimator, while no analogous mechanism is present in linear FE.

\ifanon
\else This paper and \citet{HalkiewiczCycle2026} form\fi a
shared program of adequacy
diagnostics for saturated fixed-effect designs: the companion treats
concentrated identifying variation with exact nuisance-annihilating contrasts,
this paper treats noisy continuous regressors, and the same implementation also
screens heterogeneous TWFE effects. The point is diagnostic separation---small
within variation alone is not weak identification, but concentrated scores and
measurement error are distinct failures requiring distinct remedies.

The conjecture becomes true once we introduce a bias source. We study the canonical one: classical measurement error in the treatment. Its empirical relevance is not marginal: linking administrative records to three major household surveys, \citet{MeyerMittag2021} find total survey error equal to $4$--$60\%$ of actual transfer dollars and measurement error to be the largest component in every survey--program comparison; without it, combined error from the remaining components is below $10\%$ throughout. Observed $X_i^* = X_i + \nu_i$ replaces the true $X_i$, with $\nu_i$ a mean-zero noise of variance $\sigma_\nu^2$ independent of $(X_i, u_i)$. Under Condition~(B), FE saturation asymptotically absorbs the same fraction $\rho$ of the noise variance and the true signal variance (exactly so under (B1)), leaving the within reliability unaffected by $\rho$ and the residual identifying variation smaller by the common factor $(1-\rho)$. The resulting diagnostic separates two adjacent objects: within reliability $\lambda$ is a property of the data, whereas breakdown reliability $\lambda^{\dagger}$ is the minimum requirement conventional inference places on that property. We work under the local-to-zero drift $\sigma_\nu^2 = c^2/n$, which keeps the noise commensurable with the residual treatment variance, and show that the FE-OLS $t$-statistic converges to a non-central normal:
\[
T_n^* \Rightarrow N(\eta, 1), \qquad
\eta = -\frac{\beta_0\, c^2\sqrt{1-\rho}}{\sigma \sqrt{\tau^2 + c^2}}.
\]
The non-centrality is non-zero precisely when $c > 0$, and decreasing in $\tau^2$, $\sigma$, and (at fixed $\tau^2$) in $\rho$ through the overall $\sqrt{1-\rho}$ scaling of the design's total identifying variation, not through any differential absorption of noise relative to signal. The two-sided size is even in $\eta$, so its leading distortion is quadratic, $z_{1-\alpha/2}\phi(z_{1-\alpha/2})\,\eta^2$; inverting it at a tolerance $\delta$ yields a closed-form Stock--Yogo critical value:
\[
\tau^2_{\mathrm{crit}}(\rho, c^2, \beta_0, \sigma; \alpha, \delta)
= \frac{\beta_0^2\, c^4 (1-\rho)\, z_{1-\alpha/2}\,\phi(z_{1-\alpha/2})}{\sigma^2\delta} - c^2.
\]
The operational threshold is zero when this algebraic boundary is negative.

\begin{framed}
\noindent\textbf{What the diagnostic certifies.} The non-centrality is proportional to $\beta_0$: classical attenuation shrinks a zero to a zero, so under $H_0:\beta_0=0$ there is no size distortion, and the significance test of \emph{no effect} is never the object at risk. This regularity is not generic to measurement-error procedures. In nonlinear EIV models, \citet{EvdokimovZeleneev2020Inference} show that nuisance parameters describing latent-variable or error distributions can become weakly identified as the structural coefficient approaches zero, invalidating standard inference even with strong instruments. They also identify the boundary relevant here: that mechanism does not arise in linear IV and is absent when the nuisance input is separately identified through repeated measurements or validation data. The present linear model, with its noise moment supplied externally, lies in that regular region. What measurement error distorts is inference about the magnitude: the $t$-test of $H_0:\beta=\beta_0$ at the true nonzero $\beta_0$ --- equivalently, the coverage of the conventional confidence interval, which is centered near the attenuated $\lambda\beta_0$ and misses the truth at the non-central rate. Throughout the paper, the object the diagnostic speaks to is confidence-interval coverage of $\beta_0$ --- the credibility of the reported magnitude, not the significance verdict; Definition~\ref{def-breakdown} encodes this by setting the breakdown reliability to zero when $\beta_0=0$. That definition also fixes the paper's two verdict words, which are never used interchangeably: a \emph{point pass} is descriptive, whereas a \emph{certificate} is conservative and accounts separately for coefficient and reliability uncertainty.
\end{framed}

A second question the reader should ask up front: if the within reliability $\lambda$ is knowable, why diagnose rather than simply correct --- report $\hbeta^*/\hlambda$ with standard errors inflated by $1/\hlambda$ and dispense with thresholds? The answer is the Stock--Yogo logic itself: \emph{the correction requires a point estimate of $\lambda$; the diagnostic requires only a lower bound}. The verdict ``conventional inference is adequate provided the data's within reliability exceeds the specification's breakdown reliability $\lambda^{\dagger}$, a bound the validation literature comfortably supports'' survives any imprecision in the noise pilot that keeps it above $\lambda^{\dagger}$, whereas the corrected estimator inherits the pilot's error one-for-one, together with its sampling noise (Corollary~\ref{cor-slope}). Certification is the robust use of weak side information; correction is the fragile use of strong side information.

The structure of this threshold differs from the IV case in three ways. First, it is two-dimensional in the regime parameter: $\tau^2$ enters through the standardization $\sqrt{\tau^2+c^2}$ exactly as in a signal-to-noise ratio, and $\rho$ enters separately, through the overall $\sqrt{1-\rho}$ discount common to signal and noise. Heavier FE saturation lowers the threshold, but not because it absorbs noise preferentially: it shrinks the design's total identifying variation, which raises the standard error at any fixed $\tau^2$. Second, the threshold depends on $|\beta_0|$ and $\sigma$ (the parameter being estimated and the residual standard deviation), so unlike the unit-free Stock--Yogo IV thresholds, this one is data-dependent. Third, the threshold requires an external estimate of the measurement-error variance $c^2$. This estimate typically comes from test-retest validation, published measurement-reliability ratios, or, for expert-coded indices, the measurement model's own posterior uncertainty; the diagnostic cannot be computed from the regression output alone.

\subsection{Scope and relation to prior work}\label{sec-scope}\label{sec-notjust}

\paragraph{Relation to Griliches--Hausman.} The interaction between panel transformations and measurement error is classical: \citet{GH} showed that within and difference transformations amplify attenuation bias, because they remove signal variance while (largely) preserving noise variance, and they proposed estimators (contrasts across differencing lengths, and instrumenting with lags) that exploit the differential amplification to identify and remove the bias. The present paper is complementary and answers a different question. \citet{GH} treat the problem as one of estimation: given that the bias exists, construct a consistent estimator. We treat it as one of \emph{inference adequacy}: given that the practitioner will run conventional FE-OLS on the observed regressor, characterize exactly when the resulting $t$-test remains size-controlled, and supply the critical value that certifies it. The two-dimensional $(\rho, \tau^2)$ threshold, the separate roles of the within reliability $\lambda$ and the FE dimension $\rho$ in it, and the local-drift non-centrality analysis are, to our knowledge, new. When a specification fails the threshold, the Griliches--Hausman estimators (or a fixed-effect Anderson--Rubin test, which we develop in companion work in progress) are the natural remedies; the diagnostic tells the practitioner whether a remedy is needed at all.

\paragraph{Classical error in a continuous regressor only.}

One restriction is important enough to state before the results rather than after them, because violating it inverts the diagnostic's conclusion rather than merely weakening it. Everything in this paper assumes \emph{classical} measurement error: $X^*_i = X_i + \nu_i$ with $\nu_i$ mean-zero and independent of $(X_i, u_i)$, formalized in Assumption~\ref{ass-nu}. The entire apparatus---the attenuation factor $\lambda$, the sign of the non-centrality, the breakdown reliability, the direction in which the naive pilot errs---rests on that independence. The scope also excludes Berkson error, $X_i=X_i^*+\nu_i$, under which the latent true value varies around an assigned or observed proxy rather than the proxy around truth; \citet{LiMa2024} review the distinction and models that combine both error types.

It fails, by construction and not by approximation, for a mismeasured binary treatment. If $X_i \in \{0,1\}$ then the misclassification error $\nu_i = X^*_i - X_i$ can only be $+1$ when $X_i = 0$ and $-1$ when $X_i = 1$, so $\Cov(\nu_i, X_i) < 0$ necessarily: the error is nonclassical, mean-independence fails, and the attenuation factor is no longer the classical measurement-reliability ratio. This is not a new observation --- \citet{Aigner1973} derived the attenuation for a mismeasured binary regressor and showed it does not take the classical form, and \citet{Bollinger1996} shows that what is identified under misclassification is a bound on the coefficient rather than a point-corrected estimate. \citet{MeyerMittag2017} make the same necessary-nonclassical point for a misclassified binary \emph{dependent} variable; their outcome-bias formulas do not transfer to a regressor, but the support argument does. This is precisely why a within-reliability threshold has nothing to certify there. Since binary treatments---union status, program participation, treatment take-up---are among the most commonly mismeasured regressors in fixed-effect panels \citep[see][on the union case]{Card1996}, we state the consequence plainly: \textbf{applying the diagnostic of this paper to a binary treatment is a misuse, and its verdict carries no warrant.} A measurement-reliability ratio computed for such a regressor does not index the object $\lambda$ that appears in Definition~\ref{def-breakdown}, and the resulting pass or flag is not interpretable. The applications below use continuous regressors (expert-coded democracy indices and years of schooling) throughout, and step~0 of the protocol in Section~\ref{sec-protocol} makes the check explicit. Extending the derivation pattern to nonclassical misclassification requires its own non-centrality calculation and is left to future work (Section~\ref{sec-ext}).

\paragraph{Not simply bias over standard error.}

The non-centrality $\eta$ is, algebraically, an attenuation bias divided by a standard error. Three things make it an econometric object rather than a ratio, and each is a place where the informal calculation gives the wrong answer.

\emph{First, the ratio is not a verdict until it has a limit experiment attached.} It becomes a statement about size only once one knows the law of the $t$-statistic under contamination; under the local drift that law is $N(\eta,1)$ (Theorem~\ref{thm-noncentral}). Without it, the natural conjecture --- that attenuation of the point estimate translates one-for-one into over-rejection --- is simply wrong. The two-sided size is even in $\eta$, so the leading distortion is quadratic and the critical value is not linear in the ratio; reasoning informally from ``bias over standard error'' misprices the threshold by an order in $\eta$.

\emph{Second, correcting and certifying need different information.} Correction requires a point estimate of $\lambda$ and inherits its error one-for-one; certification requires only a lower bound, because the verdict is monotone in $\lambda$. This is the Stock--Yogo logic, and it is what makes the exercise feasible with the side information validation studies actually supply --- within-reliability ranges, not point estimates. A bare ratio cannot exploit the asymmetry, having no threshold to be monotone against.

\emph{Third, the obvious way to compute the ratio is anti-conservative.} Plugging in the estimated coefficient uses the attenuated $\hbeta^*_K$, understating $|\eta|$ by the factor $\lambda$ and passing specifications that should fail (Proposition~\ref{prop-pilot}); Design~3 exhibits one whose true size is $16\%$. The error runs in the direction users care about, and removing it requires the within-reliability correction --- exactly the step a bookkeeping account omits.

The paper's most portable result is downstream of those primitives. Correcting the coefficient pilot by $\hbeta^*/\lambda$ makes the required threshold depend on the same supplied within reliability $\lambda$, so evaluating the comparison at a different value is not self-consistent. The fixed point resolves the problem exactly:
\[
\lambda^{\dagger}=\frac{|t^*_n|}{|t^*_n|+\eta^{\dagger}(\alpha,\delta)}.
\]
Conditional on the reported $t$-statistic, this formula is algebraically $\rho$-free and needs neither the underlying data nor a noise pilot; its inferential warrant still rests on the model assumptions. Its \emph{form} also survives the variance-estimator choice---one uses the corresponding reported $t$---but its \emph{value} need not, and clustering can reverse the verdict (Corollary~\ref{cor-cluster-feasible}(c)).

Our contributions are as follows. The first four establish the theoretical core -- the baseline no-threshold result, the non-centrality and its critical value, and the corrected pilot; the last two make the diagnostic usable on real, clustered panel data.

\begin{enumerate}[label=(\roman*)]
\item Formal demonstration that no $\tau^2$-driven size threshold arises in the baseline FE-OLS model (Theorem~\ref{thm-no-distortion}): FE saturation alone is not a weak-identification problem.
\item Derivation of the non-centrality parameter under EIV with local drift (Theorem~\ref{thm-noncentral}), with the FE dimension $\rho$ entering only through a common $\sqrt{1-\rho}$ discount of signal and noise alike (not through differential absorption), and a local power result (Proposition~\ref{prop-power}): power is attenuated by the factor $\sqrt{(1-\rho)\lambda}$ and shifted by $\eta$.
\item A feasible representation of the non-centrality, $|\eta| = (|\beta_0|/\sigma)(1-\lambda)\sqrt{\tau^{*2}}$ (Corollary~\ref{cor-feasible}), in which every quantity except one external within-reliability pilot is regression output; a closed-form critical-value formula (Corollary~\ref{cor-cv}); and the self-consistent breakdown reliability above, computable from the reported $t$-statistic alone (Definition~\ref{def-breakdown}).
\item A within-reliability-corrected pilot for $\beta_0$ (Proposition~\ref{prop-pilot}). Plugging the attenuated point estimate $\hbeta^*_K$ into the threshold understates $|\beta_0|$ by the factor $\lambda$ and makes the diagnostic anti-conservative; the corrected pilot removes this.
\item Specialized cluster-robust limit theory (Section~\ref{sec-cluster}), parallel to the general many-covariate framework of \citet{AnatolyevNg2026}: a cluster-level score CLT (Lemma~\ref{lem-cluster-clt}), consistency of the ordinary Arellano variance estimator under a checkable projection-compatibility condition (Lemmas~\ref{lem-nest}--\ref{lem-crve}), and the resulting diagnostic non-centrality $\eta_{CR}=\eta/\sqrt\psi$ (Theorem~\ref{thm-cluster}). Separately, a \emph{formal certification theorem} (Proposition~\ref{prop-certificate}) replaces $|t|$ by $|t|+z_{1-\gamma_\beta}$ and compares the resulting certified breakdown reliability with a supplied lower bound on $\lambda$; the false-certification probability is at most $\gamma_\beta+\gamma_\lambda$, without requiring independence between the coefficient and reliability bounds. Three by-products are of independent practical interest: the conventional $\tfrac{n-1}{n-K_n}$ small-sample factor must be omitted from this CRVE in saturated designs; the feasible rescaling remains valid even though its two ingredients are separately inconsistent (Corollary~\ref{cor-cluster-feasible}); and population $\psi$ and reported $\hpsi$ are distinct objects---$0.64$ versus $0.99$ in the same design---so clustering is not mechanically conservative and its applied direction must be read from the reported standard error (Remark~\ref{rem-psi-sign}).
\item Six simulation designs that calibrate the local-drift approximation, the finite-$n$ and clustered mappings, and the new certification guarantee---including a direct Monte Carlo check that false certification is no more frequent than its nominal bound---and two empirical applications with observable noise information: measurement-model posteriors in a democracy--growth panel and repeated schooling reports in a twin-pair wage design.
\end{enumerate}

Companion work in progress develops a complementary tool: a fixed-effect Anderson--Rubin test \citep[after][]{AndersonRubin} with uniform validity over $\tau^2 \in (0, \infty]$, providing valid inference precisely for the specifications that fail this paper's threshold.

\subsection{Related literature} Classical measurement error in panel regressions goes back to \citet{GH}; \citet{BBM} survey validation-study evidence on the magnitude of survey measurement error, \citet{Schennach2016} surveys the modern identification-based literature, and \citet{LiMa2024} review subsequent statistical developments across classical, Berkson, and combined-error models. \citet{MeyerMittag2021} show with linked administrative records that measurement error can dominate other components of total survey error. \citet{BK} provide the canonical external measurement-reliability estimates for earnings data (CPS matched to Social Security records), and \citet{BoundEtAl94} the PSID validation-study estimates, including the sharp drop in within reliability after within/difference transformations, the empirical fact that motivates our drift calibration in Section~\ref{sec-drift}. For expert-coded political indices, \citet{PMM} estimate an explicit measurement model whose posterior dispersion supplies a direct observation-level noise estimate; this is the basis of our lead application. Many landmark nonlinear corrections, including corrected-score, conditional-score, and SIMEX methods, were developed under Gaussian measurement error; \citet{WoolseyHuang2026} document the resulting non-Gaussian problem and construct higher-moment corrections. The linear attenuation and breakdown map here require no parametric measurement-error distribution: they use the noise second moment, together with the finite-moment and CLT conditions stated below. The weak-identification literature we borrow from is \citet{SS, SY, Moreira2003, OleaPflueger2013, ASS, LMMP}. \citet{CJN} and \citet{JV} develop the high-dimensional FE asymptotics in which $\rho > 0$ (with $\tau^2 = \infty$ in our parametrization); we work within that framework but with the additional dimension $\tau^2 < \infty$. \citet{AnatolyevNg2026} establish asymptotic normality and a leave-cluster-out, many-covariate-robust variance estimator under clustered errors, including unbalanced and growing clusters. Their result is the general estimation framework closest to Section~\ref{sec-cluster}; our ordinary-Arellano argument is a narrower specialization, while the measurement-error drift, the $\psi$ versus $\hpsi$ distinction, and the reported-$t$ breakdown map are diagnostic results not supplied by their estimator. \citet{Anatolyev2018} refines the heteroskedastic CJN correction, and \citet{KSS} and \citet{Jochmans2022} develop complementary leave-out methods. The two i.i.d. auxiliary results this paper itself needs --- a self-normalized score CLT and a Hessian concentration lemma under the $(\rho,\tau^2)$ drift --- are stated and proved in the Online Supplement, keeping the paper self-contained. To our knowledge, no formal Stock--Yogo-style critical values for FE saturation under measurement error have been derived.

The repeated-report design of \citet{AshenfelterKrueger1994} supplies a particularly transparent application: each identical twin reports both siblings' schooling, creating two measurements of the within-pair schooling difference. \citet{Rouse1999} reproduces the original estimates, extends the sample from 149 to 453 pairs, and shows why person-specific correlation across reports must be treated as a sensitivity rather than assumed away. Our use of these studies is inferential: their estimators correct a coefficient, whereas the diagnostic below quantifies the coverage failure of the uncorrected interval.

Sections~\ref{sec-setup}--\ref{sec-baseline} develop the setup and the baseline no-threshold result; Section~\ref{sec-eiv} introduces the EIV model and derives the non-centrality; Section~\ref{sec-cv} gives the critical-value formula, feasible diagnostic, and protocol; Sections~\ref{sec-sim}--\ref{sec-app} report the simulations and applications; Sections~\ref{sec-ext}--\ref{sec-conc} record extensions and conclude.

\section{Setup}\label{sec-setup}

Let $\{(Y_i, X_i)\}_{i=1}^n$ be a sample from a population on $(\Omega, \F, \Pp)$, with $Y_i \in \R$ and $X_i \in \R$ a scalar treatment. The structural model is
\begin{equation}\label{eq-model}
Y_i = X_i \beta_0 + u_i, \qquad \E[u_i \mid X_i, \G_\infty] = 0,
\end{equation}
where $\G_\infty = \sigma(\bigcup_{K \geq 1} \G_K)$ is the limit of an increasing sequence of sub-$\sigma$-algebras representing successive fixed-effect specifications. Let $D_K$ be the $n \times d_K$ matrix of FE dummies, $P_K$ the orthogonal projection onto the FE column space, and $M_K = I_n - P_K$ the FE annihilator; all partialling-out identities used below are instances of the Frisch--Waugh--Lovell theorem \citep{FrischWaugh,Lovell1963}, of which \citet{LovellFWL} gives a compact modern proof. Write $\tX_K = M_K X$. Throughout, the fitted dummies saturate the specification: $D_K$ spans the $\G_K$-measurable functions of the cell labels, so that $\E[X \mid \G_K] \in \mathrm{col}(D_K)$ and hence $M_K\,\E[X\mid\G_K]=0$ (and likewise any additive $\G_K$-measurable component of the outcome is annihilated). This is the assumption under which $X'M_KX$ equals the pure-deviation quadratic form $\xi'M_K\xi$, $\xi := X - \E[X\mid\G_K]$, used in Lemma~\ref{lem-clt} below.

The FE-residualized OLS estimator (computed in this section from the true $X$) is $\hbeta_K = (X' M_K Y) / (X' M_K X)$. The population residual treatment variance is
\[
Q_K = \E[(X - \E[X \mid \G_K])^2] = \E[\Var(X \mid \G_K)],
\]
with sample analogue $\hQ_K = n^{-1} X' M_K X$. The asymptotic regime is governed by
\begin{equation}\label{eq-drift}
\rho_n = \frac{d_{K_n}}{n} \to \rho \in [0, 1), \qquad \tau_n^2 = nQ_{K_n} \to \tau^2 \in (0, \infty].
\end{equation}
The case $\tau^2 = \infty$ recovers the many-covariates regime of \citet{CJN} and \citet{JV}; the case $\rho = 0$, $\tau^2 = \infty$ is textbook FE asymptotics. The non-centrality formula nests these limits by continuity (Remark~\ref{rem-nesting}); the theorems below are proved for the interior $0 < \tau^2 < \infty$, $\rho \in [0,\bar\rho]$ with $\bar\rho < 1$, and the boundaries $\tau^2 \in \{0,\infty\}$ and $\rho \to 1$ are separate limiting regimes at which the proofs' leverage and concentration constants are not uniform.

The joint CLT under this drift is the foundational result on which this paper builds. We state it as Lemma~\ref{lem-clt} and prove it in full in the Online Supplement, so that the paper is self-contained.

\begin{assumption}[Regularity]\label{ass-reg}\hfill
\begin{enumerate}[label=(\roman*)]
\item $\{(X_i, u_i, \G_K)\}$ is i.i.d.\ across $i$ for each $K$, with finite eighth moments of $X$ and $u$.
\item $\rho_n \to \rho \in [0, 1)$, $\tau_n^2 \to \tau^2 \in (0, \infty]$.
\item Leverage control: there is $\kappa > 0$ with $\max_i (P_{K_n})_{ii} \le 1 - \kappa$ with probability approaching one, and $\max_i \tX_{K_n, i}^2 / (nQ_{K_n}) = o_p(1)$. (A vanishing $\max_i (P_{K_n})_{ii}$ would force $\rho = 0$, since $\max_i (P_{K_n})_{ii} \ge \tr(P_{K_n})/n = \rho_n$; the many-fixed-effect regime $\rho > 0$ requires only that no single cell dominates, i.e.\ leverage bounded away from one. Only the treatment-leverage half is used in the score CLT below.)
\item Conditional variance: there exists $\omega^2 \in (0, \infty)$ with
\[
(nQ_{K_n})^{-1} \sum_i \tX_{K_n, i}^2 \sigma^2(X_i, \G_{K_n, i}) \to_p (1-\rho)\,\omega^2,
\]
so that $\omega^2$ is the design-free (fully-diffuse) conditional-variance limit and the observable score variance carries the same $(1-\rho)$ discount as the Hessian limit in Lemma~\ref{lem-clt}; under conditional homoskedasticity $\omega^2 = \sigma^2$.
\item Lindeberg condition: for every $\varepsilon > 0$,
\[
\resizebox{0.97\linewidth}{!}{$\displaystyle (nQ_{K_n})^{-1} \sum_i \E\big[\tX_{K_n, i}^2 u_i^2 \, \mathbf 1\{|\tX_{K_n, i} u_i| > \varepsilon \sqrt{nQ_{K_n}}\}\big] \to 0$}.
\]
\item Uniform conditional moment: for some $\epsilon > 0$,
$\sup_{i,n} \E[\,|u_i|^{2+\epsilon} \mid X_i,\G_{K_n}\,]
\le C_u < \infty$. (The conditioning includes the regressor because the score
proof conditions on $\F_n=\sigma(X,D_{K_n})$; a bound conditional only on
$\G_{K_n}$ would not in general survive further conditioning on $X_i$.
Together with a maximal-weight condition this yields the Lyapunov, hence
Lindeberg, condition for \emph{independent but not necessarily identically
distributed} weighted sums of the $u_i$ --- the form used for the contaminated
score in Theorem~\ref{thm-noncentral}; it is implied by the finite eighth moment
of (i) when the conditional law of $u_i$ does not degenerate across cells and
regressor values.)
\end{enumerate}
\end{assumption}

\begin{lemma}[Joint CLT under the $(\rho,\tau^2)$ drift]\label{lem-clt}
Set
\[
\bar\sigma^2_{M,n}:=\frac{1}{n-d_{K_n}}
\sum_i(M_{K_n})_{ii}\sigma^2(X_i,\G_{K_n,i}).
\]
Under Assumption~\ref{ass-reg} and the balance conditions (B)--(C) of the Online Supplement,
\[
\begin{pmatrix}
\dfrac{X' M_{K_n} u}{\sqrt{nQ_{K_n}}} \\[4pt]
\dfrac{X' M_{K_n} X}{nQ_{K_n}} \\[4pt]
\dfrac{u' M_{K_n} u}{n - d_{K_n}}-\bar\sigma^2_{M,n}
\end{pmatrix}
\Rightarrow
\begin{pmatrix} \mathcal Z \\ 1-\rho \\ 0 \end{pmatrix},
\qquad \mathcal Z \sim N(0, (1-\rho)\,\omega^2),
\]
with the second and third coordinates converging in probability. If
$\bar\sigma^2_{M,n}\to_p\sigma^2$---in particular, under conditional homoskedasticity---then
$u'M_{K_n}u/(n-d_{K_n})\to_p\sigma^2$. Under conditional homoskedasticity also $\omega^2=\sigma^2$, so the score variance is $(1-\rho)\sigma^2$, not $\sigma^2$.
\end{lemma}

\emph{Why the Hessian limit is $1-\rho$.} The second coordinate is worth a word, since it drives the critical-value formula. With $\xi_i := X_i - \E[X_i\mid\G_{K_n}]$ and $X'M_{K_n}X=\xi'M_{K_n}\xi$ under saturation, the conditional mean is $\sum_i(1-h_{ii})q_i$, where $q_i=\Var(\xi_i\mid\mathcal D_n)$ and $h_{ii}=(P_{K_n})_{ii}$. Under (B1) this equals $Q_{K_n}\tr(M_{K_n})=(1-\rho_n)nQ_{K_n}$ exactly; under (B2) leverage balance gives the same expression up to $1+o_p(1)$. Saturation alone does not imply the common trace discount. The exact trace identity does apply to homoskedastic errors and to the i.i.d.\ noise form $\nu'M_{K_n}\nu$, while Condition~(B) supplies it for the signal. The proof in the Online Supplement makes these distinctions explicit and conditions on $\F_n=\sigma(X,D_{K_n})$ for the score CLT.

\section{No baseline threshold}\label{sec-baseline}

A natural conjecture, building on the analogy with weak instruments, is that small $Q_K$ produces a Stock--Yogo problem requiring its own critical-value table. We first show this conjecture is incorrect in the baseline model.

Let $\widehat u = M_{K_n}(Y - X\hbeta_{K_n})$ be the regression residual after removing both the fixed effects and the treatment (equivalently, the residual from the full OLS of $Y$ on $[D_{K_n}, X]$; by the Frisch--Waugh--Lovell theorem \citep{FrischWaugh, Lovell1963, LovellFWL}, $\hbeta_{K_n}$ is the coefficient on $X$ there), and let $\widehat\sigma^2_{\mathrm{CJN}} = \widehat u'\widehat u / (n - d_{K_n} - 1)$ be the associated homoskedastic residual-scale estimator, with
\[
T_n^{\mathrm{CJN}} = \frac{\hbeta_{K_n} - \beta_0}{\widehat\sigma_{\mathrm{CJN}} / \sqrt{X' M_{K_n} X}}.
\]
The subscript records the many-covariate normalization used throughout; this scalar RSS correction is not the heteroskedastic CJN covariance estimator based on $(M\odot M)^{-1}$. \citet{Anatolyev2018} proposes an almost-unbiased refinement of that latter estimator, not a replacement for $\widehat\sigma^2_{\mathrm{CJN}}$ in the exact homoskedastic $t$ result below. Section~\ref{sec-ext} and \ref{sec:supp-anatolyev} report the appropriate coefficient-standard-error sensitivity.

\begin{theorem}[No $\tau^2$-driven size distortion in the baseline model]\label{thm-no-distortion}
Under Assumption~\ref{ass-reg}, suppose $u\mid(D_{K_n},X)\sim N(0,\sigma^2I_n)$, $X' M_{K_n} X > 0$, and $n - d_{K_n} \ge 2$ (so the residual degrees of freedom $n-d_{K_n}-1$ used below is at least one; Assumption~\ref{ass-reg}(ii) guarantees this for all large $n$, since $\rho_n\to\rho<1$). Then $T_n^{\mathrm{CJN}} \sim t_{n - d_{K_n} - 1}$ exactly conditional on $(D_{K_n},X)$. The distortion relative to $N(0,1)$ satisfies
\[
\Pp(|T_n^{\mathrm{CJN}}| > z_{1-\alpha/2}) - \alpha = \frac{(z_{1-\alpha/2}^3 + z_{1-\alpha/2}) \phi(z_{1-\alpha/2})}{2(n - d_{K_n} - 1)} + O(n^{-2}),
\]
which is $O(1/n)$ and does not depend on $\tau_n^2$.
\end{theorem}

\begin{proof}[Proof sketch]
Conditional on $(D_{K_n},X)$, set $q=M_{K_n}X/\sqrt{X'M_{K_n}X}$ and $R=M_{K_n}-qq'$. Then $q'q=1$, $R$ is an orthogonal projector of rank $n-d_{K_n}-1$, and $Rq=0$. Cochran's theorem therefore makes $q'u/\sigma\sim N(0,1)$ independent of $u'Ru/\sigma^2\sim\chi^2_{n-d_{K_n}-1}$. Their ratio is $T_n^{\mathrm{CJN}}$, so it is exactly Student $t$; the Fisher--Cornish expansion \citep[26.7.8]{AS72} gives the displayed $O(1/n)$ distortion, uniformly in $\rho\in[0,\bar\rho]$ and independently of $\tau_n^2$. The identity $R=M_{[K,X]}$ also shows why the residual variance must project out $X$. Full details are in the Online Supplement.
\end{proof}

\begin{remark}[Why the IV analogy breaks]\label{rem-no-IV}
TSLS is biased of order $1/\mu^2$ because substituting $\widehat X$ for $X$ multiplies the first-stage residual into the second-stage residual, producing a term with expectation proportional to $\E[vu]$. FE-residualized OLS contains no such substitution: $M_KX$ is an exact projection onto a known subspace, so $\E[X'M_Ku]=0$ in finite samples and the estimator is unbiased. Any bias source capable of generating a Stock--Yogo threshold must therefore come from outside the baseline model.
\end{remark}

\begin{remark}[Heteroskedasticity]\label{rem-het}
Theorem~\ref{thm-no-distortion} is stated under homoskedasticity, where the exact $t$ distribution is available. Under heteroskedasticity, leave-one-out (HC2-type) variance estimators restore asymptotic validity of the baseline $t$-test in many-covariate designs \citep{CJN, KSS, Jochmans2022}; the ``no $\tau^2$ threshold'' conclusion is unchanged, though the exact finite-sample distribution is lost.
\end{remark}

We now introduce a natural bias source: classical measurement error in the treatment.

\section{The non-centrality theorem}\label{sec-eiv}

\subsection{Model and drift}\label{sec-eiv-model}

Suppose the analyst observes
\begin{equation}\label{eq-eiv}
X_i^* = X_i + \nu_i,
\end{equation}
where $X_i$ is the true treatment (unobserved), and $\nu_i$ is measurement error with $\E[\nu_i \mid \G_{K_n}] = 0$, $\Var(\nu_i \mid \G_{K_n}) = \sigma_\nu^2$, and $\nu_i \perp (X_i, u_i)$ conditional on $\G_{K_n}$. We retain $Q_{K_n}$ for the residual variance of the true treatment. The following assumption records the additional moment and tail conditions on the noise used in the variance calculations of Lemma~\ref{lem-atten} and the leverage argument of Lemma~\ref{lem-lev-star}.

\begin{assumption}[Measurement-error regularity]\label{ass-nu}\hfill
\begin{enumerate}[label=(\roman*)]
\item \emph{(Independence.)} $\{\nu_i\}_{i=1}^n$ are independent conditional on $\G_{K_n}$, with $\E[\nu_i\mid\G_{K_n}]=0$, $\Var(\nu_i\mid\G_{K_n})=\sigma_{\nu,n}^2$, and $\nu\perp(X,u)\mid\G_{K_n}$.
\item \emph{(Scaled moment bound.)} For some fixed $r>1$, the standardized errors have a uniformly bounded $2r$-th moment: $\E[\,|\nu_i/\sigma_{\nu,n}|^{2r}\mid\G_{K_n}\,]\le C_\nu<\infty$ for all $i,n$ --- equivalently $\E[\,|\nu_i|^{2r}\mid\G_{K_n}\,]\le C_\nu\,\sigma_{\nu,n}^{2r}$, the natural requirement that the noise has no heavier tails than a rescaled fixed distribution (it is not a bare finite-moment bound: the $2r$-th moment must scale as the $2r$-th power of the standard deviation, which is what the leverage argument uses). The canonical case is a uniformly bounded standardized fourth moment ($r=2$), which is also what the quadratic-form variance bound of Lemma~\ref{lem-atten} uses; the fourth moments need not be identical across observations.
\end{enumerate}
\end{assumption}

Assumption~\ref{ass-nu} imposes no tail condition beyond a finite fourth moment. Two moment orders appear, and we keep them separate. The quadratic-form variance in Lemma~\ref{lem-atten} --- and hence the denominator concentration used throughout, including the leverage ratio of Lemma~\ref{lem-lev-star} and Theorem~\ref{thm-noncentral} --- uses $r=2$; the \emph{maximal-noise} bound of Lemma~\ref{lem-lev-star} (its first display) and Remark~\ref{rem-lev-moment} need only $r>1$. We therefore take $r=2$ as the standing requirement for the EIV theorems and flag the two places where $r>1$ alone suffices. This is the routine light-tail requirement of the measurement-error literature --- weaker than the sub-Gaussianity one might expect a maximal-leverage argument to need; sub-Gaussian tails only sharpen the rate (Remark~\ref{rem-lev-moment}) and are not required.

The analyst runs FE-OLS on $(Y, X^*)$, obtaining
\begin{equation}\label{eq-est-star}
\hbeta_K^* = \frac{X^{*\prime} M_K Y}{X^{*\prime} M_K X^*}.
\end{equation}

We adopt the \emph{local measurement-error drift}:
\begin{equation}\label{eq-local-drift}
\sigma_{\nu,n}^2 = \frac{c^2}{n}
\end{equation}
for a fixed scalar $c^2 \geq 0$.

\begin{lemma}[Attenuation under the drift]\label{lem-atten}
Under Assumption~\ref{ass-reg}, Conditions~(B)--(C) of the Online Supplement, Assumption~\ref{ass-nu} with $r = 2$ (uniformly bounded standardized fourth moments), and the drift~\eqref{eq-local-drift}:
\begin{enumerate}[label=(\alph*)]
\item (Weak information, $\tau^2 < \infty$.) The true and contaminated Hessians and the attenuation factor satisfy
\[
\begin{gathered}
X'M_{K_n}X \to_p (1-\rho)\tau^2, \qquad
X^{*\prime} M_{K_n} X^* \to_p \tau^{*2} := (1-\rho)(\tau^2 + c^2), \\[4pt]
\lambda := \frac{X'M_{K_n}X}{X^{*\prime}M_{K_n}X^*} \to_p \frac{\tau^2}{\tau^2+c^2} \in (0, 1].
\end{gathered}
\]
Under Condition~(B), the signal and measurement noise are discounted by the same asymptotic factor $(1-\rho)$ (exactly in conditional mean under (B1)), so $\lambda$, the within reliability, does not depend on $\rho$: saturation shrinks the design's total identifying variation but does not change the limiting relative share of signal to noise within it. The slope estimator is not consistent; its limiting law --- centered at $\beta_0\lambda$ with nondegenerate Gaussian error --- is recorded in Corollary~\ref{cor-slope}, once the score CLT of Theorem~\ref{thm-noncentral} is available.
\item (Strong information, $nQ_{K_n} \to \infty$.) The slope estimator is consistent for the (now asymptotically unattenuated) value, $\hbeta_K^* \to_p \beta_0$, since $\lambda \to 1$.
\end{enumerate}
\end{lemma}

\begin{proof}
Deferred to the Online Supplement. In the weak-information regime, Lemma~\ref{lem-clt} gives $X'M_{K_n}X\to_p(1-\rho)\tau^2$. Expanding $X^{*\prime}M_{K_n}X^*$, the cross term $X'M_{K_n}\nu$ is $o_p(1)$, while a bounded-kurtosis quadratic-form variance bound in the Online Supplement gives $\nu'M_{K_n}\nu\to_p c^2(1-\rho)$ --- exactly the same discount, applied to the noise. In the strong-information regime the same cross and noise terms are negligible relative to $nQ_{K_n}$; the estimator claim then follows by comparing the score with the denominator.
\end{proof}

The finite-$n$ attenuation factor is
\begin{equation}\label{eq-lambda}
\lambda_n := \frac{Q_{K_n}}{Q_{K_n} + \sigma_{\nu, n}^2} = \frac{\tau_n^2}{\tau_n^2 + c^2} \to \lambda,
\end{equation}
which is exactly the population conditional-within reliability of the observed regressor. Under Condition~(B) it is also the asymptotic share of sample post-projection variation in $X^*$ that is signal (and under (B1) the signal and noise conditional means have the same finite-sample trace factor). Thus no $\rho_n$ appears in this ratio; the common asymptotic factor survives only in the total scale $\tau_n^{*2}=(1-\rho_n)(\tau_n^2+n\sigma_{\nu,n}^2)$. The primitive theory is organized around this data property; the diagnostic translates it into the specification's breakdown reliability $\lambda^{\dagger}$.

\subsection{Interpreting the drift}\label{sec-drift}

The drift $\sigma_\nu^2 = c^2/n$ needs justifying, because it is doing real work: it is the unique rate at which the projected noise $c^2(1-\rho)$ and the residual signal $\tau^2$ remain of the same order, and hence the unique rate at which the limit experiment produces a non-degenerate threshold. Fixed $\sigma_\nu^2 > 0$ makes attenuation $O(1)$ and the size distortion total; $\sigma_\nu^2 = o(1/n)$ makes measurement error vanish from the limit entirely. It belongs to the small-measurement-error tradition: \citet{HahnHausmanKim2021} use small-$\sigma$ expansions to form approximate moments under conditional heteroskedasticity, while \citet{EvdokimovZeleneev2024} let an error scale $\tau\to0$ to characterize and identify nonparametric bias under weakly classical and nonclassical error. Our distinct calibration couples the error variance to $n$ so that standardized inferential bias remains $O(1)$. Three points establish that the knife-edge is the empirically relevant calibration rather than a mathematical convenience.

First, the drift is an \emph{approximation device with a finite-sample target}, in precisely the sense of \citet{SS}: weak-instrument asymptotics model the first stage as local to zero not because first stages literally shrink with $n$, but because the resulting limit distribution approximates the finite-sample distribution well whenever the concentration parameter is moderate. The same logic applies here. For a balanced treatment design satisfying (B1), the finite-$n$ calibration is
\begin{equation}\label{eq-eta-n}
\eta_n \;=\; -\frac{\beta_0\, \sigma_\nu^2\, \sqrt{n - d_{K_n}}}{\sigma\, \sqrt{\, Q_{K_n} + \sigma_\nu^2\,}}
\;=\; -\frac{\beta_0}{\sigma}\,(1 - \lambda_n)\,\sqrt{(n-d_{K_n})(Q_{K_n}+\sigma_\nu^2)},
\end{equation}
every ingredient of which is a fixed-$n$ quantity: no drifting sequence appears. The noise identity $\E[\nu'M_{K_n}\nu]=\sigma_\nu^2(n-d_{K_n})$ is exact; under (B1), $\E[X'M_{K_n}X\mid\mathcal D_n]=Q_{K_n}(n-d_{K_n})$ is exact as well. Under (B2), equation~\eqref{eq-eta-n} is the corresponding first-order calibration because the signal trace equality is only asymptotic. The theorems below say that $N(\eta_n, 1)$ approximates the law of the $t$-statistic whenever $\eta_n$ is moderate, and the balanced simulations in Section~\ref{sec-sim} (Design 4) quantify the quality of that approximation under a fixed finite-$n$ $\sigma_\nu^2$ (a single constant per cell), which is the direct check that the device is an approximation and not an artifact.

Second, the regime $\eta_n$ moderate is where applied work actually lives when FE saturate. Rewrite the noise-to-signal ratio as
\[
\frac{c^2}{\tau_n^2} = \frac{\sigma_\nu^2}{Q_{K_n}} = \frac{1 - \lambda_n}{\lambda_n},
\]
a population conditional-within ratio in which the fixed-effect dimension $\rho_n$ does not appear. Under Condition~(B), it also equals the limiting sample post-projection noise-to-signal ratio.
The empirical question is therefore: is the \emph{within reliability} $\lambda_n$ of real saturated-FE regressors close enough to one that measurement error is negligible, or close enough to zero that attenuation is total, or in between, where the threshold binds? The validation-study literature answers: in between, and increasingly so as saturation rises. Cross-sectional measurement reliability of survey earnings is high (on the order of $0.7$--$0.85$ in the CPS--Social Security match of \citealp{BK}); but within reliability after differencing or within transformations drops sharply, to roughly $0.5$--$0.65$ in the PSID validation study \citep{BoundEtAl94}, precisely because the transformation removes persistent signal while retaining transitory noise (the \citet{GH} amplification). Within reliabilities of $0.5$--$0.8$ correspond to $(1-\lambda)/\lambda$ between $0.25$ and $1$, exactly the range in which $\eta_n$ is moderate and the local approximation is designed to operate. The lead application of Section~\ref{sec-app-vdem} confirms this directly on real data: the realized $(1-\hlambda_n)/\hlambda_n$ is $0.11$ for the passing aggregate polyarchy specification and $0.83$ for the flagged legislative-constraints specification. Both sit inside the moderate window $[0.05, 10]$, so the local-drift regime is the empirically relevant one at both poles of the diagnostic.

Third, the drift formalizes a genuine comparative statics, not just an approximation: as researchers saturate ($\rho_n \uparrow$, $Q_{K_n} \downarrow$), the signal shrinks toward the noise floor while the noise itself is design-invariant. The sequence $\sigma_\nu^2 = c^2/n$, $nQ_K \to \tau^2$ is the stylized path of that empirical practice: it holds the noise-to-residual-signal ratio at the level the practitioner's own saturation choices have produced. In this sense the drift is not modeling the world shrinking; it is modeling the analyst saturating.

\begin{remark}[Nesting]\label{rem-nesting}
The limit experiment nests the familiar cases. $c = 0$ gives $\eta = 0$: no measurement error, no threshold, Theorem~\ref{thm-no-distortion} applies. $\tau^2 = \infty$ gives $\eta = 0$: with abundant residual signal, measurement error of order $1/n$ is asymptotically irrelevant and standard strong-identification inference is recovered, consistent with the textbook observation that classical EIV bias is proportional to the noise-to-signal ratio. $\rho = 0$ recovers the cross-sectional (no-FE) threshold exactly, since the common discount $\sqrt{1-\rho}$ is then $1$; more generally $\rho$ enters $\eta$ only through this overall scale factor, not through the within reliability $\lambda$, which is $\rho$-free by Lemma~\ref{lem-atten}. The interesting region is $\tau^2 < \infty$, $c > 0$: saturated designs with residual signal of the same order as residual noise.
\end{remark}

\subsection{Non-centrality}\label{sec-eiv-thm}

The asymptotic normality of the score in the theorem below rests on a leverage condition for the contaminated regressor $X^*$. We isolate it as a lemma; it is the one place where the moment bound of Assumption~\ref{ass-nu}(ii) is used, and it goes through under that bound alone, with no tail condition beyond a finite $2r$-th moment ($r>1$) needed.

\begin{lemma}[Leverage condition for the contaminated regressor]\label{lem-lev-star}
Suppose Assumption~\ref{ass-reg} and Conditions~(B)--(C) of the Online Supplement hold for the true regressor $X$ (so $\max_i \tX_{K_n,i}^2/(X'M_{K_n}X)\to_p 0$ and $X'M_{K_n}X\to_p(1-\rho)\tau^2$, Lemma~\ref{lem-clt}), Assumption~\ref{ass-nu} holds (with moment order $r>1$ for the maximal-noise bound below, and $r=2$ for the leverage-ratio conclusion, which invokes Lemma~\ref{lem-atten}), and the drift~\eqref{eq-local-drift} holds with $\tau^2<\infty$. Write $\widetilde X^{*}_{K_n,i}:=(M_{K_n}X^*)_i=\tX_{K_n,i}+\widetilde\nu_{K_n,i}$ with $\widetilde\nu_{K_n}:=M_{K_n}\nu$. Then
\[
\max_{1\le i\le n}\ \widetilde\nu_{K_n,i}^2 = O_p\!\big(c^2\, n^{1/r-1}\big)=o_p(1),
\qquad
\max_{1\le i\le n}\ \frac{\widetilde X^{*2}_{K_n,i}}{X^{*\prime}M_{K_n}X^*}\to_p 0 .
\]
\end{lemma}

\begin{proof}
Deferred to the Online Supplement. Applying Rosenthal's inequality directly to each row of the projection $M_{K_n}$ gives $\sup_i\E[|\widetilde\nu_{K_n,i}|^{2r}\mid\G_{K_n}]\lesssim\sigma_{\nu,n}^{2r}$, because every row has squared Euclidean norm $(M_{K_n})_{ii}\le1$. Markov's inequality and a union bound then give $\max_i\widetilde\nu_{K_n,i}^2=O_p(c^2n^{1/r-1})=o_p(1)$. The leverage ratio follows from $\widetilde X^{*2}_{K_n,i}\le2\tX_{K_n,i}^2+2\widetilde\nu_{K_n,i}^2$ and Lemma~\ref{lem-atten}.
\end{proof}

\begin{theorem}[Non-centrality under EIV with local drift]\label{thm-noncentral}
Let $T_n^{\mathrm{CJN}*}(\beta_0) = (\hbeta_K^* - \beta_0) / (\hsigma^*_{\mathrm{CJN}} / \sqrt{X^{*\prime} M_{K_n} X^*})$ be the CJN $t$-statistic computed from the contaminated regressor $X^*$, with $\hsigma^{*2}_{\mathrm{CJN}} = \widehat u^{*\prime}\widehat u^*/(n - d_{K_n} - 1)$ and $\widehat u^* = M_{K_n}(Y - X^*\hbeta^*_K)$. The degrees of freedom match the baseline estimator of Section~\ref{sec-baseline}: the residual is formed after projecting out both the fixed effects and $X^*$, so $d_{K_n}+1$ parameters have been fitted. Nothing asymptotic turns on the $-1$, but it is the convention the replication code uses and the one printed regression output reports. Under Assumption~\ref{ass-reg} and Conditions~(B)--(C) of the Online Supplement applied to the true treatment $X$, Assumption~\ref{ass-nu} on $\nu$ with $r = 2$, the local drift~\eqref{eq-local-drift}, conditional homoskedasticity with the error conditionally independent of the regressor, $u \perp X \mid \G_{K_n}$ (so that, with $u \perp \nu \mid \G_{K_n}$, conditioning on $X^*$ preserves the mean-zero, variance-$\sigma^2$ conditional law of each $u_i$; the $u_i$ remain independent but need not be identically distributed across cells, and the uniform conditional moment of Assumption~\ref{ass-reg}(vi) discharges the contaminated-score Lindeberg condition), $\tau^2 < \infty$, and $H_0: \beta = \beta_0$,
\[
T_n^{\mathrm{CJN}*}(\beta_0) \Rightarrow N(\eta, 1), \qquad
\eta = -\frac{\beta_0\, c^2\sqrt{1-\rho}}{\sigma \sqrt{\tau^2 + c^2}}.
\]
\end{theorem}

\begin{proof}[Proof sketch]
Write $M=M_{K_n}$ and $H_n^*=X^{*\prime}MX^*$. Under $H_0$ the statistic has the exact decomposition
\[
T_n^{\mathrm{CJN}*}(\beta_0)
=\frac{X^{*\prime}Mu}{\sigma\sqrt{H_n^*}}\frac{\sigma}{\hsigma^*_{\mathrm{CJN}}}
-\frac{\beta_0X^{*\prime}M\nu}{\hsigma^*_{\mathrm{CJN}}\sqrt{H_n^*}}.
\]
The first factor converges to $N(0,1)$ by the conditional Lyapunov CLT: Lemma~\ref{lem-lev-star} supplies its maximal-weight condition, and Assumptions~\ref{ass-reg}(vi) and~\ref{ass-nu}(i) preserve the required moments after conditioning on $X^*$. Lemma~\ref{lem-atten} gives $H_n^*\to_p\tau^{*2}=(1-\rho)(\tau^2+c^2)$ and $X^{*\prime}M\nu=X'M\nu+\nu'M\nu\to_p c^2(1-\rho)$. Finally, expanding the measured-equation residual sum of squares gives $\hsigma^{*2}_{\mathrm{CJN}}\to_p\sigma^2$. Slutsky therefore yields the displayed $N(\eta,1)$ limit. The Online Supplement verifies the conditional CLT and residual-scale limit in full.
\end{proof}

In words: once the treatment is measured with noise at this local scale, the usual $t$-test no longer centers at zero under the truth it is testing, but at $\eta$. Its size is governed by how much measurement noise there is relative to residual signal, and discounted only by the overall saturation factor $\sqrt{1-\rho}$, not by any preferential absorption of noise over signal. This single number is what the rest of the paper turns into a critical value, a breakdown reliability, and a certificate.

\begin{corollary}[Limiting law of the slope estimator]\label{cor-slope}
Under the conditions of Theorem~\ref{thm-noncentral} (in particular conditional homoskedasticity), with $\lambda=\tau^2/(\tau^2+c^2)$ and $\tau^{*2}=(1-\rho)(\tau^2+c^2)$ as in Lemma~\ref{lem-atten},
\[
\hbeta_K^* \;\Rightarrow\; \beta_0\lambda + N\!\big(0,\ \sigma^2/\tau^{*2}\big).
\]
The estimator is asymptotically centered at the attenuated value $\beta_0\lambda$, the \emph{classical} attenuation factor (no $\rho$-dependence); when $\tau^2 < \infty$ it is not consistent, and consistency ($\hbeta_K^*\to_p\beta_0\lambda$) is recovered only under strong information $nQ_{K_n}\to\infty$. The sampling variance $\sigma^2/\tau^{*2}$ does carry the $\rho$-dependence, since heavier saturation shrinks the total identifying variation $\tau^{*2}$ at fixed $(\tau^2,c^2)$.
\end{corollary}

\begin{proof}
Deferred to the Online Supplement: write $\hbeta_K^* = (\beta_0\, X^{*\prime}MX + X^{*\prime}Mu)/X^{*\prime}MX^*$ and apply Lemma~\ref{lem-atten} together with the score limit established in Theorem~\ref{thm-noncentral} and Slutsky.
\end{proof}

\begin{remark}[Sub-Gaussian tails only sharpen the rate]\label{rem-lev-moment}
Lemma~\ref{lem-lev-star} needs nothing beyond the finite $2r$-th moment of Assumption~\ref{ass-nu}(ii): any $r>1$ delivers $\max_i\widetilde\nu_{K_n,i}^2=O_p(c^2n^{1/r-1})=o_p(1)$, which is all the theorem uses. Stronger tails only sharpen the rate --- sub-Gaussian noise replaces $n^{1/r}$ by $\log n$ --- so we state the condition in the moment form standard in the measurement-error literature.
\end{remark}

\begin{remark}[Role of $\rho$ in the non-centrality]\label{rem-rho}
It is tempting to read $\eta = -\beta_0 c^2\sqrt{1-\rho}/(\sigma\sqrt{\tau^2+c^2})$ as saying that saturation preferentially absorbs measurement noise, so that adding fixed effects is a targeted remedy for attenuation. It is not. Under Condition~(B), the within reliability $\lambda=\tau^2/(\tau^2+c^2)$ does not depend on $\rho$ (Lemma~\ref{lem-atten}): the signal and noise Hessians have the same asymptotic degrees-of-freedom factor, with the signal statement relying on treatment balance. What $\rho$ does is shrink the total identifying variation $\tau^{*2}=(1-\rho)(\tau^2+c^2)$ uniformly, so the whole $\rho$-dependence of $\eta$ collapses to $(1-\rho)/\sqrt{1-\rho}=\sqrt{1-\rho}$: heavier saturation lowers $|\eta|$ only by making the estimator noisier overall, exactly as any further covariate would. At $|\beta_0|/\sigma=1$, $\tau^2=c^2=5$ (so $\lambda=0.5$ throughout), $|\eta|$ falls from $1.50$ at $\rho=0.1$ to $1.00$ at $\rho=0.6$, the ratio $\sqrt{0.4/0.9}=0.667$ being exactly the $\sqrt{1-\rho}$ channel.

This is consistent with \citet{GH}, whose amplification result concerns a different comparative static: they hold the absolute noise variance fixed and let a finer transformation shrink the persistent component of the signal while leaving transitory noise intact, so $\lambda$ falls. That mechanism operates on $X$'s covariance structure, which Assumption~\ref{ass-reg} does not engage; $\rho$ enters the within reliability in neither analysis.
\end{remark}

\subsection{Local power}\label{sec-power}

The same limit experiment delivers the power side of the diagnostic.

\begin{proposition}[Local power]\label{prop-power}
Under the conditions of Theorem~\ref{thm-noncentral}, but with the true coefficient set to the information-standardized alternative $\beta_n = \beta_0 + b/\sqrt{nQ_{K_n}}$ for fixed $b \in \R$ (when $\tau^2 < \infty$ this is a fixed, non-vanishing displacement $\beta_n - \beta_0 \to b/\tau$, the natural scaling since the total information $nQ_{K_n}$ stays finite; it reduces to a conventional vanishing local alternative only in the strong-information limit $\tau^2 = \infty$),
\[
T_n^{\mathrm{CJN}*}(\beta_0) \Rightarrow N\!\Big(\frac{b}{\sigma}\sqrt{(1-\rho)\lambda} \;+\; \eta,\ 1\Big),
\qquad \lambda = \frac{\tau^2}{\tau^2 + c^2}.
\]
Consequently, relative to the fully diffuse no-FE benchmark (slope $b/\sigma$), measurement error and FE saturation jointly attenuate the information-standardized power slope by the factor $\sqrt{(1-\rho)\lambda}$ and shift the power curve by $\eta$. Relative to the error-free case at the same $\rho$, whose slope is $(b/\sigma)\sqrt{1-\rho}$, measurement error alone contributes the factor $\sqrt\lambda$. Thus within reliability supplies the classical attenuation-of-power channel, while $\sqrt{1-\rho}$ is the overall-scale effect that also governs $\eta$ (Remark~\ref{rem-rho}).
\end{proposition}

\begin{proof}
Deferred to the Online Supplement. Relative to Theorem~\ref{thm-noncentral}, the alternative adds only
\[
\frac{b}{\sigma}\,
\frac{X^{*\prime}MX}{\sqrt{nQ_{K_n}}\sqrt{X^{*\prime}MX^*}}
\;\to_p\;\frac{b}{\sigma}\sqrt{1-\rho}\,
\frac{\tau}{\sqrt{\tau^2+c^2}}
=\frac{b}{\sigma}\sqrt{(1-\rho)\lambda}
\]
to the standardized numerator; all other terms have the limit established there.
\end{proof}

\begin{remark}[Reading the power result]\label{rem-power}
Proposition~\ref{prop-power} says the threshold is not only about size. Even a specification that passes the size threshold pays a power tax of $\sqrt{(1-\rho)\lambda}$; at within reliability $\lambda = 0.6$ and modest saturation $\rho=0.1$, the local slope falls to $\sqrt{0.9\times0.6}\approx74\%$ of its error-free value. Reporting $\hlambda_n$ and $\hrho_n$ alongside the size diagnostic therefore serves double duty --- $\hlambda_n$ summarizes the within-reliability tax, $\hrho_n$ the saturation tax, and the two multiply rather than one subsuming the other. In the applications audited here $\hrho_n$ is small ($0.03$ and $0.14$), so the $\sqrt{1-\rho}$ correction is minor in practice; it need not be for more heavily saturated designs (e.g.\ worker--firm networks), where reporting it separately matters more.
\end{remark}

\subsection{Cluster dependence}\label{sec-cluster}

Clustered inference is part of the diagnostic rather than an after-the-fact adjustment: because the diagnostic is a bias-to-standard-error ratio, changing the variance estimator can reverse its verdict. A rescaling imposed by assumption would establish neither the relevant limit law nor the validity of the reported standard error. \citet{AnatolyevNg2026} provide the general neighboring theory: asymptotic normality for a fixed number of OLS coefficients with many nuisance covariates and clustered errors, together with a consistent leave-cluster-out covariance estimator under unbalanced and potentially growing clusters. We establish the narrower auxiliary result required here, in a form parallel to theirs: a cluster-score CLT and consistency of the ordinary Arellano estimator under a projection-compatibility condition computable from the design. The distinct contribution used by the diagnostic is then the measurement-error drift and its reported-$t$ map, including the fact that population $\psi$ and reported $\hpsi$ need not agree.

Theorem~\ref{thm-noncentral} standardizes by the homoskedastic-i.i.d.\ variance, whereas applied panel practice clusters standard errors by unit \citep{Arellano1987, BDM}; see \citet{CameronMiller2015} and \citet{MNW2023} for current practice. The arguments below identify the cluster score, state the conditions under which its empirical analogue consistently estimates the relevant variance, and derive the resulting non-centrality.

\paragraph{Notation.} Partition $\{1,\dots,n\}$ into clusters $\mathcal C_1,\dots,\mathcal C_{G_n}$ with $n_g := |\mathcal C_g|$ and $\bar n_n := \max_g n_g$. Write $\F^*_n := \sigma(X,\nu,D_{K_n})$, so that conditional on $\F^*_n$ the residualized regressor $\widetilde X^*_{K_n}=M X^*$ is fixed. Define the cluster loading vectors and their energies
\[
\begin{gathered}
a^{(g)} \in \R^n, \quad a^{(g)}_i := \widetilde X^{*}_{K_n,i}\,\mathbf 1\{i \in \mathcal C_g\}, \\[3pt]
A_g := \|a^{(g)}\|^2, \qquad \sum_{g} A_g = X^{*\prime}MX^* = \tau^{*2}_n ,
\end{gathered}
\]
the cluster scores $\zeta_g := a^{(g)\prime} u$ with $S_n := \sum_g \zeta_g = X^{*\prime}Mu$, and
\[
\Sigma_g := \Var(u_{\mathcal C_g} \mid \F^*_n), \qquad
\Psi_n := \Var(S_n \mid \F^*_n) = \sum_{g} a^{(g)\prime}\Sigma_g\, a^{(g)} .
\]
The Arellano cluster meat at score scale and the cluster-robust statistic are
\[
\widehat V^{\mathrm{sc}}_{CR} := \sum_{g}\big(a^{(g)\prime}\widehat u^*\big)^2 ,
\qquad
T^{CR}_n(\beta_0) := \frac{\hbeta^*_K - \beta_0}{\sqrt{\widehat V_{CR}}}
= \frac{X^{*\prime}M(Y - X^*\beta_0)}{\sqrt{\widehat V^{\mathrm{sc}}_{CR}}},
\]
where $\widehat V_{CR} = \widehat V^{\mathrm{sc}}_{CR}/\tau^{*4}_n$ is the CRVE of the coefficient and $\hpsi := \tau^{*2}_n \widehat V_{CR}/\hsigma^{*2}_{\mathrm{CJN}}$ is the reported variance-inflation factor, i.e.\ the ratio of the squared cluster-robust to the squared i.i.d.\ standard error.

\begin{assumption}[Cluster regularity]\label{ass-cluster}\hfill
\begin{enumerate}[label=(\roman*)]
\item \emph{(Between-cluster independence.)} Conditional on $\F^*_n$, the subvectors $u_{\mathcal C_1},\dots,u_{\mathcal C_{G_n}}$ are independent with mean zero and arbitrary within-cluster dependence; $\sup_{i,n}\E[|u_i|^{2+\delta}\mid\F^*_n]\le C_u$ for some $\delta \in (0,2]$ and $\sup_{i,n}\E[u_i^4\mid\F^*_n]\le\kappa_u$.
\item \emph{(Many clusters, bounded cluster sizes, no dominant cluster.)} $G_n\to\infty$; the cluster sizes are bounded, $\bar n_n\le\bar n<\infty$ for all $n$; and $\max_g A_g/\tau^{*2}_n\to_p0$. (Given boundedness, the last condition is implied by the maximal-leverage limit in (v), since $\max_gA_g\le\bar n\max_i\widetilde X^{*2}_{K_n,i}$; it is stated separately because it is what the proofs use. Growing cluster sizes are possible under the strengthened conditions of Remark~\ref{rem-clustersize}.)
\item \emph{(Non-degeneracy, and the baseline for $\sigma^2$.)} $n^{-1}\sum_i\E[u_i^2\mid\F^*_n]\to_p\sigma^2\in(0,\infty)$, and $\Psi_n/(\sigma^2\tau^{*2}_n)\to_p\psi\in(0,\infty)$. The first limit is a normalization rather than a restriction, and it is needed: the second alone pins down only the product $\sigma^2\psi$, so without a declared baseline the assertion $\psi>1$ has no content. Fixing $\sigma^2$ as the average marginal error variance makes $\psi$ the ratio of the realized score variance to the score variance the same design would carry with independent, homoskedastic errors of that same average variance --- the reading used in Remark~\ref{rem-psi-sign}. Under heteroskedasticity that benchmark is not the independent-error score variance $\sum_i\widetilde X^{*2}_{K_n,i}\E[u_i^2\mid\F^*_n]$, so $\psi$ then mixes dependence with the covariance between the squared residualized regressor and the conditional variance; the theory is unaffected, since only the product $\psi\sigma^2$ enters, but the interpretation of $\psi$ as a pure dependence effect requires the homoskedastic reading. Under the conditional homoskedasticity of Theorem~\ref{thm-noncentral} it coincides with the $\sigma^2$ there.
\item \emph{(Projection compatibility.)} $\sum_{g}\big\|Ma^{(g)}-a^{(g)}\big\|^2 = o_p(\tau^{*2}_n)$.
\item \emph{(Design.)} The blocks $\{X_{\mathcal C_g}\}_{g\le G_n}$ are independent across clusters, with unrestricted dependence within a cluster; the measurement errors continue to satisfy Assumption~\ref{ass-nu}; and the design-side conclusions of Lemmas~\ref{lem-atten} and~\ref{lem-lev-star} hold:
\[
X'M\nu = o_p(1), \quad \nu'M\nu\to_p c^2(1-\rho), \quad X'MX\to_p(1-\rho)\tau^2, \quad \hlambda_n\to_p\lambda,
\]
\[
\max_i\widetilde\nu^2_{K_n,i}=o_p(1), \qquad \max_i\widetilde X^{*2}_{K_n,i}\big/\tau^{*2}_n\to_p0 ,
\]
the first three giving $X^{*\prime}MX^*\to_p\tau^{*2}=(1-\rho)(\tau^2+c^2)$. The two bias limits are listed separately rather than deduced from the Hessian limits, which would deliver only the combination $2X'M\nu+\nu'M\nu\to_p c^2(1-\rho)$, and it is $X'M\nu+\nu'M\nu$ that the non-centrality needs.
\end{enumerate}
\end{assumption}

Conditions (i)--(iii) are the standard many-clusters apparatus, with cluster sizes bounded. Boundedness is imposed rather than assumed away: the proofs of Lemma~\ref{lem-crve} use it at three separate points, and Remark~\ref{rem-clustersize} states exactly what must be strengthened to allow clusters to grow.

Condition (v) exists because Assumption~\ref{ass-reg}(i) cannot simply be carried over. That condition makes $\{(X_i,u_i)\}$ i.i.d.\ across $i$, which is inconsistent with the within-cluster dependence in $u$ that this section is about, and also excludes the serially dependent regressor that drives the empirically common case (Remark~\ref{rem-psi-sign}). What the cluster proofs actually use from the i.i.d.\ theory is not the sampling scheme but the six design-side limits listed in (v), all of which concern $(X,\nu,D_{K_n})$ alone and none of which involves $u$. Lemmas~\ref{lem-atten} and~\ref{lem-lev-star} deliver them under i.i.d.\ sampling. Analogous bounded-cluster primitive conditions can also deliver them with within-cluster dependence in $X$, but bounded cluster size alone is not asserted to be sufficient; the six limits are stated as hypotheses precisely to keep the cluster results free of a stronger sampling claim they do not need.

Condition (iv) is the one with no counterpart in the i.i.d.\ theory, and it is where the many-fixed-effect regime bites. It asks that the fixed-effect projection not distort the cluster loadings, and it fails silently if the fixed effects cut across clusters in a high-dimensional way. This is the design-specialized counterpart of the central trade-off in \citet{AnatolyevNg2026} between cluster-size growth and the accuracy of the auxiliary regression of the target regressors on nuisance covariates. Their conditions are more general; (N1)--(N2) below expose the same tension directly for nested and non-nested fixed-effect blocks. The next lemma gives checkable primitive conditions.

\begin{lemma}[When projection compatibility holds]\label{lem-nest}
\begin{enumerate}[label=(\alph*)]
\item \emph{(Nested fixed effects.)} If every fixed-effect cell is contained in a single cluster, then $Ma^{(g)} = a^{(g)}$ exactly for every $g$, and Assumption~\ref{ass-cluster}(iv) holds with the left-hand side equal to zero.
\item \emph{(Nested plus a low-dimensional non-nested block.)} Suppose $D_{K_n} = [\,D^{\mathrm{nest}}, D^{\mathrm{ne}}\,]$ with $D^{\mathrm{nest}}$ nested in clusters, and write $\Lambda_n := D^{\mathrm{ne}\prime}M_{\mathrm{nest}}D^{\mathrm{ne}}$. If $\rank(\Lambda_n)=0$, then $M=M_{\mathrm{nest}}$ and part~(a) applies. Otherwise,
\[
\sum_{g}\big\|Ma^{(g)}-a^{(g)}\big\|^2 \;\le\; \varpi_n \max_g A_g ,
\qquad
\varpi_n := \frac{\tr\big(D^{\mathrm{ne}\prime}D^{\mathrm{ne}}\big)}{\lambda^+_{\min}(\Lambda_n)} ,
\]
with $\lambda^+_{\min}$ the smallest non-zero eigenvalue. The displayed inequality is unconditional. The calibration that follows is not, and needs two balance conditions stated as such: (N1) a spectral condition $\lambda^+_{\min}(\Lambda_n)\asymp n/d^{\mathrm{ne}}_n$ with $d^{\mathrm{ne}}_n:=\rank(D^{\mathrm{ne}})$, which together with $\tr(D^{\mathrm{ne}\prime}D^{\mathrm{ne}})\asymp n$ gives $\varpi_n\asymp d^{\mathrm{ne}}_n$; and (N2) an energy condition $\max_gA_g = O_p(\tau^{*2}_n/G_n)$. Under (N1)--(N2), \emph{condition (iv) reduces to $d^{\mathrm{ne}}_n/G_n \to 0$}: the number of fixed effects that cut across clusters must be small relative to the number of clusters. Neither condition follows from cell counts alone. (N1) holds exactly in the balanced two-way panel --- with $G$ units and $T$ periods, $\Lambda_n = G(I_T-T^{-1}\mathbf 1\mathbf 1')$, so $\lambda^+_{\min}=G=n/T$ --- but equal cell sizes do not by themselves control the smallest non-zero eigenvalue of the residualized Gram matrix. (N2) concerns the residualized regressor rather than the design: equal cluster sizes with $\widetilde X^*$ energy concentrated in $\sqrt{G_n}$ clusters give $\max_gA_g/\tau^{*2}_n\asymp G_n^{-1/2}$, and the reduction then requires $d^{\mathrm{ne}}_n/\sqrt{G_n}\to0$ instead. Both are computable from the design and the residualized regressor before any inference, which is the point of stating them.
\end{enumerate}
\end{lemma}

\begin{proof}
Deferred to the Online Supplement.
\end{proof}

Part (b) is the practically binding condition and it is computable before any estimation. In the clustered lead application (Section~\ref{sec-app-vdem}) the country effects are nested in country clusters while the $59$ year effects are not, against $G_n=163$ clusters, so $d^{\mathrm{ne}}_n/G_n\approx0.36$. The direct projection diagnostic reported there is much smaller than this conservative rank ratio. The twin-pair application (Section~\ref{sec-app-twins}) instead samples independent pairs and therefore exercises the baseline i.i.d.\ theorem without invoking this cluster condition.

\begin{lemma}[Cluster score CLT]\label{lem-cluster-clt}
Under Assumption~\ref{ass-nu}, Assumption~\ref{ass-cluster}(i)--(iii) and (v), and the drift~\eqref{eq-local-drift},
\[
\frac{X^{*\prime}M u}{\sqrt{\Psi_n}} \;\Rightarrow\; N(0,1).
\]
\end{lemma}

\begin{lemma}[Consistency of the cluster-robust variance estimator]\label{lem-crve}
Under the conditions of Lemma~\ref{lem-cluster-clt} and Assumption~\ref{ass-cluster}(iv),
\[
\widehat V^{\mathrm{sc}}_{CR}\big/\Psi_n \;\to_p\; 1 .
\]
The statement is for the uncorrected meat. Applying the conventional small-sample factor $\tfrac{G_n}{G_n-1}\cdot\tfrac{n-1}{n-K_n}$, $K_n = d_{K_n}+1$, multiplies the estimator by a factor converging to $1/(1-\rho)$ rather than to one, and so over-corrects whenever $\rho$ is non-negligible. Under nesting there is no residual-shrinkage bias to correct --- part (a) of Lemma~\ref{lem-nest} makes $a^{(g)\prime}\widehat u^*$ an exact linear functional of the errors with no projection loss --- and the $\tfrac{n-1}{n-K_n}$ factor should be omitted.
\end{lemma}

The scope of Lemma~\ref{lem-crve} is deliberately narrow. When projection compatibility is doubtful, the leave-cluster-out covariance estimator of \citet{AnatolyevNg2026}, if its own conditions and block-invertibility checks hold, is the better variance estimator to feed into the reported $t$. This substitution changes the numerical breakdown reliability through the $t$-statistic but not its algebraic form (Corollary~\ref{cor-cluster-feasible}(c)). A separate warning applies to the heteroskedastic CJN estimator: its Hadamard correction requires invertibility of $M\odot M$, for which maximum nuisance leverage below $1/2$ is a familiar sufficient condition; \citet{AnatolyevNg2026} emphasize, citing \citet{Jochmans2022}, that numerical non-existence can occur even at smaller leverage. A breakdown calculation cannot repair a variance estimator that does not exist.

\begin{remark}[Growing cluster sizes]\label{rem-clustersize}
Assumption~\ref{ass-cluster}(ii) imposes bounded cluster sizes, and the proofs use boundedness in three identifiable bounds in Lemma~\ref{lem-crve}: the conditional-covariance operator norm $\|\Var(u\mid\F^*_n)\|_{\mathrm{op}}\le\bar n_n\kappa_u^{1/2}$; the measurement-error remainder $\max_g\sum_{i\in\mathcal C_g}\widetilde\nu^2_{K_n,i}\le\bar n_n\max_i\widetilde\nu^2_{K_n,i}$; and the fourth-moment bound $\E[\zeta_g^4\mid\F^*_n]\le\kappa_u\bar n_n^2A_g^2$. Clusters may be allowed to grow, $\bar n_n\to\infty$, provided each is restored explicitly:
\begin{enumerate}[label=(\roman*$'$)]
\item \emph{(CLT.)} $\bar n_n^{(2+\delta)/\delta}\,\max_gA_g/\tau^{*2}_n\to_p0$, which is what the Lyapunov bound of Lemma~\ref{lem-cluster-clt} actually requires;
\item \emph{(Score-square concentration.)} $\bar n_n^{2}\,\max_gA_g/\tau^{*2}_n\to_p0$, implied by (i$'$) when $\delta\le2$;
\item \emph{(Projection remainder.)} $\bar n_n\sum_g\|Ma^{(g)}-a^{(g)}\|^2 = o_p(\tau^{*2}_n)$, a strengthening of Assumption~\ref{ass-cluster}(iv) by the factor $\bar n_n$;
\item \emph{(Measurement-error remainder.)} $\bar n_n\max_i\widetilde\nu^2_{K_n,i}=o_p(1)$, which by Lemma~\ref{lem-lev-star} holds whenever $\bar n_n = o(n^{1-1/r})$.
\end{enumerate}
We state the bounded-size version in the text because it is the empirically relevant one for the clustered lead application ($\bar n_n\le59$ against $G_n=163$) and because the projection-compatibility condition (iv) alone does not absorb the extra $\bar n_n$ in (iii$'$), so advertising ``bounded or slowly growing'' without (i$'$)--(iv$'$) would leave a gap between the statement and the proof.
\end{remark}

\begin{theorem}[Cluster-robust non-centrality]\label{thm-cluster}
Under Assumption~\ref{ass-nu} with $r=2$, Assumption~\ref{ass-cluster}, the local drift~\eqref{eq-local-drift}, $0<\tau^2<\infty$, and $H_0:\beta=\beta_0$,
\[
T^{CR}_n(\beta_0) \;\Rightarrow\; N\big(\eta_{CR},\,1\big),
\qquad
\eta_{CR} \;=\; \frac{\eta}{\sqrt\psi} \;=\; -\frac{\beta_0\,c^2\sqrt{1-\rho}}{\sigma\sqrt{\psi}\,\sqrt{\tau^2+c^2}} ,
\]
with $\eta$ as in Theorem~\ref{thm-noncentral}. Cluster dependence therefore enters the diagnostic through the denominator alone: the attenuation bias in the numerator is $\E[\nu'M_{K_n}\nu]\beta_0$, a functional of the measurement error only, and $\nu\perp u$ by Assumption~\ref{ass-nu}(i), so no dependence structure in $u$ can touch it. In words: clustering cannot change the bias itself, only how alarming that bias looks relative to a bigger or smaller standard error, and $\sqrt\psi$ is exactly that rescaling.
\end{theorem}

Three objects must now be kept apart, because they have different statuses and the distinction is exactly where a careless argument goes wrong. The \emph{oracle} non-centrality is evaluated at the true $\beta_0$; the \emph{point plug-in} diagnostic is evaluated at the corrected pilot $\hbeta_0^{\mathrm{corr}}$, which is not consistent in the weak-information regime (Corollary~\ref{cor-slope}, Remark~\ref{rem-corr-noise}); and the \emph{formal certificate} adds sampling uncertainty to the reported $t$ and, when necessary, uses a lower confidence bound for the within reliability. We take them in that order.

\begin{corollary}[Oracle rescaling, and the algebraic cancellation]\label{cor-cluster-feasible}
Define the \emph{cluster-robust scale}
\[
\widehat s^2_{CR} \;:=\; \hsigma^{*2}_{\mathrm{CJN}}\,\hpsi \;=\; \widehat V^{\mathrm{sc}}_{CR}\big/\tau^{*2}_n ,
\]
the second equality being an algebraic identity, not an approximation: $\hsigma^{*2}_{\mathrm{CJN}}$ cancels exactly. Under the conditions of Theorem~\ref{thm-cluster}:
\begin{enumerate}[label=(\alph*)]
\item $\widehat s^2_{CR}\to_p\psi\sigma^2$. This holds even though $\hsigma^{*2}_{\mathrm{CJN}}$ is not consistent for $\sigma^2$ under cluster dependence, and $\hpsi$ is not consistent for $\psi$: if $\hsigma^{*2}_{\mathrm{CJN}}\to_p\varsigma^2\in(0,\infty)$ --- an extra requirement, since nothing in Assumption~\ref{ass-cluster} forces this estimator to have a limit --- then $\hpsi\to_p\psi\sigma^2/\varsigma^2$, and the two errors are reciprocal. Part~(a) itself does not need $\varsigma^2$ to exist: $\widehat s^2_{CR}=\widehat V^{\mathrm{sc}}_{CR}/\tau^{*2}_n$ never mentions $\hsigma^{*2}_{\mathrm{CJN}}$.
\item \emph{(Oracle.)} At the true $\beta_0$, the feasible non-centrality formed with $\widehat s_{CR}$ in place of $\hsigma^*_{\mathrm{CJN}}$ is consistent for the cluster-robust target:
\[
\frac{|\beta_0|}{\widehat s_{CR}}\,(1-\hlambda_n)\sqrt{\tau^{*2}_n} \;\to_p\; \frac{|\beta_0|}{\sigma\sqrt\psi}(1-\lambda)\sqrt{\tau^{*2}} \;=\; \frac{|\eta|}{\sqrt\psi} \;=\; |\eta_{CR}| .
\]
More generally, for any deterministic $b$ substituted for $|\beta_0|$, the same statement holds with $|\eta|$ replaced by its value at $b$.
\item \emph{(Breakdown, exactly.)} $t^*_n/\sqrt{\hpsi}$ is identically the reported cluster-robust $t$-statistic $t^{CR}_n := |\hbeta^*_K|/\sqrt{\widehat V_{CR}}$, so
\[
\lambda^{\dagger}_{CR} \;=\; \frac{t^{CR}_n}{\,t^{CR}_n+\eta^{\dagger}(\alpha,\delta)\,}
\]
is Definition~\ref{def-breakdown} evaluated at the cluster-robust $t$. This is an algebraic identity requiring no limit theory, and it is the operational content of the rescaling: \emph{the cluster-robust diagnostic is the i.i.d.\ diagnostic run on the reported cluster-robust $t$-statistic.}
\end{enumerate}
\end{corollary}

\begin{remark}[The point plug-in is descriptive, not consistent]\label{rem-plugin}
Substituting the corrected pilot $\hbeta_0^{\mathrm{corr}}=\hbeta^*_K/\hlambda_n$ for $|\beta_0|$ in Corollary~\ref{cor-cluster-feasible}(b) does not yield a consistent estimator of $|\eta_{CR}|$, and no argument in this paper claims otherwise. In the weak-information regime $\tau^2<\infty$ the pilot is asymptotically centred at $\beta_0$ but nondegenerate. Its limiting variance is not the i.i.d.\ one of Corollary~\ref{cor-slope} and Proposition~\ref{prop-pilot}(ii): under cluster dependence the score variance is $\Psi_n\to_p\psi\sigma^2\tau^{*2}$ rather than $\sigma^2\tau^{*2}$, so the same Slutsky argument applied to Lemma~\ref{lem-cluster-clt} in place of the i.i.d.\ score CLT gives
\[
\hbeta_0^{\mathrm{corr}}\;\Rightarrow\;B \;:=\; \beta_0+N\big(0,\ \psi\,\sigma^2/(\lambda^2\tau^{*2})\big) ,
\]
the extra factor $\psi$ being exactly what makes $\widehat{\mathrm{se}}(\hbeta_0^{\mathrm{corr}})=\widehat s_{CR}/(\hlambda_n\sqrt{\tau^{*2}_n})$ the right standard error for it in Proposition~\ref{prop-certificate}. Consequently, by the continuous mapping theorem and part (a), when $\beta_0\ne0$,
\[
|\widehat\eta_{CR}| \;:=\; \frac{|\hbeta_0^{\mathrm{corr}}|}{\widehat s_{CR}}(1-\hlambda_n)\sqrt{\tau^{*2}_n}
\;\Rightarrow\; \frac{|B|}{|\beta_0|}\,|\eta_{CR}| ,
\]
a nondegenerate random multiple of the target, with median close to it but no concentration. (At $\beta_0=0$ the ratio is undefined and the statement is instead that $|\eta_{CR}|=0$ while $|\widehat\eta_{CR}|\Rightarrow|B|(1-\lambda)\sqrt{\tau^{*2}}/(\sigma\sqrt\psi)$ with $B$ centred at zero: the plug-in remains nondegenerate even though the target vanishes, which is the same point.) This is precisely why the plug-in verdict is labelled a \emph{point pass} and not a certificate: it is a descriptive statistic, and its sampling variability is a first-order feature of the weak-information regime rather than an approximation error that vanishes. Consistency is recovered in two boundary regimes, and only there: under strong information $nQ_{K_n}\to\infty$, where the Gaussian term disappears and $B$ itself collapses to $\beta_0$; and at $c^2=0$ (no measurement error), where $\lambda=1$ so $B$ stays nondegenerate but the multiplier $(1-\hlambda_n)\to_p0$ annihilates it, giving $|\widehat\eta_{CR}|\to_p0=|\eta_{CR}|$ trivially (see the proof in the Online Supplement). Away from both boundaries $|\widehat\eta_{CR}|$ is not consistent for $|\eta_{CR}|$, and a size-controlled statement requires Proposition~\ref{prop-certificate}.
\end{remark}

\begin{proposition}[Certified breakdown reliability]\label{prop-certificate}
Fix $0<\alpha<1$, a size tolerance $0<\delta<1-\alpha$, and a coefficient-uncertainty level $\gamma_\beta\in(0,\tfrac12]$. Let
\[
t^{CR}_n:=\frac{|\hbeta^*_K|\sqrt{\tau^{*2}_n}}{\widehat s_{CR}}
\]
be the absolute reported cluster-robust $t$-statistic, and let $\ell_n\in(0,1]$ be the within reliability supplied to the certification rule. Define
\begin{equation}\label{eq-certified-breakdown}
\widehat\eta^{\,U}_n(\ell_n)
:=\big(t^{CR}_n+z_{1-\gamma_\beta}\big)\frac{1-\ell_n}{\ell_n},
\qquad
\lambda^{\dagger}_{\gamma_\beta,n}
:=\frac{t^{CR}_n+z_{1-\gamma_\beta}}
{t^{CR}_n+z_{1-\gamma_\beta}+\eta^{\dagger}(\alpha,\delta)}.
\end{equation}
Then $\widehat\eta^{\,U}_n(\ell_n)\le\eta^{\dagger}$ if and only if
$\ell_n\ge\lambda^{\dagger}_{\gamma_\beta,n}$. If $\ell_n\le0$, the rule does not certify.
Under the conditions of Theorem~\ref{thm-cluster}:
\begin{enumerate}[label=(\alph*)]
\item With the consistent pilot $\ell_n=\hlambda_n\to_p\lambda$, $\widehat\eta^{\,U}_n(\hlambda_n)$ is an asymptotically valid $(1-\gamma_\beta)$ upper confidence bound for $|\eta_{CR}|$. Consequently,
\[
\limsup_n\Pp\Big(\hlambda_n\ge\lambda^{\dagger}_{\gamma_\beta,n},\ 
\text{yet true asymptotic size}>\alpha+\delta\Big)\le\gamma_\beta .
\]
\item More generally, suppose $\ell_n$ is an asymptotic lower confidence bound for the within reliability with error $\gamma_\lambda$:
$\liminf_n\Pp(\ell_n\le\lambda)\ge1-\gamma_\lambda$. Then the same certification rule has false-certification probability at most $\gamma_\beta+\gamma_\lambda$. No independence between the coefficient and reliability bounds is required.
\end{enumerate}
The i.i.d.\ version replaces $t^{CR}_n$ by the reported i.i.d.\ $t$-statistic. When reliability is treated as known, set $\ell_n=\lambda$ and $\gamma_\lambda=0$.
\end{proposition}

\begin{remark}[The tolerance and the two uncertainty budgets]\label{rem-gamma-delta}
The size tolerance $\delta$ bounds the object being certified; $\gamma_\beta$ controls coefficient uncertainty and $\gamma_\lambda$ the coverage error of a supplied lower-reliability bound. They should be reported as $(\alpha,\delta,\gamma_\beta,\gamma_\lambda)$, with total false-certification probability bounded by $\gamma_\beta+\gamma_\lambda$. If reliability is known or treated as a consistent point pilot, $\gamma_\lambda=0$ and the original triple $(\alpha,\delta,\gamma_\beta)$ remains. A noisy finite-sample reliability estimate should not be inserted as though known: use its lower confidence bound in~\eqref{eq-certified-breakdown}. Design~6 shows why this distinction matters.
\end{remark}

\begin{proof}
Lemmas~\ref{lem-cluster-clt}--\ref{lem-crve}, Theorem~\ref{thm-cluster}, Corollary~\ref{cor-cluster-feasible} and Proposition~\ref{prop-certificate} are proved in the Online Supplement.
\end{proof}

\begin{remark}[The direction of the cluster adjustment]\label{rem-psi-sign}
It is tempting to reason that serial dependence in $u$ inflates standard errors, hence $\psi>1$, hence clustering mechanically deflates the measured distortion and the i.i.d.\ verdict is conservative. That reasoning is incomplete, and its conclusion can fail. Because $\Psi_n = \sum_g a^{(g)\prime}\Sigma_g a^{(g)}$ weights the within-cluster covariance by the residualized regressor, and because $\widetilde X^{*}_{K_n}$ is orthogonal to every fixed-effect dummy, an equicorrelated component of $u$ at a level that is itself a fixed effect is annihilated exactly, as in Proposition~\ref{prop-equicorr}. With a serially dependent error but a within-cluster serially independent regressor, the surviving terms are predominantly negative and $\psi<1$: clustering then makes the diagnostic more alarming, not less. What drives $\psi$ above one is persistence in the regressor and the error together, over a panel long enough for that persistence to accumulate. Simulations on a unit-and-time FE design confirm this: at $G=200$, an i.i.d.\ treatment with AR(1) errors ($\varrho_u=0.6$, $T=10$) gives $\psi\approx0.75$, while $\varrho_x=0.9,\varrho_u=0.9$ at $T=20$ gives $\psi\approx1.27$ and $\varrho_x=0.98,\varrho_u=0.95$ at $T=40$ gives $\psi\approx1.72$ (Section~\ref{sec-design5}(vi); the panel length rises with the persistence across the three, and $\psi$ depends on both).

Two cautions follow. First, the empirically common case --- a persistent regressor such as a democracy index or a log wage, together with persistent errors --- is the $\psi>1$ case, but that is an empirical fact about those designs and not a theorem. Second, and more easily missed, $\psi$ and its reported estimate $\hpsi$ are different objects whose comparisons with one need not agree: by Corollary~\ref{cor-cluster-feasible}(a), $\hpsi\to_p\psi\sigma^2/\varsigma^2$, and $\varsigma^2<\sigma^2$ whenever the fixed-effect projection removes positively correlated error. The gap is not small. In the nested configuration of Section~\ref{sec-design5} the population $\psi$ is $0.64$ --- clustering genuinely amplifies the distortion there --- while the reported $\hpsi$ sits at $0.99$. So $\psi<1$ does not show up as $\hpsi<1$, and one may not read the population statement off the reported number or the reverse. What settles the applied verdict is $\hpsi$, not $\psi$: by Corollary~\ref{cor-cluster-feasible}(c) the cluster-robust diagnostic is an exact function of $\hpsi$ in the sample, whereas $\psi$ enters only the population target. The direction of the adjustment must therefore be read off $\hpsi$, and this paragraph's mechanism read as an account of $\psi$.
\end{remark}

\begin{remark}[What the cluster results do and do not require]\label{rem-cluster-scope}
Theorem~\ref{thm-cluster} requires many clusters ($G_n\to\infty$) with no dominant cluster, bounded cluster sizes, projection compatibility, and the design-side limits of Assumption~\ref{ass-cluster}(v). It does not require the i.i.d.\ sampling of Assumption~\ref{ass-reg}(i), which would be incompatible both with clustered errors and with the persistent regressors of Remark~\ref{rem-psi-sign}. Cluster sizes may be allowed to grow only under the four strengthened conditions of Remark~\ref{rem-clustersize}, which we state explicitly rather than gesture at, because Assumption~\ref{ass-cluster}(iv) alone does not absorb the extra $\bar n_n$ factors the proofs require. It does not require the within-cluster dependence to have any particular form, nor cluster sizes to be equal, nor the errors to be homoskedastic within clusters. It does not cover the few-clusters regime, where $G_n$ is fixed and the CRVE is not consistent; there the appropriate tools are those of \citet{IbragimovMuller2010} or \citet{CRS2017}, and the diagnostic would have to be re-derived against a non-normal limit. It also does not cover multiway clustering.
\end{remark}

\section{Critical values}\label{sec-cv}

\subsection{Threshold and feasible form}

We invert the leading (quadratic) size distortion to obtain a Stock--Yogo-style threshold.

\begin{corollary}[Critical-value formula]\label{cor-cv}
Under the conditions of Theorem~\ref{thm-noncentral}, the asymptotic size of the nominal-$\alpha$ test $|T_n^{\mathrm{CJN}*}| > z_{1-\alpha/2}$ is even in $\eta$ with a vanishing first derivative, so its leading distortion is quadratic:
\[
\Pp(|T_n^{\mathrm{CJN}*}| > z_{1-\alpha/2}) = \alpha + z_{1-\alpha/2}\,\phi(z_{1-\alpha/2})\,\eta^2 + O(\eta^4)
\]
as $|\eta| \to 0$. The critical value bounding the leading distortion below $\delta$ is, writing $z := z_{1-\alpha/2}$,
\begin{equation}\label{eq-cv}
\tau^2_{\mathrm{crit}}(\rho, c^2, \beta_0, \sigma; \alpha, \delta)
= \frac{\beta_0^2\, c^4\,(1-\rho)\, z\,\phi(z)}{\sigma^2\delta} - c^2,
\end{equation}
positive whenever $\beta_0^2\, c^2 (1-\rho)\, z\,\phi(z) > \sigma^2\delta$. Two features of the formula trace directly to the Hessian limit of Lemma~\ref{lem-clt} and are worth flagging because they are easy to get wrong: the leading term is linear in $(1-\rho)$, not quadratic, and the subtracted term is $c^2$, not $c^2(1-\rho)$ (Remark~\ref{rem-rho}).
\end{corollary}

\begin{proof}[Proof sketch]
For $T\sim N(\eta,1)$ the two-sided rejection probability $s(\eta)=\Phi(-z-\eta)+1-\Phi(z-\eta)$ is even with $s'(0)=0$ and $s''(0)=2z\phi(z)$, so $s(\eta)=\alpha+z\phi(z)\eta^2+O(\eta^4)$; substituting $\eta^2=\beta_0^2 c^4(1-\rho)/(\sigma^2(\tau^2+c^2))$ (Theorem~\ref{thm-noncentral}) and setting the leading term equal to $\delta$ rearranges to~\eqref{eq-cv}. Full derivation in the Online Supplement.
\end{proof}

\begin{remark}[Quadratic versus exact-inversion threshold]\label{rem-exact-cv}
Formula~\eqref{eq-cv} inverts the leading quadratic distortion $z\phi(z)\eta^2$, equivalently the threshold $|\eta| \le \eta^{\mathrm{quad}} := \sqrt{\delta/(z\phi(z))}$. For $0<\alpha<1$ and $0<\delta<1-\alpha$, the \emph{exact-inversion} threshold replaces $\eta^{\mathrm{quad}}$ by $\eta^{\dagger}(\alpha,\delta)$, the unique positive root of $\Phi(-z-\eta)+1-\Phi(z-\eta)=\alpha+\delta$, giving
\begin{equation}\label{eq-cv-exact}
\tau^{2,\mathrm{exact}}_{\mathrm{crit}} = \left(\frac{|\beta_0| c^2\sqrt{1-\rho}}{\sigma\,\eta^{\dagger}(\alpha,\delta)}\right)^2 - c^2.
\end{equation}
For either formula, the operational threshold is the positive part of the displayed algebraic boundary: if the right-hand side is negative, every admissible $\tau^2\ge0$ passes that inversion.
At $\alpha=\delta=0.05$, $\eta^{\mathrm{quad}}=0.661$ versus $\eta^{\dagger}=0.652$: the quadratic closed form is within $1.4\%$ of the exact root in $\eta$. Because $\eta^{\mathrm{quad}}>\eta^{\dagger}$ it is very slightly anti-conservative (at its $\delta=0.05$ boundary the true distortion is $0.0513$, just above the nominal $\delta$), so it is an accurate approximation rather than a strict upper bound. This is a different and far milder failure than the discarded linear surrogate $\delta/(2\phi)=0.428$, which over-demanded $\tau^2$ by a factor $(\eta^{\dagger}/0.428)^2\approx 2.3$ and, worse, \emph{fails to control size} at larger $\delta$ (at $\delta=0.15$ its boundary $\eta=1.28$ carries a true distortion of $0.20$). The quadratic form~\eqref{eq-cv} degrades further only in the far tail $\eta \gtrsim 1$ ($\lambda \lesssim 0.5$), where the exact inversion~\eqref{eq-cv-exact} should be used; Design~2 of Section~\ref{sec-sim} confirms that the two nearly coincide across $(n,\rho,c^2)$ at the moderate $\delta$ the paper targets. We recommend~\eqref{eq-cv-exact} as headline, with~\eqref{eq-cv} as its closed-form companion.
\end{remark}

The threshold has an equivalent form that is more useful in practice, because it removes the drift-scale quantities $c^2$ and $\tau^2$ in favor of directly observable objects.

\begin{corollary}[Feasible non-centrality]\label{cor-feasible}
Under the hypotheses of Lemma~\ref{lem-atten}, let $\tau^{*2}_n := X^{*\prime} M_{K_n} X^*$ denote the observed (random) residual variation of the contaminated regressor --- distinct from its deterministic limit $\tau^{*2} = (1-\rho)(\tau^2 + c^2)$, to which it converges, $\tau^{*2}_n \to_p \tau^{*2}$ --- and $\lambda_n$ the within reliability of~\eqref{eq-lambda}. \emph{This corollary and everything built on it (Definition~\ref{def-breakdown}, the protocol of Section~\ref{sec-protocol}, and the empirical applications of Section~\ref{sec-app}) never requires the practitioner to decompose the diagnostic into $(\rho,\tau^2,c^2)$: $\tau^{*2}_n$ is observed from the regression, while $\lambda_n$ is supplied by the external within-reliability or noise pilot. The feasible diagnostic is therefore insulated from the primitive decomposition once that pilot is specified.} The population non-centrality obeys the exact identity $|\eta| = (|\beta_0|/\sigma)(1-\lambda)\sqrt{\tau^{*2}}$, and the plug-in
\[
|\heta_n| := \frac{|\beta_0|}{\sigma}\,(1 - \lambda_n)\, \sqrt{\tau^{*2}_n} \quad\text{satisfies}\quad |\heta_n| = |\eta|\,(1 + o_p(1)),
\]
so the size condition ``leading distortion $\le \delta$'' is asymptotically equivalent to
\begin{equation}\label{eq-feasible}
\frac{|\beta_0|}{\sigma}\,(1 - \lambda_n)\, \sqrt{\tau^{*2}_n} \;\le\; \sqrt{\frac{\delta}{z_{1-\alpha/2}\,\phi(z_{1-\alpha/2})}}\;=\;\eta^{\mathrm{quad}},
\end{equation}
its exact-normal-inversion analogue being asymptotically equivalent to $(|\beta_0|/\sigma)(1-\lambda_n)\sqrt{\tau^{*2}_n} \le \eta^{\dagger}(\alpha,\delta)$ (both feasible forms carry a $1+o_p(1)$ from replacing $\tau^{*2}$ by the random $\tau^{*2}_n$).
\end{corollary}

\begin{proof}[Proof sketch]
The identity $|\eta| = (|\beta_0|/\sigma)(1-\lambda)\sqrt{\tau^{*2}}$ is algebraic: with $\lambda=\tau^2/(\tau^2+c^2)$ and $\tau^{*2}=(1-\rho)(\tau^2+c^2)$, $(1-\lambda)\sqrt{\tau^{*2}} = \frac{c^2}{\tau^2+c^2}\sqrt{(1-\rho)(\tau^2+c^2)} = \frac{c^2\sqrt{1-\rho}}{\sqrt{\tau^2+c^2}}$, which matches $\eta$ of Theorem~\ref{thm-noncentral} up to the sign and the $|\beta_0|/\sigma$ factor. Since $X^{*\prime}M_{K_n}X^*\to_p\tau^{*2}$ by Lemma~\ref{lem-atten}, rearranging $z\phi(z)\eta^2\le\delta$ of Corollary~\ref{cor-cv} gives~\eqref{eq-feasible}. See the Online Supplement.
\end{proof}

Equation~\eqref{eq-feasible} is the form we recommend reporting: $\tau^{*2}_n$ is printed by any FE regression, $\sigma$ is the CJN residual standard deviation, and $(1-\lambda_n)$, the within noise share, is the single quantity requiring external input. At $\alpha = 0.05$, $\eta^{\mathrm{quad}}=\sqrt{\delta/(z\phi(z))}=\sqrt{8.73\,\delta}$: for instance, tolerating one percentage point of size distortion ($\delta = 0.01$) requires $(|\beta_0|/\sigma)(1-\lambda_n)\sqrt{\tau^{*2}_n} \le 0.296$ (the exact-inversion bound $\eta^{\dagger}(0.05,0.01)=0.295$ nearly coincides, confirming the quadratic form).

\subsection{Pilots and the correction}\label{sec-pilot}

The threshold depends on four quantities. Only $\rho_n = d_{K_n}/n$ and $\tau^{*2}_n$ are directly observable, and $\sigma$ is estimated by $\hsigma^*_{\mathrm{CJN}}$ (shown consistent in the proof of Theorem~\ref{thm-noncentral}). The remaining inputs are $|\beta_0|$ and the noise level. Two warnings are in order, the first of which is substantive.

\begin{proposition}[The naive pilot is anti-conservative; the corrected pilot removes it]\label{prop-pilot}
Under the conditions of Theorem~\ref{thm-noncentral}, suppose the noise pilot $\widehat{a}_n \to_p c^2(1-\rho)$ (equivalently $\hlambda_n := \tau^{*2}_n{}^{-1}(\tau^{*2}_n - \widehat a_n) \to_p \lambda$), obtained from any of the sources in Remark~\ref{rem-est}. Then:
\begin{enumerate}[label=(\roman*)]
\item \emph{(Anti-conservative centering.)} By Corollary~\ref{cor-slope} the naive point estimate is asymptotically centered at the attenuated value $\beta_0\lambda$ --- the center of its weak-information limiting law, where $\lambda < 1$ and attenuation bites; under strong information $\lambda \to 1$ and the attenuation vanishes --- so $|\hbeta^*_K|$ is centered at $\lambda|\beta_0| < |\beta_0|$. Since the distortion $z\phi(z)\eta^2$ is strictly increasing in $|\beta_0|$, evaluating it at this center understates the true distortion by the factor $\lambda^2$, so the naive diagnostic is biased toward passing specifications that fail the true threshold. Two caveats bound this statement to the center: under weak information $|\hbeta^*_K|$ is a folded normal, not a constant, so its sampling spread is a separate source of uncertainty (addressed below); and because the limiting law of $\hbeta^*_K$ has second moment $\lambda^2\beta_0^2+\sigma^2/\tau^{*2}$, a squared plug-in must use the recentered $\hbeta_0^{\mathrm{corr}}$ rather than $\hbeta^{*2}_K$, whose limit is inflated by the variance term.
\item \emph{(Bias-corrected pilot.)} The within-reliability-corrected pilot
\[
\hbeta_0^{\mathrm{corr}} := \frac{\hbeta^*_K}{\hlambda_n}
\]
removes the attenuation bias: $\hbeta_0^{\mathrm{corr}} \Rightarrow \beta_0 + N(0,\sigma^2/(\lambda^2\tau^{*2}))$, i.e.\ it is asymptotically centered at $\beta_0$. In the strong-information regime $nQ_{K_n}\to\infty$ the Gaussian term vanishes and $\hbeta_0^{\mathrm{corr}}\to_p\beta_0$ (consistent). In the weak-information regime $\tau^2<\infty$ it is centered at $\beta_0$ but carries sampling error of order $\sigma/(\lambda\sqrt{\tau^{*2}})$; this is a coefficient-uncertainty source distinct from within-reliability uncertainty, and controlling it requires a confidence interval (or worst-case bound) for $\beta_0$, not merely the $\lambda$ band of Remark~\ref{rem-corr-noise}. Strong information removes the coefficient uncertainty; the precision of an external reliability pilot is a separate matter.
\end{enumerate}
\end{proposition}

\begin{proof}[Proof sketch]
(i) follows from the centering of $\hbeta^*_K$ at $\beta_0\lambda$ (Corollary~\ref{cor-slope}) and the strict monotonicity of $z\phi(z)\eta^2$ in $|\beta_0|$. (ii) With $\hbeta^*_K\Rightarrow\beta_0\lambda+N(0,\sigma^2/\tau^{*2})$ and $\hlambda_n\to_p\lambda>0$, Slutsky gives $\hbeta_0^{\mathrm{corr}}\Rightarrow\beta_0+N(0,\sigma^2/(\lambda^2\tau^{*2}))$, a constant $\beta_0$ under strong information. See the Online Supplement.
\end{proof}

\begin{remark}[Sampling noise of the corrected pilot at small $\lambda$]\label{rem-corr-noise}
The correction $\hbeta_0^{\mathrm{corr}}=\hbeta^*_K/\hlambda_n$ divides by an estimated within reliability, so its sampling variance is inflated by roughly $\hlambda_n^{-2}$ (delta method), and the inflation is largest precisely in the low-within-reliability regime where the diagnostic bites. Two distinct uncertainties must therefore be separated. \emph{Within-reliability} uncertainty (how noisy is the regressor?) is handled by reporting the diagnostic over an interval of $\lambda$ values and, for a formal certificate, supplying a lower confidence bound $\ell_n$. \emph{Coefficient} uncertainty is separate: by Corollary~\ref{cor-slope}, $\hbeta_0^{\mathrm{corr}} \Rightarrow \beta_0 + N(0,\sigma^2/(\lambda^2\tau^{*2}))$ is nondegenerate under weak information, so a point plug-in yields a descriptive verdict rather than a certificate. Proposition~\ref{prop-certificate} absorbs this uncertainty by replacing $|t|$ with $|t|+z_{1-\gamma_\beta}$ and comparing $\ell_n$ with the resulting certified breakdown reliability. The coefficient and reliability coverage errors add, without an independence assumption (Remark~\ref{rem-gamma-delta}). Under strong information the coefficient uncertainty vanishes; reliability uncertainty remains whatever the external pilot supplies.
\end{remark}

Proposition~\ref{prop-pilot}(i) is the point made in Section~\ref{sec-notjust}: the most natural implementation of the diagnostic is biased in the dangerous direction, by exactly the within reliability $\lambda$ whose smallness motivates running it. The correction costs one line.

\begin{remark}[Sources for the noise pilot]\label{rem-est}
For $\sigma_\nu^2$ (and hence $c^2=n\sigma_\nu^2$ at the sample's local-drift scale, or the equivalent within noise share $1-\lambda_n$), three sources are available in decreasing order of directness. \emph{(i) Measurement-model posteriors}: producers of some expert-coded or latent-trait indices publish observation-level posterior uncertainty, as the V-Dem model and the Unified Democracy Scores of \citet{PMM} do, making $\hsigma_\nu^2$ directly available; this is the basis of the lead application. \emph{(ii) Validated survey, administrative-matched, or repeated measures}: after applying the identical projection and sample to two measurements $X^*_{i,1},X^*_{i,2}$, mutually uncorrelated classical reporting errors identify the first measure's within reliability as $\Cov(X^*_{1},X^*_{2})/\Var(X^*_{1})$ without requiring equal error variances. Under the additional equal-error-variance restriction, $\hsigma_\nu^2=\tfrac12\E_i[(X^*_{i,1}-X^*_{i,2})^2]$. With three common-scale measurements, the familiar difference-variance decomposition identifies product-specific noise variances only under pairwise-uncorrelated errors; this is a transparent special case of the multiple-proxy model of \citet{BernalMittagQureshi2016}, whose Assumption A3 states the required mutual orthogonality explicitly and whose design chooses proxies reported by different people. Their reported OLS attenuation of $44\%$ for the teaching proxy and $59\%$ for the resource proxy corresponds, as a rough common-scale calibration rather than an identification result, to reliabilities of about $0.56$ and $0.41$---the same range reached by the V-Dem pilots below. Different pairwise correlations can reject an equal-error-variance convention, but with unequal noise variances they do not by themselves reject classical error; under common scale, equality of the three pairwise \emph{covariances} is the relevant overidentifying implication. Reinterview and test--retest studies supply repeats; a survey measure matched to an administrative benchmark supplies a validation discrepancy, subject to error in the benchmark. Beyond the raw repeat formulas, \citet{EvdokimovZeleneev2025} show how a second measurement $Q=\alpha_1X+\epsilon_Q$, with unknown scale $\alpha_1$ and potentially nonclassical surrogate error, can identify the low-order noise moment inside a corrected-moment model. \citet{MeyerMittag2021} show how population administrative records linked to survey microdata can supply precisely this validation information. The twin reports of \citet{AshenfelterKrueger1994} furnish the second application, while \citet{Rouse1999} shows why correlation across reporting errors must be examined explicitly. Other examples include the CPS--Social Security earnings match of \citet{BK} and the PSID Validation Study of \citet{BoundEtAl94}. \emph{(iii) Sensitivity bands}: absent external information, evaluate the diagnostic over a grid of $\lambda$ and report the breakdown reliability (Definition~\ref{def-breakdown}). Thus the point diagnostic is usable beyond specialist expert-coded data whenever a defensible within-reliability range exists; a formal certificate still requires a defensible lower within-reliability bound. A sensitivity grid alone describes where the verdict changes but cannot manufacture that external information.
\end{remark}

\subsection{Breakdown reliability}\label{sec-breakdown}

The notation separates a measured property from a diagnostic requirement: $\lambda$ is the data's within reliability, whereas $\lambda^{\dagger}$ is the minimum within reliability the specification requires.

\begin{definition}[Breakdown reliability]\label{def-breakdown}
Fix $0<\alpha<1$ and $0<\delta<1-\alpha$. The \emph{breakdown reliability} $\lambda^{\dagger}$ of a specification is the minimum within reliability at which the feasible size condition~\eqref{eq-feasible}, evaluated with the within-reliability-corrected pilot $\hbeta_0^{\mathrm{corr}} = \hbeta^*_K/\lambda$, holds --- i.e.\ the fixed point $\lambda = \lambda^{\dagger}(\lambda)$ of the required-within-reliability map. Writing $t^*_n := |\hbeta^*_K|\sqrt{\tau^{*2}_n}/\hsigma^*_{\mathrm{CJN}}$ (the specification's conventional $t$-statistic), the fixed point has the closed form
\[
\lambda^{\dagger} := \frac{t^*_n}{\,t^*_n + \eta^{\dagger}(\alpha,\delta)\,},
\]
with the quadratic version replacing $\eta^{\dagger}$ by $\sqrt{\delta/(z\phi(z))}$. The descriptive point diagnostic passes iff its supplied within reliability $\hlambda_n \ge \lambda^{\dagger}$; under weak information this random plug-in comparison is not itself a population adequacy guarantee. When $\beta_0 = 0$ the population non-centrality vanishes and every within reliability is adequate (the population analogue of $t^*_n$ is zero, so $\lambda^{\dagger}=0$). The fixed point is the coherent point evaluation: at the boundary $\hlambda_n = \lambda^{\dagger}$ the corrected pilot $\hbeta^*_K/\lambda^{\dagger}$ is evaluated at exactly the within reliability used in the comparison. Two practical properties follow from the closed form: $\lambda^{\dagger}$ is computed from the \emph{printed regression output alone} --- indeed from the reported $t$-statistic alone --- with no noise pilot, and it is increasing in $|t^*_n|$: the more precisely a specification claims to estimate a large standardized effect, the higher the breakdown reliability. This is the headline number of the empirical section (Table~\ref{tab-vdem}), where we report it against known within reliabilities. \emph{Terminology (used without exception in what follows).} The paper reports two verdicts and never uses one word for both.
\begin{itemize}[leftmargin=1.4em,itemsep=1pt,topsep=2pt]
\item A specification obtains a \textbf{point pass} when $\hlambda_n \ge \lambda^{\dagger}$ with the pilot evaluated at its point estimate. This is a \emph{descriptive} verdict: it says the plug-in diagnostic clears the threshold, and it carries no coverage guarantee, because it ignores the sampling uncertainty of the coefficient pilot.
\item A specification is \textbf{certified} when a supplied within-reliability lower bound $\ell_n$ exceeds the certified breakdown reliability $\lambda^{\dagger}_{\gamma_\beta,n}$ of~\eqref{eq-certified-breakdown}. Equivalently, replace $|t|$ by $|t|+z_{1-\gamma_\beta}$ in the point-breakdown formula. This is the conservative, size-controlled object, and it is the only one the word ``certified'' is used for.
\end{itemize}
Every verdict in every table, figure and paragraph below is labelled as one or the other; a point pass is never reported as a certificate, and where a specification is a point pass but not a certificate we say so explicitly. The point verdict compares $\hlambda_n$ with $\lambda^{\dagger}$; the certificate compares a defensible lower bound $\ell_n$ with $\lambda^{\dagger}_{\gamma_\beta,n}$ and accounts for both coefficient and reliability uncertainty as in Proposition~\ref{prop-certificate}.
\end{definition}

The primitive threshold $\tau^2_{\mathrm{crit}}$ cannot have a universal
Stock--Yogo lookup: it depends on $(\rho,c^2,|\beta_0|/\sigma)$. The feasible
fixed point is different. Once the reported $|t^*_n|$ is given, $\rho$ and the
other primitives have already been absorbed into that statistic, so
$\lambda^{\dagger}$ has the compact lookup in Table~\ref{tab-breakdown-lookup}.

\begin{table}[ht]
\centering
\caption{Point and certified breakdown reliability from the reported $t$-statistic
($\alpha=\delta=0.05$, $\eta^{\dagger}=0.652$).}
\label{tab-breakdown-lookup}
\begin{tabular}{lrrrrrrr}
\toprule
$|t^*_n|$ & $0.5$ & $1.0$ & $1.96$ & $2.0$ & $3.0$ & $5.0$ & $10.0$ \\
\midrule
$\lambda^{\dagger}$ (point) & $0.434$ & $0.605$ & $0.750$ & $0.754$ & $0.821$ & $0.885$ & $0.939$ \\
$\lambda^{\dagger}_{0.05}$ (certified) & $0.767$ & $0.802$ & $0.847$ & $0.848$ & $0.877$ & $0.911$ & $0.947$ \\
\bottomrule
\end{tabular}
\par\smallskip
\begin{minipage}{0.93\textwidth}\footnotesize
Notes: $\lambda^{\dagger}=|t^*_n|/(|t^*_n|+0.652)$ and
$\lambda^{\dagger}_{0.05}=(|t^*_n|+1.645)/(|t^*_n|+1.645+0.652)$, treating
the supplied reliability as known ($\gamma_\lambda=0$). There is no separate
$\rho$ row because, conditional on the printed $t$-statistic, the fixed-point
threshold is exactly $\rho$-free; $\rho$ remains present in the primitive
$\tau^2$ threshold and in the statistic's sampling law.
\end{minipage}
\end{table}

\subsection{Verdict margins and clustered implementation}

\begin{proposition}[Exact-normal verdict margin]\label{prop-verdict-margin}
Fix $0<\alpha<1$ and $0<\delta<1-\alpha$ and write $z=z_{1-\alpha/2}$ and
\[
s(r):=\Phi(-z-r)+1-\Phi(z-r),\qquad r\ge0.
\]
Then: (i) $s(r)$ is strictly increasing for $r>0$, with
$s(0)=\alpha$ and $s(\eta^{\dagger})=\alpha+\delta$; (ii) at any supplied
within reliability $\ell\in(0,1]$, the corrected-pilot point diagnostic satisfies
\[
\widehat r(\ell):=|t^*_n|\frac{1-\ell}{\ell}\le\eta^{\dagger}
\quad\Longleftrightarrow\quad
\ell\ge\lambda^{\dagger};
\]
and (iii) if an alternative calibration produces a non-centrality magnitude
$\widetilde r$, its pass/fail verdict agrees with that based on $\widehat r$
whenever
\[
|\widetilde r-\widehat r|<|\widehat r-\eta^{\dagger}|.
\]
Thus $|\widehat r-\eta^{\dagger}|$ is an explicit robustness margin for the
binary verdict within the non-central limit experiment. This margin is exact
irrespective of how far $\widehat r$ lies into the tail; it is not a claim
about the finite-sample calibration accuracy of the associated percentage
$s(\widehat r)$ in that regime.
\end{proposition}

\begin{proof}
For $r>0$, $s'(r)=\phi(z-r)-\phi(z+r)>0$, because
$|z-r|<z+r$ (using $z>0$, which holds since $\alpha<1$) and the
standard-normal density is strictly decreasing in its absolute argument, so
$s$ is strictly increasing on $(0,\infty)$. As $r\to\infty$,
$\Phi(-z-r)\to0$ and $\Phi(z-r)\to0$, so $s(r)\to1$; $s$ is continuous. Since
$s(0)=\alpha$ and $\alpha+\delta\in(\alpha,1)$ (using $0<\delta<1-\alpha$), the
intermediate value theorem gives a unique $r=\eta^{\dagger}\in(0,\infty)$ with
$s(\eta^{\dagger})=\alpha+\delta$, which is part (i). For (ii), multiply
$|t^*_n|(1-\ell)/\ell\le\eta^{\dagger}$ by $\ell>0$ and rearrange to obtain
$\ell\ge |t^*_n|/(|t^*_n|+\eta^{\dagger})$. For (iii), the verdict is
$r\le\eta^{\dagger}$ (pass) versus $r>\eta^{\dagger}$ (fail): if
$\widehat r\le\eta^{\dagger}$ then
$\widetilde r\le\widehat r+|\widetilde r-\widehat r|<\widehat
r+(\eta^{\dagger}-\widehat r)=\eta^{\dagger}$, the same verdict; if
$\widehat r>\eta^{\dagger}$ then
$\widetilde r\ge\widehat r-|\widetilde r-\widehat r|>\widehat
r-(\widehat r-\eta^{\dagger})=\eta^{\dagger}$, again the same verdict.
\end{proof}

\begin{remark}[Clustered and serially dependent regression errors]\label{rem-cluster}
Applied panel practice clusters standard errors by unit \citep{Arellano1987, BDM}; see \citet{CameronMiller2015} and \citet{MNW2023} for current practice. The cluster-robust version of the diagnostic is not a heuristic rescaling but a theorem: Section~\ref{sec-cluster} establishes a cluster-level score CLT (Lemma~\ref{lem-cluster-clt}), consistency of the Arellano variance estimator in the many-fixed-effect regime (Lemma~\ref{lem-crve}), and the limit $T^{CR}_n(\beta_0)\Rightarrow N(\eta/\sqrt{\psi},1)$ (Theorem~\ref{thm-cluster}). Clustering enters through the denominator alone, because the attenuation bias in the numerator of $\eta$ is a functional of $\E[\nu'M_{K_n}\nu]$ and $\nu\perp u$ by Assumption~\ref{ass-nu}(i); the feasible consequence is
\[
|\eta_{CR}| \;=\; \frac{|\eta|}{\sqrt{\psi}}, \qquad
|\widehat\eta_{CR}| \;=\; \frac{|\widehat\eta|}{\sqrt{\hpsi}}, \qquad
\lambda^{\dagger}_{CR} \;=\; \frac{t^*_n/\sqrt{\hpsi}}{\,t^*_n/\sqrt{\hpsi} + \eta^{\dagger}(\alpha,\delta)\,},
\]
with $\hpsi = \tau^{*2}_n\widehat V_{CR}/\hsigma^{*2}_{\mathrm{CJN}}$ the reported variance-inflation factor. The first display is the population target and the second the computed diagnostic; they are separate objects, since $\hpsi$ estimates $\psi\sigma^2/\varsigma^2$ rather than $\psi$ (Corollary~\ref{cor-cluster-feasible}(a), Remark~\ref{rem-psi-sign}). Corollary~\ref{cor-cluster-feasible} shows this recipe is legitimate even though, under cluster dependence, neither $\hpsi$ nor the feasible $|\widehat\eta|$ is individually consistent for its i.i.d.\ target: the two inconsistencies are the same factor and cancel. Equivalently, the cluster-robust diagnostic is the i.i.d.\ diagnostic run on the reported cluster-robust $t$-statistic.

Three conditions must be checked, and all three are cheap. (1) \emph{Many clusters}: $G_n\to\infty$ with no dominant cluster; the few-clusters regime is not covered (Remark~\ref{rem-cluster-scope}). (2) \emph{Projection compatibility}: fixed effects nested within clusters cost nothing, while Lemma~\ref{lem-nest}(b) supplies the coarse sufficient condition $d^{\mathrm{ne}}_n/G_n\to0$ for a balanced non-nested block. The direct sample analogue is more informative:
\[
\widehat\chi_{\mathrm{proj}}
:=\frac{\sum_g\|Ma^{(g)}-a^{(g)}\|^2}{\tau^{*2}_n}.
\]
Although the conservative rank ratio is $d^{\mathrm{ne}}_n/G_n=59/163\approx0.36$ for V-Dem, direct computation gives $\widehat\chi_{\mathrm{proj}}=0.0069$--$0.0073$ across its three static treatments. The realized loadings therefore satisfy the compatibility requirement closely; this finite-panel diagnostic supports the reported cluster columns but, by itself, does not prove the asymptotic sequence condition. The twin-pair application is i.i.d.\ and has no cluster column. (3) \emph{No degrees-of-freedom correction}: the conventional $\tfrac{n-1}{n-K_n}$ factor inflates the CRVE by an asymptotic factor $1/(1-\rho)$ and must be omitted (Lemma~\ref{lem-crve}).

Finally, the direction of the adjustment is an empirical matter and must be read off $\hpsi$ rather than assumed: clustering deflates the measured distortion only when the regressor is persistent as well as the error, and otherwise amplifies it (Remark~\ref{rem-psi-sign}). It deflates it in the V-Dem application, where $\hpsi$ runs from $19.2$ to $25.2$.
\end{remark}

\subsection{Protocol}\label{sec-protocol}

\begin{framed}
\noindent\textbf{Diagnostic protocol.} Given a saturated FE regression of $Y$ on observed $X^*$:
\begin{enumerate}[label=\arabic*., nosep, start=0]
\item \emph{Scope check.} Verify that $X^*$ is a continuous regressor whose measurement error is plausibly classical (Assumption~\ref{ass-nu}). If $X$ is binary, \textbf{stop}: misclassification error is nonclassical by construction and the diagnostic does not apply (Section~\ref{sec-scope}).
\item Record $\hrho_n = d_{K_n}/n$, $\tau^{*2}_n = X^{*\prime} M_{K_n} X^*$, $\hsigma^*_{\mathrm{CJN}}$, and $\hbeta^*_K$ from the regression output.
\item Obtain a noise pilot: $\hsigma_\nu^2$ from a measurement model, a validation study, or a within-reliability grid (Remark~\ref{rem-est}); set $\widehat a_n = \hsigma_\nu^2 (n - d_{K_n})$ and $\hlambda_n = 1 - \widehat a_n / \tau^{*2}_n$. If $\hlambda_n\le0$, flag the specification and stop; the corrected pilot and certificate require positive within reliability.
\item Correct the coefficient pilot: $\hbeta_0^{\mathrm{corr}} = \hbeta^*_K / \hlambda_n$ (Proposition~\ref{prop-pilot}).
\item \emph{Descriptive point verdict.} Compute the feasible non-centrality $|\widehat\eta| = (|\hbeta_0^{\mathrm{corr}}|/\hsigma^*_{\mathrm{CJN}})\,(1-\hlambda_n)\sqrt{\tau^{*2}_n}$ and compare it to the exact-inversion root $\eta^{\dagger}(\alpha,\delta)$ (or the quadratic $\eta^{\mathrm{quad}}=\sqrt{\delta/(z_{1-\alpha/2}\phi(z_{1-\alpha/2}))}$); equivalently compare $\hlambda_n$ to the breakdown $\lambda^{\dagger} = t^*_n/(t^*_n+\eta^{\dagger})$ of Definition~\ref{def-breakdown}, or the distortion $z_{1-\alpha/2}\phi(z_{1-\alpha/2})|\widehat\eta|^2$ to $\delta$, or $\tau^{*2}_n - \widehat a_n$ to $\tau^2_{\mathrm{crit}}$ of~\eqref{eq-cv}. This point verdict is descriptive.
\item \emph{Cluster-robust verdict.} If the reported inference clusters standard errors, form the variance-inflation factor $\hpsi$ from the CRVE --- without the $\tfrac{n-1}{n-K_n}$ small-sample factor, which over-corrects by $1/(1-\rho)$ (Lemma~\ref{lem-crve}) --- and repeat step~4 with $|\widehat\eta|/\sqrt{\hpsi}$ and $\lambda^{\dagger}_{CR}$ (Theorem~\ref{thm-cluster}). First check projection compatibility. The direct route is to evaluate $\sum_g\|Ma^{(g)}-a^{(g)}\|^2/\tau^{*2}_n$, which costs one pass over the residualized regressor; the shortcut is to verify many clusters together with $d^{\mathrm{ne}}_n/G_n$ small, where $d^{\mathrm{ne}}_n$ counts fixed effects not nested within clusters, but that reduction is licensed only under the balance conditions (N1)--(N2) of Lemma~\ref{lem-nest}(b) and should not be used without them. Report both verdicts; when they differ, the cluster-robust one is the one that matches the reported standard errors.
\item \emph{Formal certificate (conservative; the only step that yields a certificate).} Let $t$ be the operative absolute reported $t$-statistic---cluster robust when step~5 applies. Choose a coefficient-uncertainty budget $\gamma_\beta$ and compute
\[
\lambda^{\dagger}_{\gamma_\beta,n}
=\frac{|t|+z_{1-\gamma_\beta}}
{|t|+z_{1-\gamma_\beta}+\eta^{\dagger}(\alpha,\delta)}.
\]
Compare it with a defensible lower bound $\ell_n$ on within reliability. If reliability is known or consistently estimated, take $\ell_n=\hlambda_n$ and $\gamma_\lambda=0$; for a noisy finite-sample pilot, use a $(1-\gamma_\lambda)$ lower confidence bound. Certify iff $\ell_n\ge\lambda^{\dagger}_{\gamma_\beta,n}$. Proposition~\ref{prop-certificate} then bounds false certification by $\gamma_\beta+\gamma_\lambda$, without requiring the two bounds to be independent. A noisy reliability point estimate inserted as if known has no such guarantee (Design~6).
\item Report $\hlambda_n$ (the power tax of Proposition~\ref{prop-power}), the lower bound $\ell_n$ and uncertainty budgets used for any certificate, and the sensitivity band over $\lambda$ alongside the verdict. State in every table and figure whether the entry is a point pass, a certificate, or a cluster-rescaled reading. If the specification fails, use the Griliches--Hausman estimators or a fixed-effect Anderson--Rubin confidence set (companion work in progress).
\end{enumerate}
\end{framed}

Unlike the IV case, no universal numerical table for $\tau^2$ can exist. Stock--Yogo's IV thresholds are unit-free constants; here the primitive threshold is data-dependent through $|\beta_0|/\sigma$ and the external noise scale. What can be tabulated is the breakdown reliability as a function of the reported $t$-statistic, as Table~\ref{tab-breakdown-lookup} does. The feasible one-line check~\eqref{eq-feasible} together with Definition~\ref{def-breakdown} then delivers the verdict from regression output plus the noise pilot. When a numerical threshold for $\tau^2$ itself is wanted, it should be read from the exact inversion~\eqref{eq-cv-exact}; the quadratic closed form~\eqref{eq-cv} is mildly anti-conservative ($\eta^{\dagger} < \eta^{\mathrm{quad}}$), with exact-to-quadratic threshold ratios from about $1.01$ far from the vanishing boundary to $1.34$ near it.

\section{Simulations}\label{sec-sim}

Six Monte Carlo designs assess the calibration, threshold accuracy, plug-in correction, finite-$n$ mapping, cluster-robust extension, and false-certification guarantee. Designs~1, 4, 5, and 6 are reported here; Designs~2 and 3 are supporting checks in the Online Supplement. Designs~1--4 and~6 share a two-way FE DGP (units $i=1,\dots,N$, periods $t=1,\dots,T$, $d_K=N+T-1$, $\rho_n=(N+T-1)/(NT)$). Fixed effects are removed by the within transformation, so we generate $X_{it}=s\,\varepsilon_{it}$ with $s=\sqrt{\tau^2/n}$, making the primitive $nQ_{K_n}=\tau^2$ hit the target grid exactly, $u_{it}\sim N(0,\sigma^2)$, and $\nu_{it}\sim N(0,\sigma_\nu^2)$. Under this calibration Lemma~\ref{lem-clt} predicts $\E[X'M_{K_n}X]=(1-\rho_n)\tau^2$, not $\tau^2$ --- a prediction the simulations test directly (Design~1). We state the calibration convention explicitly because the alternative, $s^2=\tau^2/(n-d_K)$ (which holds $\E[X'M_{K_n}X]=\tau^2$), makes the $(1-\rho)$ discount invisible to the simulation and so cannot adjudicate it. Design~5 replaces the i.i.d. structural errors with within-cluster AR(1) dependence and is described separately below. Seeds are fixed and the Julia scripts are released with the paper. Reported table numbers use $4{,}000$--$5{,}000$ replications per cell; the calibration figure (Figure~\ref{fig-design1}) and Design~1 use $6{,}000$, while the tail-probability experiment in Design~6 uses $20{,}000$.

\subsection{Design 1 (calibration)} With $\sigma_\nu^2=c^2/n$, we verify that the empirical mean of $T_n^{\mathrm{CJN}*}$ tracks the $\eta_n$ of Theorem~\ref{thm-noncentral} and that empirical size matches the exact non-central prediction $\Phi(-z-\eta_n)+1-\Phi(z-\eta_n)$ across the $(\tau^2,c^2)$ grid. At $N=200$, $T=5$ ($\rho_n=0.204$), $|\beta_0|/\sigma=1$: for $(\tau^2,c^2)=(1,5)$, $\eta_n=-1.82$, mean $T = -1.79$, size $0.433$ (predicted $0.445$); for $(5,1)$, $\eta_n=-0.36$, mean $T=-0.36$, size $0.069$ (predicted $0.065$); for $(10,5)$, $\eta_n=-1.15$, mean $T=-1.17$, size $0.218$ (predicted $0.210$). The empirical within reliability matches the $\rho$-free $\lambda_n$ of~\eqref{eq-lambda} across the grid (maximum deviation $0.003$), and the empirical Hessian ratio $X'M_{K_n}X/(nQ_{K_n})$ sits on $1-\rho_n=0.796$ in every cell --- squarely on the limit of Lemma~\ref{lem-clt}, and far from the naive value $1$. The empirical mean and rejection frequency of $T_n^{\mathrm{CJN}*}$ match the $N(\eta_n,1)$ prediction across the grid (Figure~\ref{fig-design1}), confirming Theorem~\ref{thm-noncentral}.

\begin{figure}[htbp]\centering
\includegraphics[width=0.86\linewidth]{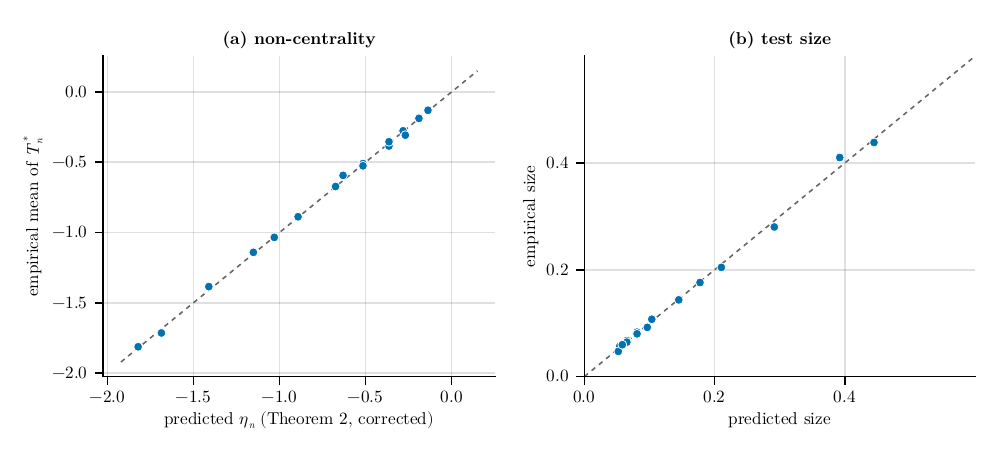}
\caption{Design 1 calibration ($N=200$, $T=5$, $6{,}000$ replications per cell over $\tau^2\in\{1,2,5,10\}$, $c^2\in\{0.5,1,2,5\}$). (a) empirical mean of the CJN $t$-statistic against the theoretical non-centrality $\eta_n$; (b) empirical size against the exact non-central prediction. Points lie on the $45^\circ$ line (dotted), confirming $T_n^{\mathrm{CJN}*}\Rightarrow N(\eta_n,1)$.}\label{fig-design1}
\end{figure}

\subsection{Designs 2 and 3 (supplement)} Two supporting checks are reported in full in the Online Supplement. Design~2 (threshold accuracy) sweeps $\tau^2$ and locates the empirical size-crossing: the corrected closed-form threshold~\eqref{eq-cv} matches the empirical threshold to within Monte Carlo accuracy (about $5\%$ on average, at most $9\%$ in any cell) with no systematic dependence on $(N,T,c^2)$, whereas the discarded linear surrogate over-demands the residual signal by a factor of $2.5$--$2.8$ and is anti-conservative at larger $\delta$. Design~3 (pilot sensitivity) confirms Proposition~\ref{prop-pilot}(i) on simulated data --- the naive attenuated pilot passes the point diagnostic for a specification whose true size is $16\%$, which the corrected pilot correctly flags --- and shows that misspecifying the noise pilot by a factor $k\in\{0.5,\dots,2\}$ maps monotonically and boundedly into the distortion estimate, motivating a within-reliability sensitivity band around the breakdown reliability (Definition~\ref{def-breakdown}).

\subsection{Design 4 (the finite-$n$ mapping check)} For each cell we fix a finite, non-drifting numerical $\sigma_\nu^2$ (a single constant per cell, not an asymptotic sequence) chosen to hit a target within reliability $\lambda_n\in\{0.95,0.9,0.8,0.65,0.5,0.3,0.2\}$, generate finite-sample data, and compare the exact empirical size to both the quadratic prediction $\alpha+z_{1-\alpha/2}\phi(z_{1-\alpha/2})\eta_n^2$ and the exact non-central prediction, with $\eta_n$ read off the finite-$n$ mapping~\eqref{eq-eta-n} (no drifting sequence appears in that mapping). Figure~\ref{fig-design4} and Table~\ref{tab-design4} show that the exact non-central prediction matches empirical size at every $\lambda_n$, and that size depends on the design only through $\lambda_n$. Reproducing the same $\lambda_n$ at $n=1000$ and $n=5000$ requires $\sigma_\nu^2=\tau^2(1-\lambda_n)/(\lambda_n n)=O(1/n)$, which is exactly the local drift, so the coincidence across $n$ is not a claim of robustness to a genuinely fixed $\sigma_\nu^2$ as $n\to\infty$ (under which $\lambda_n\to0$ and size $\to1$); it confirms that $\lambda_n$ is the sufficient statistic the drift makes it, and that the finite-$n$ non-central mapping is accurate at fixed finite noise. The drift is thus an \emph{approximation device}, not an artifact of the asymptotics. The quadratic prediction is accurate for $\lambda_n\ge0.8$, under-predicts mildly in the band $\lambda_n\in[0.5,0.65]$ (by at most $1.5$ percentage points), and over-predicts in the far tail ($\lambda_n\le0.3$); once $\eta_n$ is large the exact inversion~\eqref{eq-cv-exact} should be used, the behavior anticipated in Section~\ref{sec-drift}. The new $\lambda_n=0.20$ cell reaches $|\eta_n|=3.58$: empirical size is $0.943$ at $n=1000$ and $0.939$ at $n=5000$, against the exact-normal predictions $0.946$ and $0.947$. This extends the calibrated range past the twin application of Section~\ref{sec-app-twins}.

\begin{figure}[htbp]\centering
\includegraphics[width=0.66\linewidth]{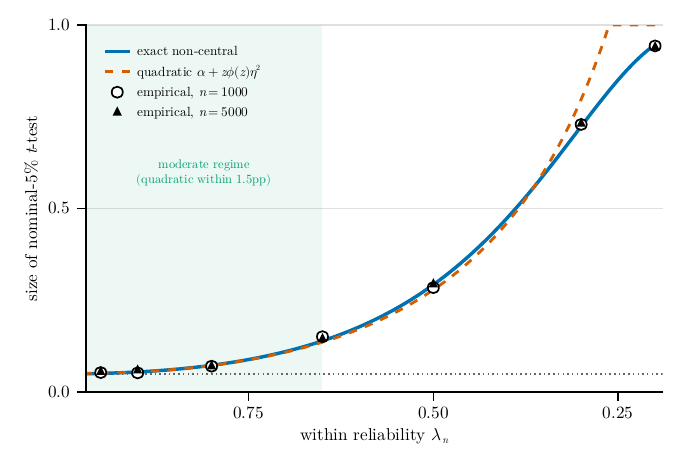}
\caption{Design 4 (fixed-noise validation). Size of the nominal-$5\%$ CJN $t$-test under fixed finite-$n$ noise (a single $\sigma_\nu^2$ per cell), against within reliability $\lambda_n$ (worsening to the right). The exact non-central prediction (solid) tracks the empirical size at both $n=1000$ and $n=5000$, which coincide, through $|\eta_n|=3.58$ and predicted size $0.95$. The quadratic prediction $\alpha+z\phi(z)\eta_n^2$ (dashed) tracks the exact curve to within $1.5$ percentage points in the moderate regime ($\lambda_n\ge0.65$, shaded) and departs from it once attenuation is severe (under-predicting near $\lambda_n=0.5$, over-predicting in the far tail). The drift is an approximation device, not an artifact.}\label{fig-design4}
\end{figure}

\begin{table}[htbp]\centering
\caption{Design 4. Exact finite-sample size of the nominal-$5\%$ CJN $t$-test under fixed finite-$n$ noise (a single $\sigma_\nu^2$ per cell) at target within-reliability $\lambda_n$, versus the quadratic and exact non-central predictions. DGP: two-way FE, $\tau^2=5$, $|\beta_0|/\sigma=1$; $5{,}000$ replications. The exact prediction tracks empirical size to within about one percentage point at every $\lambda_n$; size is invariant to $n$ at fixed $\lambda_n$. The quadratic prediction matches through $\lambda_n=0.8$, under-predicts mildly at $\lambda_n\in\{0.65,0.5\}$, and over-predicts in the far tail. The $\lambda_n=0.20$ cell validates the exact mapping through $|\eta_n|=3.58$.}\label{tab-design4}
\begin{adjustbox}{max width=\textwidth}
\begin{tabular}{lcccc}
\toprule
$\lambda_n$ & $n$ & empirical size & quadratic pred. & exact pred. \\
\midrule
$0.95$ & $1000$ & $0.053$ & $0.051$ & $0.051$ \\
$0.80$ & $1000$ & $0.071$ & $0.073$ & $0.073$ \\
$0.65$ & $5000$ & $0.149$ & $0.136$ & $0.140$ \\
$0.50$ & $5000$ & $0.290$ & $0.279$ & $0.293$ \\
$0.30$ & $5000$ & $0.726$ & $0.798$ & $0.724$ \\
$0.20$ & $5000$ & $0.939$ & $1.000$ & $0.947$ \\
\bottomrule
\end{tabular}
\end{adjustbox}
\end{table}

\subsection{Design 5 (cluster-robust validation)}\label{sec-design5}
The cluster theory of Section~\ref{sec-cluster} is validated on a unit-and-time FE design with $G$ units (which are also the clusters), $T$ periods, within-unit AR(1) errors, and clustering by unit, so that unit effects are nested in clusters and the $T-1$ time effects are not. Seven predictions are checked against \texttt{verify\_cluster\_clt.py} and \texttt{verify\_cluster\_high\_eta.jl} ($2{,}500$--$4{,}000$ replications per configuration), all confirmed to within Monte Carlo error: \emph{(i)} the limiting law itself, matching to within $0.04$ in mean and $0.03$ in standard deviation across configurations including strongly persistent designs; \emph{(ii)} $O(1/G_n)$ CRVE consistency, with $\widehat V^{\mathrm{sc}}_{CR}/\Psi_n$ rising from $0.981$ at $G=50$ to $0.998$ at $G=800$; \emph{(iii)} that breaking projection compatibility breaks the conclusion exactly as Lemma~\ref{lem-nest}(b) predicts, the condition required for the clustered application of Section~\ref{sec-app-vdem}; \emph{(iv)} that the conventional small-sample factor over-corrects by the predicted $1/(1-\rho_n)$, agreeing to three digits; \emph{(v)} the $\sigma$-cancellation of Corollary~\ref{cor-cluster-feasible} despite both of its ingredients being individually inconsistent; \emph{(vi)} that the sign of $\psi-1$ depends jointly on regressor and error persistence rather than on either alone --- three configurations at $G=200$ give $\psi=0.75$, $1.27$ and $1.72$, as Remark~\ref{rem-psi-sign} describes; and \emph{(vii)} the high-noncentrality mapping, where $|\eta_{CR}|=3.67$ gives empirical size $0.959$ against $0.956$ predicted. Full numerical details are in \ref{sec:supp-design5}.

\subsection{Design 6 (false-certification control)}\label{sec-design6}
This design targets Proposition~\ref{prop-certificate} directly. In four balanced two-way FE cells spanning $N\in\{120,240\}$, $T\in\{5,10\}$, $\tau^2\in\{2,8\}$, and $\lambda\in\{0.80,0.90\}$, we choose $\beta_0$ so that the true exact-normal non-centrality is $\eta^{\dagger}+0.01=0.6624$. Equivalently, true reliability lies $0.0014$--$0.0024$ below the oracle population counterpart of the point breakdown reliability in every cell. Its implied size is $0.1016$, just above the admissible $0.10$, so every certificate issued is false by construction. With known reliability and $\gamma_\beta=0.05$, the false-certification frequencies over $20{,}000$ replications per cell range from $0.0427$ to $0.0474$, never exceeding $0.05$. A noisy reliability point estimate treated as known does not inherit that control: with pilot standard error $0.020$, the rate ranges from $0.0598$ to $0.1935$. Using instead a $97.5\%$ lower confidence bound for reliability and assigning $\gamma_\beta=\gamma_\lambda=0.025$ keeps the total nominal bound at $0.05$; realized rates range from $0.0032$ to $0.0061$ (and from $0.0146$ to $0.0157$ with pilot standard error $0.005$). Thus the theorem's bound is respected, while the tempting noisy-point-pilot shortcut fails precisely where the amended rule says it can. Design~5 validates the clustered pivot to which the same algebra applies; Design~6 isolates the certificate's coverage event in the i.i.d.\ branch. The complete table and code are in Online Supplement Section~S7.

\section{Applications}\label{sec-app}

\subsection{Democracy and growth}\label{sec-app-vdem}

The ideal application has three properties: the regression is canonical and FE-saturated; the regressor is known to be noisily measured; and, rarest of all, the noise variance is externally observable rather than assumed. Country-year democracy--growth panels have all three. The specification of \citet{ANRR} regresses log GDP per capita on a democracy measure with country and year fixed effects (plus lags); the democracy regressor is constructed from expert-coded indices whose disagreement is well documented, and \citet{ANRR} themselves treat measurement error as a first-order concern, addressing it by instrumenting. Crucially, the V-Dem polyarchy index and the Unified Democracy Scores \citep{PMM} publish observation-level posterior standard deviations from an explicit measurement model: $\hsigma_{\nu, it}$ is data, not a calibration.

\emph{The within reliability is invariant to the outcome equation; the verdict is not.} The within reliability $\hlambda_n=1-\widehat a_n/\tau^{*2}_n$ depends only on the regressor's measurement model and the fixed-effect design, so the same $\hlambda_n$ applies to the static regression we run, to the dynamic specification of \citet{ANRR}, and to any other regression built on the same measure and FE structure. The verdict does not: it compares $\hlambda_n$ to a breakdown reliability $\lambda^{\dagger}$ that depends on the outcome-specific $t$-statistic and must be recomputed per specification (and Theorem~\ref{thm-noncentral}, derived for the static within regression, does not automatically transfer to dynamic GMM). Since the harder-won ingredient is the one that transfers, a reader can carry $\hlambda_n$ to their own growth equation and recompute the one-line breakdown reliability there; we report a transparent static regression to make that reading as clean as possible.

The exercise: run the country-and-year FE regression of log GDP per capita \citep{Maddison2020} on a continuous V-Dem index \citep{VDem16} as the treatment, 1960--2018; record $\hrho_n$, $\tau^{*2}_n$, $\hbeta^*$, $\hsigma^*_{\mathrm{CJN}}$; set $\widehat a_n = \overline{\hsigma^2_{\nu, it}}\,(n - d_{K_n})$ from the published posterior standard deviations, recovered as $\hsigma_{\nu,it}=(\text{codehigh}_{it}-\text{codelow}_{it})/2$; and run the protocol of Section~\ref{sec-protocol}.\footnote{The adequacy diagnostics in this section are computed \ifanon\else with the open-source \textsf{PanelAdequacy} software \citep{PanelAdequacySoftware} (the \texttt{eiv\_adequacy} routine)\fi; every entry in Table~\ref{tab-vdem} reproduces from the bundled V-Dem panel and supplementary replication archive. The twin-pair calculation uses the public repeated-report extract and the separate script described in \ref{sec:supp-twins}. One convention matters for V-Dem: Lemma~\ref{lem-crve} shows the CRVE must be formed without the $\tfrac{n-1}{n-K_n}$ small-sample factor, which inflates it by an asymptotic $1/(1-\rho)$; no V-Dem verdict turns on it, but users reproducing these numbers with default software settings (which apply the factor) should expect the discrepancy quantified in \ref{sec:supp-footnote}.}

The result, in Table~\ref{tab-vdem} and Figure~\ref{fig-vdem}, is a two-pole contrast within a single dataset and a single estimator, with an instructive middle case that only the cluster-robust rescaling exposes. The aggregate polyarchy index obtains a \emph{point pass}: within reliability $\hlambda_n=0.90$, above its breakdown reliability of $0.76$, implied size $5.6\%$ under i.i.d.\ standard errors and $5.0\%$ under the country-clustered standard errors an applied author would report. At the opposite pole, judicial constraints (a component on which expert coders disagree most) is flagged under the operative country-clustered inference: its point non-centrality is $|\widehat\eta_{CR}|=2.4$, nearly four times the breakdown value $\eta^{\dagger}=0.65$, even after the standard error is inflated by $\sqrt{\hpsi}\approx5$ ($\hpsi=25.2$). The i.i.d.\ comparison is still more extreme---$\hlambda_n=0.41$ against a breakdown reliability of $0.93$, a shortfall of $0.52$ in within reliability, and $|\widehat\eta|\approx12$---but is secondary because applied practice clusters by country. Legislative constraints is the middle case: flagged under i.i.d.\ standard errors ($\hlambda_n=0.55$ against a breakdown reliability of $0.62$, implied size $14\%$) but obtaining a point pass under the country-clustered standard errors an applied author would actually report ($\hpsi=24.0$, $\lambda^{\dagger}_{CR}=0.25$). The direct projection-compatibility ratios are only $0.0069$--$0.0073$ across the three treatments, so the realized design supports applying Theorem~\ref{thm-cluster} despite the conservative rank ratio $d_n^{\mathrm{ne}}/G_n=0.36$. Legislative constraints clears the cluster point threshold by $0.30$ in within reliability but not the certified threshold $0.74$. Its i.i.d.\ flag is also fragile: Proposition~\ref{prop-equicorr} shows that an equicorrelation share of about $\omega^{*}=0.26$ would overturn it. We therefore report the standardization-dependent middle row without leaning on it substantively. We lean heavily on the judicial-constraints flag, which remains far beyond the threshold after clustering.

\emph{On the implied-size column at large $|\eta|$.} The entries are exact non-central sizes~\eqref{eq-cv-exact} evaluated at the design-computed point non-centrality, not exact finite-sample rejection probabilities. Design~4 calibrates the i.i.d.\ local-drift mapping through $\lambda_n=0.20$, where $|\eta_n|=3.58$ and implied size reaches $0.95$. Design~5 reaches $|\eta_{CR}|=3.67$ under clustered errors, with empirical size $0.959$ against $0.956$ predicted. The cluster-rescaled judicial-constraints point, $|\widehat\eta_{CR}|=2.4$ with implied size $0.67$, lies inside both calibrated ranges; the i.i.d.\ point $|\widehat\eta|\approx12$ lies far outside them, so its size entry remains ordinal. Thus $0.67$ is a quantitative point approximation under the combined theory rather than an exact finite-sample claim. The binary verdict is considerably less demanding: with $\eta^{\dagger}=0.652$, Proposition~\ref{prop-verdict-margin} shows that the i.i.d.\ magnitude $12.10$ would have to fall by more than $11.45$ (a $94.6\%$ reduction), and the cluster magnitude $2.41$ by more than $1.76$ (a $72.9\%$ reduction), before either crosses the exact-normal threshold. These are robustness margins within the limit experiment, not finite-sample approximation-error bounds.

The mechanism is transparent: aggregate indices average over many indicators and so average out coding disagreement (small $\hsigma_\nu^2$), whereas the disaggregated constraint sub-indices retain it ($\hsigma_\nu^2$ three to seven times larger). The anti-conservative pilot of Proposition~\ref{prop-pilot} is visible in the table: the naive $|\hbeta^*|$ understates the corrected $|\hbeta_0^{\mathrm{corr}}|$ by the factor $\hlambda_n$ in every row. The economics is secondary here; the point is that the diagnostic is computable from published uncertainty estimates on a specification the literature cares about, and that it discriminates: one canonical regressor obtains a point pass with margin, a second is flagged under both standardizations, and a third changes from a flag to an uncertified point pass when inference is clustered as applied practice requires.

\begin{table}[htbp]\centering
\caption{Lead application. EIV adequacy diagnostic on a country--year democracy--growth panel (V-Dem regressors, Maddison log GDP per capita, country and year FE, 1960--2018, $N=163$ countries, $n\approx 8.5$--$8.9$k, $\hrho_n\approx0.025$). Posterior noise $\hsigma_{\nu,it}$ from the V-Dem measurement model. ``Breakdown'' is the fixed-point $\lambda^{\dagger}$ of Definition~\ref{def-breakdown} at $\delta=0.05$; a specification obtains a \emph{point pass} iff $\hlambda_n\ge\lambda^{\dagger}$. The ``verdict'' columns report point passes and flags. The formal certificate of Proposition~\ref{prop-certificate} is a separate, stronger object: conditional on treating the model-implied reliability as supplied ($\gamma_\lambda=0$), only aggregate polyarchy clears the certified breakdown at $\gamma_\beta=0.05$ (see the discussion following the table). ``Size'' is the local-drift implied size of the nominal-$5\%$ $t$-test, not an exact finite-sample rejection probability. Design~4 calibrates this mapping through implied size $0.95$, which covers the judicial-constraints cluster point ($0.67$) but not its i.i.d.\ point ($\ge0.99$); see the discussion following the table. The last three columns apply the cluster-robust diagnostic of Theorem~\ref{thm-cluster}: $\hpsi$ is the CRVE variance-inflation factor and the verdict compares $\hlambda_n$ to $\lambda^{\dagger}_{CR}$. Although the coarse sufficient-condition ratio is $d^{\mathrm{ne}}_n/G_n=59/163\approx0.36$, the direct realized projection ratios $\widehat\chi_{\mathrm{proj}}$ are $0.0069$--$0.0073$; this supports the cluster columns for this design but is not a proof of the asymptotic sequence condition.}\label{tab-vdem}
\begin{adjustbox}{max width=\textwidth}
\begin{tabular}{lcccccccc}
\toprule
& \multicolumn{2}{c}{pilots} & \multicolumn{3}{c}{i.i.d.\ standard errors (theory)} & \multicolumn{3}{c}{clustered by country (rescaling)} \\
\cmidrule(lr){2-3}\cmidrule(lr){4-6}\cmidrule(lr){7-9}
V-Dem treatment & $\hlambda_n$ & $|\hbeta^*|\,/\,|\hbeta_0^{\mathrm{corr}}|$ & $\lambda^{\dagger}$ & size & verdict & $\hpsi$ & size$_{CR}$ & verdict$_{CR}$ \\
\midrule
Polyarchy (aggregate)   & $0.90$ & $0.061\,/\,0.068$ & $0.76$ & $5.6\%$ & point pass       & $19.2$ & $5.0\%$ & point pass \\
Legislative constraints & $0.55$ & $0.024\,/\,0.045$ & $0.62$ & $14\%$  & \textbf{flagged} & $24.0$ & $5.4\%$ & point pass\textsuperscript{$\S$} \\
Judicial constraints    & $0.41$ & $0.224\,/\,0.543$ & $0.93$ & $\ge0.99$ & \textbf{flagged} & $25.2$ & $\approx67\%$ & \textbf{flagged} \\
\bottomrule
\end{tabular}
\end{adjustbox}
\par\smallskip
\begin{minipage}{0.95\textwidth}\footnotesize\raggedright
\textsuperscript{$\S$} Point pass under the cluster-robust diagnostic, flagged under the i.i.d.\ one. Since inference in this literature clusters by country, the cluster-robust column is operative; the point pass is not a formal certificate. Proposition~\ref{prop-equicorr} independently shows that the i.i.d.\ flag is sensitive to persistent coder error.
\end{minipage}
\end{table}

\begin{figure}[htbp]\centering
\includegraphics[width=0.72\linewidth]{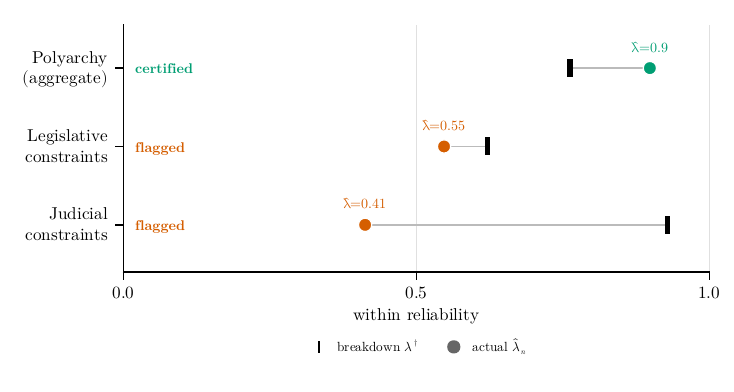}
\caption{The two-pole diagnostic on the V-Dem panel. For each treatment, the filled point is the actual within reliability $\hlambda_n$ and the vertical tick is the breakdown reliability $\lambda^{\dagger}$ (Definition~\ref{def-breakdown}, $\delta=0.05$). A specification obtains a \emph{point pass} iff $\hlambda_n\ge\lambda^{\dagger}$ --- a descriptive verdict, not a formal certificate: the aggregate polyarchy index clears its breakdown reliability, while the institutional-constraint sub-indices fall below theirs.}\label{fig-vdem}
\end{figure}

Three caveats accompany the table. \emph{First}, expert-coding error may be serially dependent within country, violating the i.i.d.\ assumption on $\nu$. The equicorrelated extension of Section~\ref{sec-ext} covers this, and its direction is one-sided: the marginal-variance diagnostic overstates the noise, so a point pass survives a fortiori under any within-country persistence, while a flag is conclusive only given a bound on the common-component share $\omega$. Proposition~\ref{prop-equicorr} makes the bound exact. Writing $\hlambda_n$ for the within reliability under the i.i.d.\ benchmark and $\lambda_\omega$ for the effective within reliability, $(1-\lambda_\omega)/\lambda_\omega=(1-\omega)(1-\hlambda_n)/\hlambda_n$, so a baseline flag is overturned when
\[
\omega\ge\omega^*:=1-\frac{\hlambda_n}{1-\hlambda_n}\frac{1-\lambda^\dagger}{\lambda^\dagger}.
\]
This gives $\omega^*\approx0.95$ for judicial constraints, so that flag stands unless coder error is almost entirely common within country, against $\omega^*\approx0.25$ for legislative constraints, which is one of the two reasons we lean on that row only lightly. \emph{Second}, the specification is static two-way FE; adding lagged output in the manner of \citet{ANRR} leaves the two-pole contrast intact but changes the magnitudes, since a lagged dependent variable introduces an incidental-parameter bias of different origin that we do not conflate with the measurement-error channel. \emph{Third}, the verdicts in Table~\ref{tab-vdem} are point passes and flags, but aggregate polyarchy also clears the \emph{formal certificate} of Proposition~\ref{prop-certificate}. Conditional on treating the model-implied reliability as supplied ($\ell_n=\hlambda_n$ and $\gamma_\lambda=0$), its certified breakdown reliability at $\gamma_\beta=0.05$ is $0.851$ under i.i.d.\ inference and $0.765$ under the operative country-clustered inference, both below $\hlambda_n=0.898$. Legislative constraints does not clear its cluster-certified threshold $0.741$, and judicial constraints does not clear $0.837$. Thus only polyarchy is certified, with conditional false-certification probability bounded by $5\%$. If uncertainty in the V-Dem reliability calibration is to be included, its lower confidence bound must replace $\hlambda_n$ and its error probability $\gamma_\lambda$ must be added; the claim is not an unconditional validation of the measurement model. Design~6 directly verifies the certificate's false-certification control and shows why a noisy reliability point estimate cannot be silently treated as known. The polyarchy verdict is unchanged when its implied reported $t$ is perturbed by $\pm20\%$ (\texttt{verification/certificate\_calc.py}), so it does not rest on rounding.

\subsection{Repeated schooling reports in identical twins}\label{sec-app-twins}

The second application uses the repeated-report design of \citet{AshenfelterKrueger1994}. Each identical twin reports both siblings' schooling, so the within-pair schooling contrast is measured twice. Let $Y_i$ be the difference in log hourly wages, $X_i^*$ the difference in self-reported schooling, $Z_i^*$ the corresponding difference constructed from the co-twin reports, and $W_i$ an intercept plus the within-pair differences in tenure, marital status, and union coverage. Differencing removes the pair effect before estimation; the $147$ complete twin pairs are the independent sampling units. This is therefore the clean i.i.d.\ branch of the theory, with no serial-clustering rescaling, projection-compatibility condition, or $\psi$.

The public extract reproduces the controlled first-difference regression: the coefficient on $X_i^*$ is $\hbeta^*=0.09088$, its conventional standard error is $0.02198$, and $|t|=4.134$. The self-consistent breakdown reliability is consequently
\[
\lambda^{\dagger}=\frac{4.134}{4.134+0.652358}=0.8637.
\]
For either within-reliability estimate to correspond to this regression, both reports must use the same projection and complete-case sample. Writing $\widetilde X=M_WX^*$ and $\widetilde Z=M_WZ^*$, mutually uncorrelated classical report errors identify
\[
\widehat\lambda_{\mathrm{cov}}
=\frac{\widetilde X'\widetilde Z}{\widetilde X'\widetilde X}=0.5749.
\]
Under the additional restriction that the two reporting errors have equal variances, the difference estimator gives
\[
\widehat\lambda_{\mathrm{eq}}
=1-\frac{(\widetilde X-\widetilde Z)'(\widetilde X-\widetilde Z)}{2\widetilde X'\widetilde X}=0.5514.
\]
Thus the independent-error within-reliability band $[0.551,0.575]$ lies far below the breakdown reliability. It implies $|\widehat\eta|=3.06$--$3.36$ and exact-normal size $86.4\%$--$92.0\%$: the conventional nominal-$95\%$ interval has only $8.0\%$--$13.6\%$ coverage at the true return. This does \emph{not} alter the test of a zero return, since classical error attenuates a zero to a zero. Instead, it shows that the uncorrected interval provides almost no coverage for the magnitude used in calculations of returns to schooling.

Design~4 validates the finite-sample mapping through $|\eta|=3.58$, beyond this application's range (Figure~\ref{fig-design4}).

\begin{table}[htbp]\centering
\caption{Repeated-report diagnostic in the twins design. The first two rows use the controlled $147$-pair public extract from the Ashenfelter--Krueger design ($|t|=4.134$, $\lambda^{\dagger}=0.864$); both schooling reports are residualized on the identical controls and sample. The last row is a sensitivity based on \citet{Rouse1999}'s controlled correlated-report specification ($0.071/0.016$, $445$ complete pairs) and her random-order mean within reliability $0.748$. ``Size'' is the exact-normal implied size of the nominal-$5\%$ test of the true nonzero coefficient; ``coverage'' is one minus that size. The Rouse row uses rounded published inputs and a different report construction designed for person-specific error.}\label{tab-twins}
\begin{adjustbox}{max width=\textwidth}
\begin{tabular}{lcccccc}
\toprule
sample and report-error model & $\hlambda$ & $\lambda^{\dagger}$ & $\hbeta^*/\hlambda$ & $|\widehat\eta|$ & size & coverage \\
\midrule
AK, equal error variances & $0.551$ & $0.864$ & $0.165$ & $3.36$ & $92.0\%$ & $8.0\%$ \\
AK, covariance repeat     & $0.575$ & $0.864$ & $0.158$ & $3.06$ & $86.4\%$ & $13.6\%$ \\
Rouse, correlated-report model & $0.748$ & $0.872$ & $0.095$ & $1.49$ & $32.1\%$ & $67.9\%$ \\
\bottomrule
\end{tabular}
\end{adjustbox}
\end{table}

The coefficient agreement is not independent validation. With the covariance-based within reliability, correcting the estimate is algebraically identical to forming a just-identified IV estimator from the same two reports:
\[
\frac{\hbeta^{\mathrm{OLS}}_X}{\widehat\lambda_{\mathrm{cov}}}
=\frac{(\widetilde X'Y)/(\widetilde X'\widetilde X)}{(\widetilde X'\widetilde Z)/(\widetilde X'\widetilde X)}
=\frac{\widetilde X'Y}{\widetilde X'\widetilde Z}.
\]
In Ashenfelter and Krueger's original no-control moment system, the two directional IV estimates are $0.157$ and $0.167$, and the restricted estimate is $0.162$. Their proximity to the within-reliability-corrected range $0.158$--$0.165$ is therefore largely mechanical. What the original correction did not report, and what Table~\ref{tab-twins} adds, is the implied coverage of the \emph{uncorrected} interval.

The independence qualification matters. \citet{Rouse1999} reproduces the original basic first-difference estimates, $0.092$ ($0.024$) before correction and $0.167$ ($0.043$) after IV correction, for $149$ pairs. In a larger $453$-pair sample she obtains $0.075$ ($0.017$) and $0.095$ ($0.027$), and formally rejects the independent-report model. Her controlled specification for person-specific reporting error instead constructs each schooling contrast from one sibling's reports of both twins. It yields $0.071$ ($0.016$) on $445$ complete pairs, and her random-order calculation gives within reliability $0.748$ when reporting correlation is allowed. The last row of Table~\ref{tab-twins} applies the diagnostic to those corresponding inputs. Because $0.748<0.872$, the specification remains flagged, but the implied size falls from $86$--$92\%$ under independent reports to $32\%$ in this sensitivity. Accordingly, this application contributes the coverage verdict and its sensitivity, not an independent recovery of a published corrected coefficient. \ref{sec:supp-twins} gives the moment audit and exact reproduction details.

\section{Extensions}\label{sec-ext}

\paragraph{Vector treatment.} With $X \in \R^p$, the construction of Theorem~\ref{thm-noncentral} carries over: the Wald statistic for $H_0:\beta=\beta_0$ is asymptotically non-central $\chi^2_p$, with a lower bound on $\lambda_{\min}(nQ_{K_n})$ sufficing for leading-order size control, mirroring \citet[\S 6]{SY}. We record this as the natural generalization. A sharp multivariate eigenvalue threshold requires a separate analysis and is not developed here; the vector-treatment subsection of \ref{sec:supp-ext} records the matrix notation and the boundary of the present claim.

\paragraph{Heteroskedasticity.} Under conditional heteroskedasticity a many-covariate-robust variance estimator replaces the homoskedastic RSS scale. The non-centrality retains its form with $\sigma$ replaced by a design-dependent limit $\omega_*$, coinciding with $\sigma$ under homoskedasticity, and the ``no $\tau^2$ threshold'' conclusion is unchanged. Candidate estimators include the leave-out approaches of \citet{KSS} and \citet{Jochmans2022}, the CJN Hadamard correction, and the almost-unbiased refinement of \citet{Anatolyev2018}; their existence and regularity conditions must be checked before their reported $t$ is used. The breakdown formula itself is estimator-agnostic conditional on that $t$. Recalculation with the CJN and Anatolyev corrections changes no V-Dem or twins verdict; \ref{sec:supp-anatolyev} reports the numbers and the exact implementation. For cluster dependence --- serial correlation of $u$ within unit, the empirically dominant departure in panels --- \citet{AnatolyevNg2026}'s leave-cluster-out estimator is the general alternative, while the corresponding rescaling for the ordinary Arellano estimator under our projection-compatibility conditions is Remark~\ref{rem-cluster}. The exact form of $\omega_*$ is given in \ref{sec:supp-ext}.

\paragraph{Serially dependent measurement error.} The classical assumption that $\nu$ is i.i.d.\ is violated when expert-coding error persists within country (Section~\ref{sec-app-vdem}) or survey error persists within person. The noise covariance enters the mean bias only through $\E[\nu'M_{K_n}\nu]=\tr(M_{K_n}\Sigma_\nu)$, so the bias term is $-\beta_0\lim\tr(M_{K_n}\Sigma_{\nu,n})$ for a general $\Sigma_\nu$. The trace alone does not, however, deliver the quadratic-form concentration or the score CLT that the non-centrality theorem also needs; those require additional structure on $\Sigma_\nu$ beyond its trace. The within-unit-equicorrelated case (the empirically relevant one for persistent coding error) admits such structure, a one-factor model, and has a clean, favorable answer.

\begin{proposition}[Equicorrelated measurement error]\label{prop-equicorr}
Under the remaining conditions of Theorem~\ref{thm-noncentral}, suppose the unit dummies are among the fixed effects and the measurement error follows the one-factor model
\[
\nu_{it}=\sigma_{\nu,n}\big(\sqrt{\omega}\,g_i+\sqrt{1-\omega}\,\epsilon_{it}\big),\qquad \omega\in[0,1),\ \ \sigma_{\nu,n}^2=c^2/n,
\]
with $\{g_i\}$ and $\{\epsilon_{it}\}$ mutually independent, mean zero, unit variance, each satisfying Assumption~\ref{ass-nu}(ii) (so $\{\epsilon_{it}\}$ is i.i.d.\ with $r=2$), and both independent of $(X,u)$ given $\G_{K_n}$. This induces exactly the within-unit-equicorrelated covariance $\Sigma_{\nu,n}=\sigma_{\nu,n}^2[(1-\omega)I_n+\omega B]$, $B_{(it),(js)}=\mathbf 1\{i=j\}$. Then Lemma~\ref{lem-atten}, Lemma~\ref{lem-lev-star}, and Theorem~\ref{thm-noncentral} hold verbatim with $c^2(1-\rho)$ replaced by $(1-\omega)\,c^2(1-\rho)$, giving
\[
\eta_\omega=-\frac{\beta_0\,(1-\omega)\,c^2\sqrt{1-\rho}}{\sigma\sqrt{\tau^2+(1-\omega)\,c^2}},\qquad |\eta_\omega|\le|\eta|\ \text{and non-increasing in }\omega.
\]
\end{proposition}

\begin{proof}
Deferred to the Online Supplement. Since the common-factor vector lies in the unit-dummy column space, $M_{K_n}\nu=\sigma_{\nu,n}\sqrt{1-\omega}\,M_{K_n}\epsilon$ identically. Thus every occurrence of the measurement error after projection is the i.i.d.\ case with $c^2$ replaced by $(1-\omega)c^2$. The stated formula follows from Theorem~\ref{thm-noncentral}; its magnitude decreases in $\omega$ because $q/\sqrt{\tau^2+q}$ increases in $q\ge0$ while $q=(1-\omega)c^2$ decreases in $\omega$.
\end{proof}

The unit fixed effect annihilates the persistent (common) component of the noise, so persistent coding error produces a smaller EIV distortion than the classical i.i.d.\ benchmark. Two consequences for the lead application. First, the diagnostic computed from the marginal posterior variances (which ignore within-country correlation) is conservative: it uses $a=c^2(1-\rho)$ in place of the true $(1-\omega)a\le a$, so it overstates the noise and issues too few passes. A specification it passes (polyarchy) passes a fortiori under any within-country error persistence. Second, the pass direction is the one that is robust: because $|\eta_\omega|$ is decreasing in $\omega$ with $|\eta_\omega|\to 0$ as $\omega\to 1$, a flag computed under the i.i.d.\ benchmark is not immune to serial correlation, while a pass survives a fortiori. I.i.d.-based failure, unlike an i.i.d.-based pass, is therefore conclusive only given a bound on $\omega$; the exact flag-reversal threshold, and its two very different values for the lead application's judicial- and legislative-constraints rows, are given where they are used, in Section~\ref{sec-app-vdem}. This is the same mechanism, and the same $M\mathbf 1_i=0$ identity, that motivates the cluster-robust rescaling of Remark~\ref{rem-cluster}; note that here it is a proved consequence of Proposition~\ref{prop-equicorr}, whereas there it is not.

\paragraph{Other bias sources and nonclassical error.} The derivation pattern (introduce a local bias source, compute the non-centrality of the $t$-statistic under the $(\rho, \tau^2)$ drift, invert the leading size distortion) applies beyond classical EIV: heterogeneous treatment effects in TWFE designs \citep{dCdH}, FE misspecification, omitted nonlinearities, and binary misclassification. Each candidate produces its own threshold structure; developing them is left to future work. An especially useful intermediate case is weakly classical error, $\E[\nu\mid X]=0$, which permits the noise variance to vary with the latent regressor. In nonparametric regression, \citet{EvdokimovZeleneev2024} show that the bias then depends on the skedastic function $v(x)=\Var(\nu\mid X=x)$ and its derivative, identify $v(\cdot)$ with a possibly discrete instrument, and extend the construction to correlated nonclassical error. For the present linear FE problem, the concrete open question is whether those features collapse after projection to a weighted noise trace, or instead alter the non-centrality and the breakdown fixed point; their discrete-IV and repeated-measure routes also suggest ways to supply the richer pilot. This is not merely a technical extension. In linked survey--administrative data, \citet{CelhayMeyerMittag2024} find false-negative rates as high as $59\%$ and show that recall, telescoping, salience, stigma, and general survey cooperation predict reporting errors; dependence on respondent characteristics or the true value is exactly what invalidates a classical-error correction. \citet{MeyerMittag2017} further show in binary-choice applications that imposing an incorrect orthogonality restriction can move every estimator they examine farther from the validated target. The practical implication is asymmetric: the classical diagnostic is useful only after its error model is defended, and a flag is a reason to seek validation data or a model of the reporting process, not an instruction to divide mechanically by a measurement-reliability estimate. The binary case remains subject to the scope warning of Section~\ref{sec-scope}: this paper's diagnostic must not be applied to a binary mismeasured treatment.

\section{Discussion}\label{sec-conc}

The weak-identification analogy between fixed-effect saturation and weak instruments is genuine but requires care. Under strict exogeneity and homoskedastic errors, $\tau^2 = nQ_K$ does not behave like an IV concentration parameter --- the FE-OLS estimator is unbiased, and no Stock--Yogo threshold arises in the baseline model. Once a bias source is introduced (classical measurement error in the treatment, in our development), the analogy is restored, with a closed-form critical value~\eqref{eq-cv}, a feasible one-line form~\eqref{eq-feasible}, a corrected plug-in protocol, and a power reading: the within reliability $\lambda$ is simultaneously the size-distortion driver and the local-power tax, while the breakdown reliability $\lambda^{\dagger}$ is the specification-specific requirement to report beside it. Under Condition~(B), the FE dimension $\rho$ enters only through an overall $\sqrt{1-\rho}$ asymptotic scaling common to signal and noise, not through the within reliability $\lambda$ itself (Remark~\ref{rem-rho}); this is a modest role for saturation, and the feasible protocol never decomposes the diagnostic into $(\rho,\tau^2,c^2)$ in the first place.

Three structural features distinguish the threshold from the IV case. The threshold combines a signal-to-noise ratio $\tau^2/c^2$, exactly as classical attenuation does, with a saturation-driven scale factor $\sqrt{1-\rho}$ that raises the standard error without touching the within reliability. The primitive $\tau^2$ threshold is data-dependent through $|\beta_0|$, $\sigma$, and the noise scale, so it cannot be summarized in a universal table the way Stock--Yogo IV thresholds can. The breakdown reliability can be tabulated against the reported $t$-statistic (Table~\ref{tab-breakdown-lookup}), while~\eqref{eq-feasible} reduces the diagnostic to one inequality in reported quantities. Finally, the threshold requires an external noise input; Remark~\ref{rem-est} identifies three practical sources, including expert-coded posteriors, validated survey measures, and administrative matches.

For specifications where the threshold is not met, two remedies exist: the estimation route of \citet{GH}, and a fixed-effect Anderson--Rubin test with uniform validity over $\tau^2 \in (0, \infty]$ --- which we develop in companion work in progress --- giving honest confidence sets even when identification is fragile.

Three limitations should travel with the results. First, the diagnostic assumes \emph{classical} error in a continuous regressor; it does not apply to binary-treatment misclassification, where the error is nonclassical by construction, and applying it there is a misuse (Section~\ref{sec-scope}). The linked-data evidence of \citet{CelhayMeyerMittag2024} makes this a substantive screening step rather than boilerplate, and \citet{MeyerMittag2017} show why correcting under the wrong dependence model can worsen bias. Second, the paper reports two distinct objects and they should not be conflated: a \emph{point pass} is the descriptive verdict that the plug-in diagnostic clears the breakdown reliability, while a \emph{certificate} survives coefficient uncertainty and, when needed, reliability uncertainty (Definition~\ref{def-breakdown}, protocol step~6). Proposition~\ref{prop-certificate} supplies the guarantee: false certification is bounded by $\gamma_\beta+\gamma_\lambda$, separately from the size tolerance $\delta$ (Remark~\ref{rem-gamma-delta}); Design~6 verifies that bound and shows that substituting a noisy reliability point estimate as though known can violate it. Conditional on the supplied V-Dem measurement model, aggregate polyarchy is certified with $\gamma_\beta=0.05$ and $\gamma_\lambda=0$; no other empirical row is reported as certified. Third, the cluster-robust theory of Section~\ref{sec-cluster} requires many clusters and an asymptotic projection-compatibility condition. The direct V-Dem sample ratios $0.0069$--$0.0073$ are reassuring, but a single finite-panel calculation cannot establish the required sequence. The twins application avoids that issue by sampling independent pairs; neither application covers the few-clusters regime. Extending the derivation pattern to nonclassical misclassification, and to few-cluster asymptotics, are the two developments that would most enlarge the diagnostic's reach.

The resulting reporting rule is compact: place the data's within reliability $\lambda$ beside the specification's breakdown reliability $\lambda^{\dagger}$. For a size-controlled claim, replace the latter by $\lambda^{\dagger}_{\gamma_\beta}$ and place beside it a lower confidence bound for $\lambda$. Breakdown reliability is therefore the paper's portable object---the minimum within reliability required for conventional inference in a saturated fixed-effect specification to remain size-controlled---and its certified form makes explicit exactly which uncertainty the conclusion has survived.

\section*{Declaration of competing interest}

The author declares that he has no known competing financial interests or personal relationships that could have appeared to influence the work reported in this paper.

\section*{Funding}

This research did not receive any specific grant from funding agencies in the public, commercial, or not-for-profit sectors.

\section*{Declaration of generative AI and AI-assisted technologies in the manuscript preparation process}

During the preparation of this work the author used an AI proof-checking tool based on Opus 5 (Anthropic) and GPT-5.5 (OpenAI) to perform supplementary consistency checks on the mathematical arguments. The author verified all content, reviewed and edited it as needed, and takes full responsibility for the content of the published article.

\section*{Data availability}

All data used in the applications are publicly available: the V-Dem indices and their measurement-model posterior standard deviations from the V-Dem Country-Year dataset \citep{VDem16}, GDP per capita from the Maddison Project Database 2020 \citep{Maddison2020}, and the public teaching extract of the Ashenfelter--Krueger repeated-report twins data distributed with \texttt{RbyExample} \citep{RbyExampleData}. \ifanon\else The diagnostics are implemented in the \textsf{PanelAdequacy} software \citep{PanelAdequacySoftware}, available as a \textsf{Julia} package (\texttt{PanelAdequacy.jl}) and an \textsf{R} package (\texttt{panelcert}); the V-Dem analysis panel is bundled with both packages. \fi A replication archive containing the twins extract and a separate script that reproduces every value in Table~\ref{tab-twins}, together with the V-Dem construction, simulations, figures, and the CJN--Anatolyev variance-estimator audit in \ref{sec:supp-anatolyev}, accompanies the journal submission.

\bibliographystyle{elsarticle-harv}
\bibliography{bibliography.bib}

\clearpage
\appendix
\begin{center}\Large\bfseries Online Supplement\end{center}
\medskip
\section{Repeated-report twins application: moment audit and sensitivity}
\label{sec:supp-twins}

This section supplies the exact calculations behind Section~\ref{sec-app-twins}
of the main text.  The design of \citet{AshenfelterKrueger1994} supplies two
measurements of the within-pair schooling difference: $X^*$, constructed from
each twin's self-report, and $Z^*$, constructed from the corresponding co-twin reports.
The outcome $Y$ is the within-pair difference in log hourly wages.  We use the
$147$ complete observations in the public \texttt{RbyExample} extract
\citep{RbyExampleData} and reproduce their controlled first-difference
specification with
\[
W=(1,\ \Delta\text{tenure},\ \Delta\text{married},\
\ \Delta\text{union coverage}).
\]
All three variables are put through the same projection: $\widetilde
X=M_WX^*$, $\widetilde Z=M_WZ^*$, and $\widetilde Y=M_WY$.  This common sample
and projection are essential; combining a controlled $t$-statistic with an
unprojected repeat reliability would not implement the diagnostic for one
specification.

The residual cross-products are
\begin{center}
\begin{tabular}{lr@{\qquad}lr@{\qquad}lr}
\toprule
$\widetilde X'\widetilde X$ & $537.295921$ &
$\widetilde Z'\widetilde Z$ & $562.534538$ &
$\widetilde X'\widetilde Z$ & $308.893401$ \\
$\widetilde X'\widetilde Y$ & $48.827247$ &
$\widetilde Z'\widetilde Y$ & $55.187409$ & & \\
\bottomrule
\end{tabular}
\end{center}
They give
\[
\begin{aligned}
\hbeta_X^{\mathrm{OLS}}&=0.0908759,
&\widehat{\mathrm{se}}(\hbeta_X^{\mathrm{OLS}})&=0.0219815,\\
|t|&=4.134199,
&\lambda^\dagger&=0.863710.
\end{aligned}
\]
The coefficient and standard error reproduce the published controlled values
$0.091$ and $0.022$.  The public extract also reproduces the reported-direction
IV estimate:
\[
\hbeta_{X\leftarrow Z}^{\mathrm{IV}}
=\frac{\widetilde Z'\widetilde Y}{\widetilde Z'\widetilde X}
=0.178662,
\]
against $0.179$ in the original table.

Under mutually uncorrelated classical report errors, no equality of their
variances is needed for the covariance repeat:
\[
\widehat\lambda_{\mathrm{cov}}
=\frac{\widetilde X'\widetilde Z}{\widetilde X'\widetilde X}
=0.574904.
\]
Under the additional restriction that the two error variances are equal,
\[
\widehat\lambda_{\mathrm{eq}}
=1-\frac{(\widetilde X-\widetilde Z)'(\widetilde X-\widetilde Z)}
{2\widetilde X'\widetilde X}
=0.551417.
\]
Applying $|\widehat\eta(\lambda)|=|t|(1-\lambda)/\lambda$ and the exact-normal
size map gives the full-precision audit:
\begin{center}
\small
\begin{tabular}{lrrrrr}
\toprule
pilot & $\hlambda$ & $\hbeta^*/\hlambda$ & $|\widehat\eta|$ & size & coverage \\
\midrule
equal error variances & $0.551417$ & $0.164804$ & $3.363211$ & $0.919728$ & $0.080272$ \\
covariance repeat     & $0.574904$ & $0.158072$ & $3.056917$ & $0.863669$ & $0.136331$ \\
\bottomrule
\end{tabular}
\end{center}

The coefficient agreement is not independent validation.  Using the covariance
repeat, the within-reliability correction is algebraically identical to the
reverse-direction IV estimator:
\[
\frac{\hbeta_X^{\mathrm{OLS}}}{\widehat\lambda_{\mathrm{cov}}}
=\frac{(\widetilde X'\widetilde Y)/(\widetilde X'\widetilde X)}
{(\widetilde X'\widetilde Z)/(\widetilde X'\widetilde X)}
=\frac{\widetilde X'\widetilde Y}{\widetilde X'\widetilde Z}
=0.158072.
\]
In the original no-control moment system, \citet{AshenfelterKrueger1994} report the two
directional IV estimates $0.157$ and $0.167$ and a restricted estimate $0.162$.
The within-reliability correction and their IV correction recycle the same repeat
covariance, so coefficient agreement is arithmetic corroboration.  The new
object in the main text is the implied rejection probability, equivalently the
coverage of the uncorrected interval at the true nonzero coefficient.

\citet{Rouse1999} supplies two checks.  First, her extractable Table~2
reproduces the original basic first-difference estimates $0.092$ ($0.024$) and
$0.167$ ($0.043$), with $149$ pairs; her expanded $1991$--$1993$ and $1995$
sample gives $0.075$ ($0.017$) and $0.095$ ($0.027$), with $453$ pairs.  Second,
she rejects the restriction that the two reporting errors are independent.
Her correlated-error construction uses the difference in both schooling
levels as reported by one sibling and instruments it with the corresponding
difference reported by the other.  In the controlled $445$-pair sample the
uncorrected coefficient is $0.071$ ($0.016$), and her random-order exercise
reports mean report reliability $0.748$ when correlated reporting error is allowed.
These corresponding inputs give
\[
\begin{aligned}
|t|&=4.4375, & \lambda^\dagger&=0.871832,
& |\widehat\eta|&=1.494987,\\
\text{size}&=0.321249, & \text{coverage}&=0.678751.&&
\end{aligned}
\]
Because $0.748<0.871832$, this alternative specification remains flagged, but
its implied size is $32.1\%$, well below the $86$--$92\%$ range under independent
errors in the original report construction.  Rouse's result is therefore
reported as a separate sensitivity rather than folded into the independent-error
within-reliability band.

Finally, first differencing removes the shared twin-pair effect before estimation.
With the twin pair as the sampling unit, the resulting cross-section uses the
i.i.d.\ branch of the theory.  Cluster rescaling, projection compatibility, and
the variance factor $\psi$ therefore play no role here; V-Dem remains the
application of the cluster-robust branch.

\section{Extensions and variance-estimator sensitivity}
\label{sec:supp-ext}

This section develops the vector-treatment and heteroskedastic extensions of
Section~\ref{sec-ext} and examines the many-covariate variance estimator used
to construct the reported $t$.

\subsection{Vector treatment}

With $X \in \R^p$, $Q_K = \E[\Var(X \mid \G_K)]$ is a $p \times p$ matrix and
$\Sigma_\nu$ the noise covariance under the drift $\Sigma_{\nu, n} = C/n$. The
construction of Theorem~\ref{thm-noncentral} carries over: the Wald statistic
for $H_0:\beta=\beta_0$ is asymptotically non-central $\chi^2_p$ with
non-centrality $\eta'\eta$, where the vector non-centrality is governed by
the matrix analogue of $\tau^{*2}=(1-\rho)(\tau^2+c^2)$, and a sufficient
condition for leading-order size control at tolerance $\delta$ is a lower
bound on $\lambda_{\min}(nQ_{K_n})$, mirroring \citet[\S 6]{SY}. We record
this as the natural generalization and do not develop it here: the scalar
diagnostic is complete on its own, and a sharp multivariate eigenvalue
threshold is a separate exercise whose details we omit.

\subsection{Heteroskedasticity}

Under conditional heteroskedasticity the CJN homoskedastic variance
estimator is replaced by a leave-one-out (HC2) estimator in the manner of
\citet{KSS} and \citet{Jochmans2022}. The bias term $-\beta_0\,c^2(1-\rho)$ is
unchanged, since it arises entirely from $\E[\nu'M_{K_n}\nu]$, which depends
on the measurement-error variance alone and not on the conditional variance
structure of $u$; heteroskedasticity changes only the standardization of the
score. That standardization, however, is not in general the
true-regressor limit $\omega$ of Assumption~\ref{ass-reg}(iv): the
contaminated score $X^{*\prime}M_{K_n}u$ weights the heterogeneous $\sigma_i^2$
through the contaminated leverage $\widetilde X^{*2}_{K_n,i}$ rather
than $\tX_{K_n,i}^2$, so its variance has its own limit $\omega_*^2 = \lim
(X^{*\prime}M_{K_n}X^*)^{-1}\sum_i \widetilde X^{*2}_{K_n,i}\sigma_i^2$. This
normalized limit can itself depend on the saturated design; no general
$\rho$-free conclusion follows under heteroskedasticity. The non-centrality
retains its form with $\omega_*$ in place of $\sigma$, and $\omega_*$ is what
the HC2 estimator consistently targets; it coincides with $\sigma$ under
homoskedasticity. For cluster dependence --- serial correlation of $u$
within unit, the empirically dominant departure in panels --- the
corresponding rescaling of the diagnostic is Remark~\ref{rem-cluster}, and
the applications report its verdicts alongside the i.i.d.\ ones.

\subsection{Many-covariate variance-estimator sensitivity}
\label{sec:supp-anatolyev}

The homoskedastic statistic in the main text uses the scalar RSS correction
$\widehat\sigma^2_{\mathrm{CJN}}=\widehat u'\widehat u/(n-d_K-1)$.
That object should not be confused with the heteroskedastic CJN covariance
estimator.  This subsection asks how the reported $t$, and therefore
$\lambda^\dagger$, changes under the latter estimator and the almost-unbiased
adjustment of \citet{Anatolyev2018}.

Let $W$ collect the nuisance covariates, let
$M=I-W(W'W)^{-1}W'$, put $v=MX^*$ and
$\widehat P=v(v'v)^{-1}v'$, and let
$\widehat u=(M-\widehat P)Y$ be the residual from the full regression.  For a
scalar coefficient, define
\[
 \kappa_{\mathrm{CJN}}=(M\odot M)^{-1},\qquad
 \kappa_{\mathrm{AU}}=
       (M\odot M-\widehat P\odot\widehat P)^{-1},
\]
where $\odot$ denotes the Hadamard product.  If
$\widehat\omega_{q,i}=\sum_j\kappa_{q,ij}\widehat u_j^2$, then
\[
 \widehat{\mathrm{se}}_q(\widehat\beta)
 =\frac{\{\sum_i v_i^2\widehat\omega_{q,i}\}^{1/2}}{v'v},
 \qquad q\in\{\mathrm{CJN},\mathrm{AU}\}.
\]
The AU correction accounts for coefficient-estimation noise that the CJN
leading-term correction omits and is exactly unbiased under homoskedasticity.
Its simulations are especially favorable with skewed regressors
\citep{Anatolyev2018}.

\begin{table}[htbp]\centering
\caption{Breakdown-reliability sensitivity to the many-covariate variance estimator}
\label{tab:supp-anatolyev}
\scriptsize
\begin{adjustbox}{max width=\textwidth}
\begin{tabular}{llrrrrrc}
\toprule
Application & specification & $\widehat\lambda$ &
$\widehat{\mathrm{se}}_{\rm CJN}$ & $\widehat{\mathrm{se}}_{\rm AU}$ &
$\lambda^\dagger_{\rm CJN}$ & $\lambda^\dagger_{\rm AU}$ & verdict (both) \\
\midrule
V-Dem & polyarchy, static & .8984 & .02940 & .02941 & .7607 & .7606 & point pass \\
      & polyarchy, $+2$ lags & .8964 & .005696 & .005697 & .3744 & .3743 & point pass \\
      & legislative constraints, static & .5472 & .02424 & .02424 & .6072 & .6071 & flag \\
      & legislative constraints, $+2$ lags & .5392 & .005079 & .005081 & .7110 & .7109 & flag \\
      & judicial constraints, static & .4125 & .02741 & .02742 & .9261 & .9261 & flag \\
      & judicial constraints, $+2$ lags & .4044 & .004833 & .004834 & .4043 & .4042 & point pass \\
Twins & AK controlled first difference & .5749 & .02867 & .02937 & .8293 & .8259 & flag \\
\bottomrule
\end{tabular}
\end{adjustbox}
\begin{minipage}{0.96\textwidth}\footnotesize
\emph{Notes:} $\eta^\dagger=0.652358$.  The V-Dem rows are an i.i.d.\
heteroskedastic sensitivity audit; the operative clustered verdicts in the
main text use the reported CRVE.  The twins row uses the covariance-repeat
within reliability, the most favorable of the two independent-report pilots in
Table~\ref{tab-twins}.  ``Point pass'' is descriptive and is not a formal
certificate.  Unrounded values are in the replication output.
\end{minipage}
\end{table}

Both inverses exist in every row.  Across all eight rows, the smallest CJN
sufficient-condition margin $\min_i(2M_{ii}-1)$ is $0.429$, and the smallest
AU margin $\min_i\{M_{ii}(2M_{ii}-1)-\widehat P_{ii}\}$ is $0.246$, both in
Grunfeld.  Restricting attention to the seven breakdown-reliability
specifications (six V-Dem rows and twins), the corresponding minima are $0.815$
and $0.726$.
No verdict changes.  The closest comparison is lagged judicial constraints:
its unrounded within reliability exceeds the CJN threshold by $1.10\times10^{-4}$
and the AU threshold by $1.79\times10^{-4}$, so the table discloses rather
than conceals that numerical fragility.  The exact sparse implementation and
all diagnostics are in
\texttt{code/anatolyev\_2018\_robustness.jl}, which writes
\texttt{results/anatolyev\_2018\_robustness.csv}.  For clustered errors and
more general many-covariate designs, the leave-cluster-out estimator of
\citet{AnatolyevNg2026} is the relevant alternative; replacing the variance
estimator changes the value obtained from the reported $t$, not the algebraic
breakdown formula.

\section{Design 5: cluster-robust simulation validation}
\label{sec:supp-design5}

Design~5 evaluates the cluster theory of Section~\ref{sec-cluster} on a
unit-and-time FE design with $G$
units (which are also the clusters), $T$ periods, within-unit AR(1) errors,
and clustering by unit, so that unit effects are nested in clusters and the
$T-1$ time effects are not. Seven predictions are checked; the scripts are
released with the paper (\texttt{verify\_cluster\_clt.py} and
\texttt{verify\_cluster\_high\_eta.jl}), and all reported statistics are from
$2{,}500$--$4{,}000$ replications.

\emph{(i) The limit itself.} At $G=200$, $T=6$, $\rho_n=0.167$, the empirical
mean of $T^{CR}_n$ is $-1.25$ against the predicted $\eta/\sqrt{\psi}=-1.21$,
with empirical standard deviation $1.03$; adding the time effects
($d^{\mathrm{ne}}_n/G_n=0.025$) gives $-1.22$ against $-1.21$ with standard
deviation $1.02$. Tested in its exact conditional form --- $T^{CR}_n$
centred at $-\beta_0(X^{*\prime}M\nu)/\sqrt{\Psi_n}$ --- the match is within
Monte Carlo error at every configuration we ran, including strongly
persistent designs.

\emph{(ii) CRVE consistency, and its rate.} The ratio
$\widehat V^{\mathrm{sc}}_{CR}/\Psi_n$ of Lemma~\ref{lem-crve} is $0.981$ at
$G=50$, $0.994$ at $G=200$ and $0.998$ at $G=800$: the bias is $O(1/G_n)$
and vanishes, as claimed.

\emph{(iii) Projection compatibility is necessary, not decorative.} Pushing
the non-nested block to $d^{\mathrm{ne}}_n/G_n=3.25$ ($G=12$, $T=40$) breaks
the conclusion exactly as Lemma~\ref{lem-nest}(b) predicts: $\widehat
V^{\mathrm{sc}}_{CR}/\Psi_n$ falls to $0.82$ and the standard deviation of
$T^{CR}_n$ rises to $1.28$, so the cluster-robust test over-rejects. This is
the condition required by the clustered V-Dem application of
Section~\ref{sec-app-vdem}; the twins application samples independent pairs.

\emph{(iv) The degrees-of-freedom trap.} Applying the conventional
$\tfrac{n-1}{n-K_n}$ small-sample factor at $\rho_n=0.167$ inflates $\widehat
V^{\mathrm{sc}}_{CR}/\Psi_n$ to $1.203$, against the predicted
over-correction factor $1/(1-\rho_n)=1.200$ --- agreement to three digits.
The factor must be omitted; it is not a harmless conservatism, and at the
saturation levels common in worker--firm designs it would be substantial.

\emph{(v) The $\sigma$-cancellation of Corollary~\ref{cor-cluster-feasible}.}
Under cluster dependence $\hsigma^*_{\mathrm{CJN}}$ converges to
$\varsigma=0.80\ne\sigma=1$, and $\hpsi$ converges to $\psi\sigma^2/\varsigma^2$
rather than to $\psi$; both ingredients of the feasible recipe are
individually inconsistent. Their ratio is not: the recipe returns $-1.22$
against the target $-1.21$. This is the step that licenses computing the
diagnostic from i.i.d.\ regression output and dividing by the reported
inflation factor.

\emph{(vi) The sign of $\psi-1$.} Three configurations at $G=200$ ---
$(\varrho_x,\varrho_u,T) = (0,\,0.6,\,10)$, $(0.9,\,0.9,\,20)$ and
$(0.98,\,0.95,\,40)$ --- give $\psi = 0.75$, $1.27$ and $1.72$: clustering
amplifies the measured distortion in the serially independent design and
deflates it in the persistent ones, as Remark~\ref{rem-psi-sign} describes.
These are three points in the design space, not a ceteris-paribus
comparative static in $\varrho_x$: the error persistence and the panel
length move with the treatment persistence across the three, and $\psi$ is
itself increasing in $T$ at fixed $(\varrho_x,\varrho_u)$. The crossing of
one is therefore a joint property of persistence and panel length, not a
function of $\varrho_x$ alone.

\emph{(vii) High-$|\eta_{CR}|$ size mapping.} A separate $4{,}000$-replication
cell uses $G=200$, $T=10$, $\rho_n=0.1045$, $\tau^2=5$, $c^2=15$, and
$\varrho_u=0.6$, giving $\lambda=0.25$ and $\psi=0.749$. The predicted
cluster non-centrality is $\eta/\sqrt\psi=-3.667$, while the empirical mean
of $T_n^{CR}$ is $-3.709$ (Monte Carlo standard error $0.016$) and its standard
deviation is $1.027$. Empirical size is $0.959$ against the exact-normal
prediction $0.956$, and $\widehat V^{\mathrm{sc}}_{CR}/\Psi_n=0.986$. This
extends the clustered validation beyond the twins range
$|\widehat\eta|=3.06$--$3.36$ as well as the V-Dem cluster-rescaled value
$|\widehat\eta_{CR}|=2.4$.

\section{Replication note: the CRVE small-sample factor}
\label{sec:supp-footnote}

This note gives the exact magnitudes behind the warning in the footnote
accompanying Section~\ref{sec-app-vdem}. The adequacy
diagnostics there are computed \ifanon
\else
with the open-source \textsf{PanelAdequacy} software \citep{PanelAdequacySoftware}
(the \texttt{eiv\_adequacy} routine)\fi;
every within-reliability, breakdown-reliability, and i.i.d.\ implied-size entry in
Table~\ref{tab-vdem} reproduces from the V-Dem panel bundled with the package.
The cluster-robust columns ($\hpsi$, size$_{CR}$) are produced by the
supplementary replication archive. The twins calculations are separate and
reproduce from the public repeated-report extract and
\texttt{code/twins\_application.R} in that archive.

Lemma~\ref{lem-crve} shows the CRVE must be formed without the
$\tfrac{n-1}{n-K_n}$ small-sample factor, which inflates it by an asymptotic
$1/(1-\rho)$. Applying that factor would raise the reported $\hpsi$ by
$2.5\%$ in the V-Dem panel ($\hrho_n=0.025$). No V-Dem verdict turns on the
difference, but users reproducing these numbers with default software
settings (which apply the factor) should expect the discrepancy. The twins
application has no cluster-robust column and is unaffected by this convention.

\section{Auxiliary results for the joint CLT under the $(\rho,\tau^2)$ drift}
\label{app-aux}

This appendix states and proves the auxiliary results behind Lemma~\ref{lem-clt}, making the paper self-contained. The strategy --- condition on the regressor and the design throughout, so the score is a weighted sum of independent errors with conditionally deterministic weights, and self-normalize by sample quantities --- is what makes the argument work at $\rho>0$. Throughout, $\F_n := \sigma(X_1,\dots,X_n, D_{K_n})$ and $\mathcal D_n := \sigma(D_{K_n}, G_1,\dots,G_n)$, where $G_i$ denotes the cell label of observation $i$; conditional on $\F_n$, the weights $\tX_{K_n,i}$ and the Hessian $V_n := X'M_{K_n}X$ are fixed. By the sampling structure of Assumption~\ref{ass-reg}(i) and strict exogeneity~\eqref{eq-model}, the errors $u_1,\dots,u_n$ are, conditional on $\F_n$, independent with $\E[u_i \mid \F_n] = 0$ and $\Var(u_i \mid \F_n) = \sigma^2(X_i, \G_{K_n,i})$.

Two balance conditions on the primitives are used; both are referenced from the statement of Lemma~\ref{lem-clt}.

\textbf{Condition (B) (treatment balance).} Conditional on $\mathcal D_n$, the deviations $\xi_i := X_i - \E[X_i \mid \G_{K_n}]$ are independent with mean zero, variances $q_i := \Var(\xi_i \mid \mathcal D_n)$, and scaled fourth moments $\E[\xi_i^4 \mid \mathcal D_n] \le \kappa_\xi\, q_i^2$ uniformly in $i,n$; and either (B1) $q_i \equiv Q_{K_n}$ for all $i$ (within-cell homoskedasticity of the treatment), or (B2) $q_i \le c_1 Q_{K_n}$ uniformly, $n^{-1}\sum_i q_i / Q_{K_n} \to_p 1$, and the uniform-leverage condition $\max_i |(P_{K_n})_{ii} - \rho_n| = o_p(1)$ holds (satisfied by balanced or approximately balanced grouped designs; it is the balance condition of the leave-out literature, cf.\ \citealp{KSS}).

\textbf{Condition (C) (error fourth moment).} $\sup_{i,n} \E[u_i^4 \mid \F_n] \le \kappa_u < \infty$. (This strengthens the $(2+\epsilon)$ moment of Assumption~\ref{ass-reg}(vi) to $\epsilon=2$; like (vi), it is implied by the finite eighth moments of Assumption~\ref{ass-reg}(i) when the conditional law of $u_i$ does not degenerate across cells.)

\begin{lemma}[Conditional-to-unconditional transfer]\label{lem-transfer}
Let $T_n$ be statistics and $\F_n$ $\sigma$-fields such that, for every $t \in \R$, $\Pp(T_n \le t \mid \F_n) \to_p \Phi_F(t)$ for a continuous deterministic cdf $\Phi_F$. Then $T_n \Rightarrow \Phi_F$ unconditionally.
\end{lemma}

\begin{proof}
Fix $t$ and set $Z_n:=\Pp(T_n\le t\mid\F_n)$. Since $0\le Z_n\le1$ and $Z_n\to_p\Phi_F(t)$, uniform integrability gives $\E|Z_n-\Phi_F(t)|\to0$. Hence $\Pp(T_n\le t)=\E Z_n\to\Phi_F(t)$; continuity of $\Phi_F$ completes the proof.
\end{proof}

\begin{lemma}[Variance of a quadratic form in independent variables]\label{lem-qf}
Let $B$ be a symmetric $n\times n$ matrix measurable with respect to a $\sigma$-field $\mathcal H$, and let $w_1,\dots,w_n$ be, conditional on $\mathcal H$, independent with mean zero, variances $s_i^2 \le s_{\max}^2$, and fourth moments $\E[w_i^4 \mid \mathcal H] \le \kappa$. Then
\[
\Var\big(w' B w \,\big|\, \mathcal H\big) \;\le\; \kappa \sum_i B_{ii}^2 \;+\; 2\,s_{\max}^4\, \tr(B^2).
\]
\end{lemma}

The bound is deliberately weaker than what is available: quadratic forms of this kind obey a central limit theorem under conditions on the eigenvalues of $B$ \citep{deJong}. Only the variance bound is used here, because every quadratic form below is a denominator required to concentrate rather than a term whose limiting law is needed, and the bound holds with no eigenvalue condition at all.

\begin{proof}
Conditional independence and centering give the standard identity
\[
\Var(w'Bw\mid\mathcal H)
=\sum_iB_{ii}^2\big\{\E[w_i^4\mid\mathcal H]-s_i^4\big\}
+2\sum_{i\ne j}B_{ij}^2s_i^2s_j^2 .
\]
The first sum is at most $\kappa\sum_iB_{ii}^2$ and the second at most $2s_{\max}^4\sum_{i,j}B_{ij}^2=2s_{\max}^4\tr(B^2)$.
\end{proof}

\begin{lemma}[Hessian and residual quadratic forms]\label{lem-hess-aux}
Under Assumption~\ref{ass-reg} and Conditions (B)--(C), with the saturation of Section~\ref{sec-setup} (so $M_{K_n}\E[X\mid\G_{K_n}]=0$ and $X'M_{K_n}X = \xi'M_{K_n}\xi$):
\begin{enumerate}[label=(\alph*)]
\item $X'M_{K_n}X/(nQ_{K_n}) \to_p 1-\rho$;
\item $u'M_{K_n}u/(n-d_{K_n}) - (n-d_{K_n})^{-1}\sum_i (M_{K_n})_{ii}\,\sigma^2(X_i,\G_{K_n,i}) \to_p 0$; under conditional homoskedasticity (or under (B2)-type uniform leverage combined with the law of large numbers for $\sigma^2(X_i,\G_{K_n,i})$) the centering constant converges to $\sigma^2$, giving $u'M_{K_n}u/(n-d_{K_n}) \to_p \sigma^2$.
\end{enumerate}
\end{lemma}

\begin{proof}
(a) Saturation gives $X'M_{K_n}X=\xi'M_{K_n}\xi$. With $h_{ii}:=(P_{K_n})_{ii}$, its conditional mean is
\[
\mu_n:=\sum_i(1-h_{ii})q_i
=\begin{cases}
(1-\rho_n)nQ_{K_n},&\text{under (B1)},\\
(1-\rho_n)\sum_iq_i+\sum_i(\rho_n-h_{ii})q_i,&\text{under (B2)}.
\end{cases}
\]
In the second case the final sum is bounded by $\max_i|h_{ii}-\rho_n|\sum_iq_i=o_p(nQ_{K_n})$, so in either case $\mu_n/(nQ_{K_n})\to_p1-\rho$. Lemma~\ref{lem-qf}, using $M_{K_n}^2=M_{K_n}$ and $\tr(M_{K_n})\le n$, gives
\[
\Var(\xi'M_{K_n}\xi\mid\mathcal D_n)\le C Q_{K_n}^2 n .
\]
Thus $(\xi'M_{K_n}\xi-\mu_n)/(nQ_{K_n})\to_p0$ by conditional Chebyshev.

(b) Conditional on $\F_n$, the mean of $u'M_{K_n}u$ is
$\sum_i(M_{K_n})_{ii}\sigma^2(X_i,\G_{K_n,i})$, while Lemma~\ref{lem-qf} gives conditional variance at most $3\kappa_u n$. Since $n-d_{K_n}=n(1-\rho_n)$ and $\rho<1$, conditional Chebyshev proves the stated centering result. Under conditional homoskedasticity the centered average equals $\sigma^2$ exactly. Under the alternative condition in the statement, uniform leverage replaces $(M_{K_n})_{ii}$ by $1-\rho_n$ up to $o_p(1)$, and the law of large numbers completes the argument.
\end{proof}

\begin{lemma}[Self-normalized score CLT]\label{lem-score-clt}
Let $s_n^2 := \sum_i \tX_{K_n,i}^2\, \sigma^2(X_i,\G_{K_n,i}) = \Var(X'M_{K_n}u \mid \F_n)$. Under Assumption~\ref{ass-reg} and Condition~(B) of \ref{app-aux},
\[
\frac{X'M_{K_n}u}{s_n} \;\Rightarrow\; N(0,1),
\qquad\text{and hence}\qquad
\frac{X'M_{K_n}u}{\sqrt{nQ_{K_n}}} \;\Rightarrow\; N\big(0,\ (1-\rho)\,\omega^2\big),
\]
the second display by Assumption~\ref{ass-reg}(iv) ($s_n^2/(nQ_{K_n}) \to_p (1-\rho)\omega^2 > 0$) and Slutsky.
\end{lemma}

\begin{proof}
Conditional on $\F_n$, $X'M_{K_n}u = \sum_i \tX_{K_n,i} u_i$ is a sum of independent mean-zero terms with deterministic weights and total variance $s_n^2$. With $\epsilon$ from Assumption~\ref{ass-reg}(vi), the conditional Lyapunov ratio of order $2+\epsilon$ satisfies
\[
\begin{aligned}
L_n &:= \frac{\sum_i |\tX_{K_n,i}|^{2+\epsilon}\, \E[|u_i|^{2+\epsilon} \mid \F_n]}{s_n^{2+\epsilon}} \\[3pt]
&\;\le\; C_u\,\Big(\frac{\max_i \tX_{K_n,i}^2}{nQ_{K_n}}\Big)^{\epsilon/2} \cdot \frac{X'M_{K_n}X/(nQ_{K_n})}{\big(s_n^2/(nQ_{K_n})\big)^{(2+\epsilon)/2}} \;=\; o_p(1),
\end{aligned}
\]
using Assumption~\ref{ass-reg}(iii), Lemma~\ref{lem-hess-aux}(a), and Assumption~\ref{ass-reg}(iv), respectively. The conditional Lyapunov theorem (with the usual subsequence argument for $L_n\to_p0$) therefore gives
$\Pp(X'M_{K_n}u/s_n\le x\mid\F_n)\to_p\Phi(x)$ for every $x$; Lemma~\ref{lem-transfer} yields the unconditional limit.
\end{proof}

\begin{proof}[Proof of Lemma~\ref{lem-clt}]
The first coordinate is Lemma~\ref{lem-score-clt}; the second is Lemma~\ref{lem-hess-aux}(a); the third is Lemma~\ref{lem-hess-aux}(b). Under conditional homoskedasticity $\sigma^2(X_i,\G_{K_n,i}) \equiv \sigma^2$, so $s_n^2 = \sigma^2\, X'M_{K_n}X$ and the score variance limit is $(1-\rho)\sigma^2$, i.e.\ $\omega^2 = \sigma^2$.
\end{proof}

\section{Proofs}
\label{app-proofs}

This appendix collects the proofs of the results stated in the main text; all cross-references are to the numbered results there.

\subsection*{Proof of Lemma~\ref{lem-atten} (attenuation under the drift)}
\label{app-atten}

\begin{proof}
Write
\[
H_n:=X'M_{K_n}X,\qquad C_n:=X'M_{K_n}\nu,\qquad
R_n:=\nu'M_{K_n}\nu .
\]
In the weak-information regime, Lemma~\ref{lem-clt} gives
$H_n\to_p(1-\rho)\tau^2$. Conditional on $(X,\G_{K_n})$,
$\E[C_n\mid X,\G_{K_n}]=0$ and
$\Var(C_n\mid X,\G_{K_n})=\sigma_{\nu,n}^2H_n=o_p(1)$, so $C_n=o_p(1)$.
Moreover,
\[
\E[R_n\mid\G_{K_n}]=c^2(1-\rho_n),\qquad
\Var(R_n\mid\G_{K_n})\le(C_\nu+2)c^4(1-\rho_n)/n
\]
by Lemma~\ref{lem-qf}; hence $R_n\to_p c^2(1-\rho)$. Since
$X^{*\prime}M_{K_n}X^*=H_n+2C_n+R_n$, the Hessian and reliability limits in part~(a) follow immediately.

Under strong information, the same calculations give
$C_n/(nQ_{K_n})\to_p0$ and $R_n/(nQ_{K_n})\to_p0$. Thus both
$X^{*\prime}M_{K_n}X$ and $X^{*\prime}M_{K_n}X^*$ equal
$H_n+o_p(nQ_{K_n})$. Also
\[
X'M_{K_n}u=O_p(\sqrt{nQ_{K_n}}),\qquad
\nu'M_{K_n}u=O_p(1),
\]
the first by Lemma~\ref{lem-clt} and the second because its conditional variance is
$\sigma_{\nu,n}^2u'M_{K_n}u=O_p(1)$. Substitution into
\[
\hbeta_K^*=\beta_0\frac{X^{*\prime}M_{K_n}X}{X^{*\prime}M_{K_n}X^*}
+\frac{X^{*\prime}M_{K_n}u}{X^{*\prime}M_{K_n}X^*}
\]
proves part~(b). Corollary~\ref{cor-slope}, placed after Theorem~\ref{thm-noncentral}, records the nondegenerate weak-information limit.
\end{proof}

\subsection*{Proof of Lemma~\ref{lem-lev-star} (leverage condition for the contaminated regressor)}
\label{app-lev}

\begin{proof}
Let $m_i'$ be row $i$ of $M_{K_n}$. Conditional on $\G_{K_n}$,
$\widetilde\nu_{K_n,i}=m_i'\nu$ is a sum of independent mean-zero terms. Because $M_{K_n}$ is an orthogonal projection,
\[
\sum_jm_{ij}^2=(M_{K_n})_{ii}\le1,\qquad
\sum_j|m_{ij}|^{2r}\le\sum_jm_{ij}^2\le1 .
\]
Rosenthal's inequality and Assumption~\ref{ass-nu}(ii) therefore give, uniformly in $i$,
\[
\E[|\widetilde\nu_{K_n,i}|^{2r}\mid\G_{K_n}]
\le C_r\!\left[
\left(\sigma_{\nu,n}^2\sum_jm_{ij}^2\right)^r
+\sum_j|m_{ij}|^{2r}\E(|\nu_j|^{2r}\mid\G_{K_n})
\right]
\le C_r(1+C_\nu)\sigma_{\nu,n}^{2r}.
\]
Markov's inequality and a union bound now yield
\[
\max_i\widetilde\nu_{K_n,i}^2
=O_p(\sigma_{\nu,n}^2n^{1/r})
=O_p(c^2n^{1/r-1})=o_p(1).
\]
Finally,
$\max_i\widetilde X^{*2}_{K_n,i}
\le2\max_i\tX_{K_n,i}^2+2\max_i\widetilde\nu_{K_n,i}^2=o_p(1)$
by Assumption~\ref{ass-reg}(iii) and Lemma~\ref{lem-clt}, while
$X^{*\prime}M_{K_n}X^*\to_p(1-\rho)(\tau^2+c^2)>0$ by Lemma~\ref{lem-atten}. Their ratio therefore converges to zero.
\end{proof}

\subsection*{Proof of Corollary~\ref{cor-slope} (limiting law of the slope estimator)}
\label{app-slope}

\begin{proof}
The exact identity
\[
\hbeta_K^*
=\frac{\beta_0X^{*\prime}MX+X^{*\prime}Mu}{X^{*\prime}MX^*}
\]
combines with Lemma~\ref{lem-atten} and the score limit proved in Theorem~\ref{thm-noncentral}:
$X^{*\prime}MX\to_p(1-\rho)\tau^2$,
$X^{*\prime}MX^*\to_p\tau^{*2}$, and
$X^{*\prime}Mu\Rightarrow N(0,\sigma^2\tau^{*2})$. Hence Slutsky gives
\[
\hbeta_K^*\Rightarrow
\beta_0\frac{(1-\rho)\tau^2}{\tau^{*2}}
+N(0,\sigma^2/\tau^{*2})
=\beta_0\lambda+N(0,\sigma^2/\tau^{*2}),
\]
because $\tau^{*2}=(1-\rho)(\tau^2+c^2)$. This use of Theorem~\ref{thm-noncentral} is acyclic: that theorem does not invoke the present corollary.
\end{proof}

\subsection*{Proof of Proposition~\ref{prop-power} (local power)}
\label{app-power}

\begin{proof}
Under $\beta_n$, the decomposition in Theorem~\ref{thm-noncentral} acquires only the standardized term
\[
\frac{b}{\sigma}\,
\frac{X^{*\prime}MX}{\sqrt{nQ_{K_n}}\sqrt{X^{*\prime}MX^*}}.
\]
The cross-term bound and Lemmas~\ref{lem-clt}--\ref{lem-atten} make this converge to
\[
\frac{b}{\sigma}\,
\frac{(1-\rho)\tau}{\sqrt{(1-\rho)(\tau^2+c^2)}}
=\frac{b}{\sigma}\sqrt{(1-\rho)\lambda}.
\]
The added regression-error component has squared $M$-norm
$O_p(X'MX)=O_p(1)$, so the residual-scale limit is unchanged. Adding the displayed deterministic shift to Theorem~\ref{thm-noncentral} proves the result.
\end{proof}

\subsection*{Proof of Proposition~\ref{prop-equicorr} (equicorrelated measurement error)}
\label{app-equicorr}

\begin{proof}
Let $g = \sum_i g_i\,\mathbf 1_i$ be the common-component vector, where $\mathbf 1_i$ is the indicator of unit $i$. Each $\mathbf 1_i$ is a column of the unit-dummy matrix, which lies in the fixed-effect column space, so $M_{K_n}\mathbf 1_i = 0$ and hence $M_{K_n} g = 0$. Therefore the unit fixed effect \emph{annihilates the entire common factor, not merely its expectation}: writing $\nu = \sigma_{\nu,n}(\sqrt\omega\, g + \sqrt{1-\omega}\,\epsilon)$,
\[
M_{K_n}\nu = \sigma_{\nu,n}\sqrt{1-\omega}\; M_{K_n}\epsilon \qquad\text{identically (not just in mean).}
\]
Every proof input involving $\nu$ depends on it through $M_{K_n}\nu$. The preceding identity therefore reduces Lemmas~\ref{lem-atten}--\ref{lem-lev-star} and Theorem~\ref{thm-noncentral} to the i.i.d.\ $\epsilon$ case with $c^2$ replaced by $q=(1-\omega)c^2$, which gives the stated $\eta_\omega$. The covariance formula follows directly from the factor model. Finally,
\[
\frac{d}{dq}\frac{q}{\sqrt{\tau^2+q}}
=\frac{\tau^2+q/2}{(\tau^2+q)^{3/2}}>0;
\]
since $q=(1-\omega)c^2$ decreases with $\omega$, so does $|\eta_\omega|$.
\end{proof}

\subsection*{Proof of Theorem~\ref{thm-no-distortion} (no $\tau^2$-driven size distortion)}
\label{app-nodist}

\begin{proof}
Work conditional on $(D_{K_n},X)$ and set
\[
q:=\frac{M_{K_n}X}{\sqrt{X'M_{K_n}X}},\qquad
R:=M_{K_n}-qq'=M_{[K,X]} .
\]
$R$ is an orthogonal projector of rank $k:=n-d_{K_n}-1$ and $Rq=0$. Cochran's theorem gives
\[
\frac{q'u}{\sigma}\sim N(0,1),\qquad
\frac{u'Ru}{\sigma^2}\sim\chi_k^2,
\qquad q'u\ \perp\ u'Ru .
\]
Since $q'u/\sigma=(\hbeta_{K_n}-\beta_0)\sqrt{X'M_{K_n}X}/\sigma$ and
$u'Ru/k=\widehat\sigma^2_{\mathrm{CJN}}$, their ratio is exactly
$T_n^{\mathrm{CJN}}\sim t_k$. Using $R$, rather than $M_{K_n}$, is essential because it removes the $q$ direction from the residual quadratic form.
The Fisher--Cornish expansion \citep{FisherCornish} of the Student-$t$ distribution \citep[26.7.8]{AS72} gives the displayed size distortion. Since $k = n - d_{K_n} - 1 = n(1 - \rho_n) - 1 \to \infty$ uniformly for $\rho \in [0, \bar\rho]$ with $\bar\rho < 1$, the rate is $O(1/n)$ regardless of $\tau_n^2$.
\end{proof}

\subsection*{Proof of Theorem~\ref{thm-noncentral} (non-centrality under EIV)}
\label{app-noncentral}

\begin{proof}
Write $M=M_{K_n}$ and $H_n^*:=X^{*\prime}MX^*$. Under $H_0$,
\begin{equation}\label{eq-direct-t-decomp}
T_n^{\mathrm{CJN}*}(\beta_0)
=
\underbrace{\frac{X^{*\prime}Mu}{\sigma\sqrt{H_n^*}}}_{W_n}
\frac{\sigma}{\hsigma^*_{\mathrm{CJN}}}
-
\frac{\beta_0X^{*\prime}M\nu}
{\hsigma^*_{\mathrm{CJN}}\sqrt{H_n^*}} .
\end{equation}
We establish the three limits in this decomposition.

First, Assumption~\ref{ass-nu}(i) and $u\perp X\mid\G_{K_n}$ together imply
$u\perp(X,\nu)\mid\G_{K_n}$. Conditional on $(X^*,\G_{K_n})$, therefore,
\[
W_n=\sum_iw_{ni}\frac{u_i}{\sigma},\qquad
w_{ni}:=\frac{\widetilde X^*_{K_n,i}}{\sqrt{H_n^*}},
\qquad \sum_iw_{ni}^2=1 ,
\]
is a sum of independent mean-zero terms of total variance one. Its conditional Lyapunov ratio is at most
$C_u\max_i|w_{ni}|^\epsilon=o_p(1)$ by Assumption~\ref{ass-reg}(vi) and Lemma~\ref{lem-lev-star}. The conditional Lyapunov theorem and Lemma~\ref{lem-transfer} give $W_n\Rightarrow N(0,1)$.

Second, Lemma~\ref{lem-atten} gives
\[
H_n^*\to_p\tau^{*2}:=(1-\rho)(\tau^2+c^2),\qquad
X^{*\prime}M\nu=X'M\nu+\nu'M\nu\to_p c^2(1-\rho).
\]

Third, let $e:=u-\beta_0\nu$ and $N_n:=X^{*\prime}Me$. Least-squares projection gives
\[
(n-d_{K_n}-1)\hsigma_{\mathrm{CJN}}^{*2}
=e'Me-\frac{N_n^2}{H_n^*}.
\]
The first two limits imply $N_n=O_p(1)$, so the second term is $O_p(1)$. Also
\[
e'Me=u'Mu-2\beta_0\nu'Mu+\beta_0^2\nu'M\nu,
\qquad
|\nu'Mu|\le\sqrt{(\nu'M\nu)(u'Mu)}=O_p(\sqrt n).
\]
Because $\nu'M\nu=O_p(1)$ and, under conditional homoskedasticity,
$u'Mu/(n-d_{K_n})\to_p\sigma^2$ by Lemma~\ref{lem-hess-aux}(b), division by
$n-d_{K_n}-1\asymp n$ yields $\hsigma_{\mathrm{CJN}}^{*2}\to_p\sigma^2$.

Substituting these limits into~\eqref{eq-direct-t-decomp} and applying Slutsky gives
\[
T_n^{\mathrm{CJN}*}(\beta_0)\Rightarrow
N\!\left(
-\frac{\beta_0c^2(1-\rho)}{\sigma\sqrt{(1-\rho)(\tau^2+c^2)}},\,1
\right)
=N(\eta,1).
\]
\end{proof}

\subsection*{Proof of Corollary~\ref{cor-cv} (critical-value formula)}
\label{app-cv}

\begin{proof}
For $T \sim N(\eta, 1)$, $s(\eta) := \Pp(|T| > z) = \Phi(-z - \eta) + 1 - \Phi(z - \eta)$. Then $s'(\eta) = -\phi(z+\eta) + \phi(z-\eta)$, so $s(\eta)$ is even and $s'(0) = 0$; a linear term $2|\eta|\phi(z)$ is impossible for a differentiable even function. Differentiating again, $s''(\eta) = (z+\eta)\phi(z+\eta) + (z-\eta)\phi(z-\eta)$, so $s''(0) = 2z\phi(z)$ and $s(\eta) = \alpha + z\phi(z)\,\eta^2 + O(\eta^4)$. Substituting $\eta^2 = \beta_0^2 c^4(1-\rho) / (\sigma^2(\tau^2 + c^2))$ (Theorem~\ref{thm-noncentral}) and setting the leading distortion equal to $\delta$ gives $\tau^2 + c^2 = \beta_0^2 c^4(1-\rho)\, z\phi(z)/(\sigma^2\delta)$, which rearranges to~\eqref{eq-cv}.
\end{proof}

\subsection*{Proof of Corollary~\ref{cor-feasible} (feasible non-centrality)}
\label{app-feasible}

\begin{proof}
With $\lambda=\tau^2/(\tau^2+c^2)$ and $\tau^{*2}=(1-\rho)(\tau^2+c^2)$ (Lemma~\ref{lem-atten}), $1-\lambda = c^2/(\tau^2+c^2)$, so
\[
(1-\lambda)\sqrt{\tau^{*2}} = \frac{c^2}{\tau^2+c^2}\sqrt{(1-\rho)(\tau^2+c^2)} = \frac{c^2\sqrt{1-\rho}}{\sqrt{\tau^2+c^2}},
\]
and $(|\beta_0|/\sigma)(1-\lambda)\sqrt{\tau^{*2}} = |\beta_0|c^2\sqrt{1-\rho}/(\sigma\sqrt{\tau^2+c^2}) = |\eta|$ by Theorem~\ref{thm-noncentral}: this is the algebraic identity underlying the feasible form. Since $X^{*\prime}M_{K_n}X^* \to_p \tau^{*2}$ by Lemma~\ref{lem-atten}, rearranging the quadratic condition $z\phi(z)\eta^2 \le \delta$ of Corollary~\ref{cor-cv} gives~\eqref{eq-feasible}.
\end{proof}

\subsection*{Proof of Proposition~\ref{prop-pilot} (naive vs.\ corrected pilot)}
\label{app-pilot}

\begin{proof}
(i) follows from Corollary~\ref{cor-slope} (centering of $\hbeta^*_K$ at $\beta_0\lambda$) and the strict monotonicity of $z\phi(z)\eta^2$ in $|\beta_0|$; the second-moment statement refers to the limiting law $\beta_0\lambda+N(0,\sigma^2/\tau^{*2})$ of Corollary~\ref{cor-slope}, whose second moment is $(\beta_0\lambda)^2+\sigma^2/\tau^{*2}$ --- a property of the limit, not a claim of moment convergence (which would require uniform integrability of $\hbeta^{*2}_K$). (ii) With $\hbeta^*_K \Rightarrow \beta_0\lambda + N(0,\sigma^2/\tau^{*2})$ and $\hlambda_n \to_p \lambda > 0$, Slutsky gives $\hbeta_0^{\mathrm{corr}}=\hbeta^*_K/\hlambda_n \Rightarrow \beta_0 + N(0,\sigma^2/(\lambda^2\tau^{*2}))$; under $nQ_{K_n}\to\infty$ the numerator's noise is $o_p(1)$ and the limit is the constant $\beta_0$.
\end{proof}

\subsection*{Proofs for cluster dependence, Lemmas~\ref{lem-nest}--\ref{lem-crve}, Theorem~\ref{thm-cluster} and Corollary~\ref{cor-cluster-feasible}}
\label{app-cluster}

Throughout write $M=M_{K_n}$, $\widetilde X^* = MX^*$, $e := u - \beta_0\nu$ (the error of the analyst's regression of $Y$ on $X^*$ under $H_0$), $\widehat\Delta := \hbeta^*_K-\beta_0$, and recall $a^{(g)}_i = \widetilde X^*_i\mathbf 1\{i\in\mathcal C_g\}$, $A_g=\|a^{(g)}\|^2$, $\sum_gA_g=\tau^{*2}_n$. Two identities are used repeatedly: $M$ is symmetric idempotent, so
\begin{equation}\label{eq-cl-id}
\begin{gathered}
\big(Ma^{(g)}\big)'X^* = a^{(g)\prime}MX^* = a^{(g)\prime}\widetilde X^* = A_g, \\[3pt]
\big(Ma^{(g)}\big)'\nu = a^{(g)\prime}M\nu = \sum_{i\in\mathcal C_g}\widetilde X^*_i\,\widetilde\nu_{K_n,i} .
\end{gathered}
\end{equation}

\begin{proof}[Proof of Lemma~\ref{lem-nest}]
(a) Let $d$ range over the cell indicators spanning the fixed-effect column space. Each is supported in some cluster $h$. If $h\ne g$, then $d'a^{(g)}=0$ by disjoint support; if $h=g$, then $d'a^{(g)}=d'\widetilde X^*=d'MX^*=0$. Thus $a^{(g)}$ is orthogonal to the fixed-effect column space, so $Ma^{(g)}=a^{(g)}$.

(b) The same argument applied to $D^{\mathrm{nest}}$ gives
$M_{\mathrm{nest}}a^{(g)}=a^{(g)}$. If $\rank(\Lambda_n)=0$, then
$M=M_{\mathrm{nest}}$ and the claim is immediate. Otherwise, the Frisch--Waugh--Lovell decomposition gives
\[
M = M_{\mathrm{nest}} - M_{\mathrm{nest}}D^{\mathrm{ne}}\Lambda_n^{-}D^{\mathrm{ne}\prime}M_{\mathrm{nest}},
\qquad \Lambda_n=D^{\mathrm{ne}\prime}M_{\mathrm{nest}}D^{\mathrm{ne}} .
\]
Applying this to $a^{(g)}$ and using $M_{\mathrm{nest}}a^{(g)}=a^{(g)}$,
\[
Ma^{(g)}-a^{(g)} = -\,M_{\mathrm{nest}}D^{\mathrm{ne}}\Lambda_n^{-}v_g ,
\qquad v_g := D^{\mathrm{ne}\prime}a^{(g)} .
\]
Because $\Lambda_n^{-}\Lambda_n\Lambda_n^{-}=\Lambda_n^{-}$,
$\|Ma^{(g)}-a^{(g)}\|^2=v_g'\Lambda_n^{-}v_g
\le\lambda^+_{\min}(\Lambda_n)^{-1}\|v_g\|^2$. Moreover,
\[
\begin{aligned}
\sum_g \|v_g\|^2 &= \tr\Big(D^{\mathrm{ne}\prime}\Big(\sum_g a^{(g)}a^{(g)\prime}\Big)D^{\mathrm{ne}}\Big) \\[3pt]
&\le \Big\|\sum_g a^{(g)}a^{(g)\prime}\Big\|_{\mathrm{op}}\tr\big(D^{\mathrm{ne}\prime}D^{\mathrm{ne}}\big)
= \big(\max_g A_g\big)\,\tr\big(D^{\mathrm{ne}\prime}D^{\mathrm{ne}}\big),
\end{aligned}
\]
because the matrix inside the norm is block diagonal, with rank-one cluster blocks of norms $A_g$. This proves the stated bound. Conditions~(N1)--(N2) make its right-hand side
$O_p(\tau^{*2}_n d_n^{\mathrm{ne}}/G_n)$, proving the calibration. In the balanced two-way panel,
$\Lambda_n=G(I_T-T^{-1}\mathbf1\mathbf1')$, whose non-zero eigenvalues are all $G=n/T$, so (N1) holds exactly.
\end{proof}

\begin{proof}[Proof of Lemma~\ref{lem-cluster-clt}]
Conditional on $\F^*_n$ the cluster scores $\zeta_g=a^{(g)\prime}u$ are independent with mean zero and $\sum_g\Var(\zeta_g\mid\F^*_n)=\Psi_n$. We verify the conditional Lyapunov condition of order $2+\delta$. By Cauchy--Schwarz, $|\zeta_g|\le A_g^{1/2}\big(\sum_{i\in\mathcal C_g}u_i^2\big)^{1/2}$, and by the power-mean inequality $\big(\sum_{k\le m}b_k\big)^{p}\le m^{p-1}\sum_k b_k^p$ with $p=(2+\delta)/2\ge1$,
\[
\E\big[|\zeta_g|^{2+\delta}\,\big|\,\F^*_n\big]
\le A_g^{1+\delta/2}\, n_g^{\delta/2}\sum_{i\in\mathcal C_g}\E\big[|u_i|^{2+\delta}\,\big|\,\F^*_n\big]
\le C_u\, \bar n_n^{1+\delta/2}\,A_g^{1+\delta/2}.
\]
Summing over $g$ and using $\sum_gA_g^{1+\delta/2}\le(\max_gA_g)^{\delta/2}\sum_gA_g = (\max_gA_g)^{\delta/2}\tau^{*2}_n$,
\[
\begin{aligned}
\frac{\sum_g\E[|\zeta_g|^{2+\delta}\mid\F^*_n]}{\Psi_n^{(2+\delta)/2}}
&\;\le\; \frac{C_u\,\bar n_n^{1+\delta/2}\,(\max_gA_g)^{\delta/2}\,\tau^{*2}_n}{\Psi_n^{1+\delta/2}} \\[3pt]
&\;=\; \frac{C_u\,(1+o_p(1))}{(\sigma^2\psi)^{1+\delta/2}}
\left[\bar n_n^{\frac{2+\delta}{\delta}}\,\frac{\max_gA_g}{\tau^{*2}_n}\right]^{\delta/2} ,
\end{aligned}
\]
using Assumption~\ref{ass-cluster}(iii). The bracket is $o_p(1)$ under Assumption~\ref{ass-cluster}(ii), and under growing clusters this is exactly condition (i$'$) of Remark~\ref{rem-clustersize}. Thus the conditional Lyapunov theorem gives
$\Pp(S_n/\sqrt{\Psi_n}\le x\mid\F_n^*)\to_p\Phi(x)$; Lemma~\ref{lem-transfer} gives the unconditional limit.
\end{proof}

\begin{proof}[Proof of Lemma~\ref{lem-crve}]
Since $\widehat u^* = M\big(Y-X^*\hbeta^*_K\big) = Me - \widehat\Delta\,MX^*$ under $H_0$, the identities~\eqref{eq-cl-id} give
\[
\begin{gathered}
\widehat\zeta_g := a^{(g)\prime}\widehat u^* = \big(Ma^{(g)}\big)'e - \widehat\Delta A_g
= \zeta_g + \epsilon_g, \\[3pt]
\epsilon_g := r^{(g)\prime}u \;-\; \beta_0\sum_{i\in\mathcal C_g}\widetilde X^*_i\widetilde\nu_{K_n,i} \;-\; \widehat\Delta A_g ,
\end{gathered}
\]
where $r^{(g)} := Ma^{(g)}-a^{(g)}$ and we used $\big(Ma^{(g)}\big)'u = \zeta_g + r^{(g)\prime}u$. Then $\widehat V^{\mathrm{sc}}_{CR} = \sum_g\zeta_g^2 + 2\sum_g\zeta_g\epsilon_g + \sum_g\epsilon_g^2$, and it suffices to show $\sum_g\zeta_g^2 = \Psi_n + o_p(\tau^{*2}_n)$ and $\sum_g\epsilon_g^2 = o_p(\tau^{*2}_n)$; the cross term is then $o_p(\tau^{*2}_n)$ by Cauchy--Schwarz, and Assumption~\ref{ass-cluster}(iii) converts $o_p(\tau^{*2}_n)$ into $o_p(\Psi_n)$.

\emph{Step 1.} Conditional independence and the fourth-moment bound give
\[
\E\!\left[\sum_g\zeta_g^2\middle|\F_n^*\right]=\Psi_n,\qquad
\Var\!\left(\sum_g\zeta_g^2\middle|\F_n^*\right)
\le\sum_g\E[\zeta_g^4\mid\F_n^*]
\le\kappa_u\bar n_n^2(\max_gA_g)\tau^{*2}_n.
\]
The last bound uses
$\zeta_g^4\le A_g^2(\sum_{i\in\mathcal C_g}u_i^2)^2$ and
$\sum_gA_g^2\le(\max_gA_g)\tau^{*2}_n$. After division by
$(\tau^{*2}_n)^2$, Assumption~\ref{ass-cluster}(ii) makes the variance $o_p(1)$, so conditional Chebyshev yields
$\sum_g\zeta_g^2=\Psi_n+o_p(\tau^{*2}_n)$. With growing clusters, the same conclusion uses condition (ii$'$) of Remark~\ref{rem-clustersize}.

\emph{Step 2.} The cluster CLT and Assumption~\ref{ass-cluster}(iii),(v) imply
\[
\widehat\Delta
=\frac{S_n-\beta_0(X'M\nu+\nu'M\nu)}{X^{*\prime}MX^*}
=O_p(1),
\]
without invoking the i.i.d.\ Corollary~\ref{cor-slope}. Using
$(x+y+z)^2\le3(x^2+y^2+z^2)$, it remains to bound three sums. Projection compatibility, Cauchy--Schwarz, and the no-dominant-cluster condition give, respectively,
\[
\begin{aligned}
\E\!\left[\sum_g(r^{(g)\prime}u)^2\middle|\F_n^*\right]
&\le \bar n_n\kappa_u^{1/2}\sum_g\|r^{(g)}\|^2
=o_p(\tau^{*2}_n),\\
\sum_g\left(\sum_{i\in\mathcal C_g}
\widetilde X_i^*\widetilde\nu_{K_n,i}\right)^2
&\le \bar n_n\max_i\widetilde\nu_{K_n,i}^2\,\tau^{*2}_n
=o_p(\tau^{*2}_n),\\
\widehat\Delta^2\sum_gA_g^2
&\le\widehat\Delta^2(\max_gA_g)\tau^{*2}_n
=o_p(\tau^{*2}_n).
\end{aligned}
\]
Conditional Markov's inequality applies to the first line. Thus
$\sum_g\epsilon_g^2=o_p(\tau^{*2}_n)$. For growing clusters the first two lines are exactly conditions (iii$'$) and (iv$'$) of Remark~\ref{rem-clustersize}.

\emph{Step 3 (the small-sample factor).} Under nesting, $r^{(g)}\equiv0$ by Lemma~\ref{lem-nest}(a), so $\widehat\zeta_g$ loses nothing to the fixed-effect projection and $\widehat V^{\mathrm{sc}}_{CR}$ is already correctly centred; multiplying by $\tfrac{n-1}{n-K_n}\to(1-\rho)^{-1}$ therefore introduces an asymptotic factor $(1-\rho)^{-1}>1$. The factor $\tfrac{G_n}{G_n-1}\to1$ is harmless.
\end{proof}

\begin{proof}[Proof of Theorem~\ref{thm-cluster}]
Under $H_0$, $Y-X^*\beta_0 = u-\beta_0\nu$, so
\[
N_n := X^{*\prime}M(Y-X^*\beta_0) = \underbrace{X^{*\prime}Mu}_{=\,S_n} \;-\; \beta_0\big(X'M\nu+\nu'M\nu\big),
\]
and by Assumption~\ref{ass-cluster}(v) the bracket converges in probability to $c^2(1-\rho)$. This is imposed as a design-side limit rather than quoted from Lemma~\ref{lem-atten}, whose proof uses i.i.d.\ sampling. Lemmas~\ref{lem-cluster-clt}--\ref{lem-crve} give
$S_n/\sqrt{\Psi_n}\Rightarrow N(0,1)$ and
$\widehat V^{\mathrm{sc}}_{CR}/\Psi_n\to_p1$, while Assumption~\ref{ass-cluster}(iii),(v) gives
$\Psi_n\to_p\psi\sigma^2\tau^{*2}$ with
$\tau^{*2}=(1-\rho)(\tau^2+c^2)$. Hence
\[
T^{CR}_n(\beta_0)=\frac{N_n}{\sqrt{\widehat V^{\mathrm{sc}}_{CR}}}
= \frac{S_n}{\sqrt{\Psi_n}}\sqrt{\frac{\Psi_n}{\widehat V^{\mathrm{sc}}_{CR}}}
\;-\; \frac{\beta_0\big(X'M\nu+\nu'M\nu\big)}{\sqrt{\widehat V^{\mathrm{sc}}_{CR}}},
\]
Slutsky gives $T^{CR}_n(\beta_0)\Rightarrow N(0,1) - \beta_0c^2(1-\rho)\big/\sqrt{\psi\sigma^2(1-\rho)(\tau^2+c^2)}$, and the shift equals $-\beta_0c^2\sqrt{1-\rho}\big/(\sigma\sqrt{\psi}\sqrt{\tau^2+c^2})=\eta/\sqrt\psi$.
\end{proof}

\begin{proof}[Proof of Corollary~\ref{cor-cluster-feasible}]
(a) The identity $\hsigma^{*2}_{\mathrm{CJN}}\hpsi = \widehat V^{\mathrm{sc}}_{CR}/\tau^{*2}_n$ is immediate from $\hpsi := \tau^{*2}_n\widehat V_{CR}/\hsigma^{*2}_{\mathrm{CJN}}$ and $\widehat V_{CR}=\widehat V^{\mathrm{sc}}_{CR}/\tau^{*4}_n$: the factor $\hsigma^{*2}_{\mathrm{CJN}}$ appears once in each of numerator and denominator and cancels before any limit is taken. Lemma~\ref{lem-crve} gives $\widehat V^{\mathrm{sc}}_{CR}=\Psi_n(1+o_p(1))$ and Assumption~\ref{ass-cluster}(iii) gives $\Psi_n/\tau^{*2}_n\to_p\psi\sigma^2$, whence $\widehat s^2_{CR}\to_p\psi\sigma^2$. For the parenthetical claims, if $\hsigma^{*2}_{\mathrm{CJN}}\to_p\varsigma^2\in(0,\infty)$ then $\hpsi = \widehat s^2_{CR}/\hsigma^{*2}_{\mathrm{CJN}}\to_p\psi\sigma^2/\varsigma^2$ by Slutsky. The existence of that limit is an assumption about the error array beyond Assumption~\ref{ass-cluster}, which constrains only the cluster score variance $\Psi_n$ and not the diagonal quadratic form $\widehat u^{*\prime}\widehat u^*$; it is stated conditionally for that reason. Nothing else in the paper depends on it, since the feasible recipe uses $\widehat s_{CR}$, in which $\hsigma^{*2}_{\mathrm{CJN}}$ has already cancelled.

(b) $\hlambda_n=1-\widehat a_n/\tau^{*2}_n$ and $\tau^{*2}_n$ are functionals of $(X,\nu,D_{K_n})$ alone; their probability limits $\lambda$ and $\tau^{*2}$ are those of Lemma~\ref{lem-atten} and are unaffected by dependence in $u$. With (a) and $|\beta_0|$ (or any deterministic $b$) fixed, Slutsky gives the display, and the last equality is Theorem~\ref{thm-cluster} together with $|\eta| = (|\beta_0|/\sigma)(1-\lambda)\sqrt{\tau^{*2}}$ from Corollary~\ref{cor-feasible}. Note that no claim is made about substituting an estimated $\beta_0$; that is Remark~\ref{rem-plugin}.

(c) By definition $t^*_n = |\hbeta^*_K|\sqrt{\tau^{*2}_n}/\hsigma^*_{\mathrm{CJN}}$ and $t^{CR}_n = |\hbeta^*_K|/\sqrt{\widehat V_{CR}}$. Hence
\[
\frac{t^*_n}{\sqrt{\hpsi}}
= \frac{|\hbeta^*_K|\sqrt{\tau^{*2}_n}}{\hsigma^*_{\mathrm{CJN}}}\cdot\frac{\hsigma^*_{\mathrm{CJN}}}{\sqrt{\tau^{*2}_n\widehat V_{CR}}}
= \frac{|\hbeta^*_K|}{\sqrt{\widehat V_{CR}}} = t^{CR}_n ,
\]
an identity in the sample. Since $\lambda^{\dagger}$ of Definition~\ref{def-breakdown} is a function of the $t$-statistic alone, substituting $t^{CR}_n$ gives $\lambda^{\dagger}_{CR}$.
\end{proof}

\begin{proof}[Proof of Remark~\ref{rem-plugin}]
The identity for $\widehat\Delta$ in the proof of Lemma~\ref{lem-crve}, followed by the cluster CLT and Assumption~\ref{ass-cluster}(iii),(v), gives
\[
\hbeta_K^*\Rightarrow\beta_0\lambda+N(0,\psi\sigma^2/\tau^{*2}),
\qquad
\hbeta_0^{\mathrm{corr}}\Rightarrow
B:=\beta_0+N(0,\psi\sigma^2/(\lambda^2\tau^{*2})).
\]
Jointly,
$(\hlambda_n,\tau^{*2}_n,\widehat s_{CR})
\to_p(\lambda,\tau^{*2},\sigma\sqrt\psi)$, so the continuous mapping theorem yields
\[
|\widehat\eta_{CR}| \Rightarrow \frac{|B|}{\sigma\sqrt\psi}(1-\lambda)\sqrt{\tau^{*2}} ,
\]
equal to $(|B|/|\beta_0|)|\eta_{CR}|$ when $\beta_0\ne0$ and to the stated folded-normal limit when $\beta_0=0$. It is nondegenerate when $\tau^{*2}<\infty$ and $\lambda<1$. Under strong information $B$ collapses to $\beta_0$; when $c^2=0$, instead, $B$ remains nondegenerate but $(1-\hlambda_n)\to_p0$ annihilates it. These are exactly the two consistency boundaries stated in the remark.
\end{proof}

\begin{proof}[Proof of Proposition~\ref{prop-certificate}]
Write $s_n:=\widehat s_{CR}/\sqrt{\tau^{*2}_n}$ for the reported standard error of $\hbeta_K^*$ and $\theta:=\lambda\beta_0$.  The cluster-score CLT and CRVE consistency give
\[
\frac{\hbeta_K^*-\theta}{s_n}\Rightarrow N(0,1),
\qquad
s_n\to_p s:=\frac{\sigma\sqrt\psi}{\sqrt{\tau^{*2}}}>0.
\]
Hence $U_{\theta,n}:=|\hbeta_K^*|+z_{1-\gamma_\beta}s_n$ is an asymptotic one-sided upper confidence bound for $|\theta|$.  Indeed, if $\theta\ge0$, then $|\hbeta_K^*|\ge\hbeta_K^*$ and
\[
\limsup_n\Pp(|\theta|>U_{\theta,n})
\le \Pp(Z<-z_{1-\gamma_\beta})=\gamma_\beta;
\]
the case $\theta<0$ is symmetric, and $\theta=0$ is immediate.

Let $f(x):=(1-x)/x$, which is decreasing on $(0,1]$.  Since
\[
\frac{U_{\theta,n}}{s_n}f(\ell_n)
=\big(t_n^{CR}+z_{1-\gamma_\beta}\big)f(\ell_n)
=\widehat\eta_n^{\,U}(\ell_n),
\qquad
|\eta_{CR}|=\frac{|\theta|}{s}f(\lambda),
\]
the preceding coverage statement and Slutsky's theorem prove part~(a) when
$\ell_n=\hlambda_n\to_p\lambda$.  For part~(b), on the intersection of the
coefficient-coverage event $\{|\theta|\le U_{\theta,n}\}$ and the
reliability-coverage event $\{\ell_n\le\lambda\}$, monotonicity of $f$ gives
$\widehat\eta_n^{\,U}(\ell_n)\ge|\eta_{CR}|+o_p(1)$.  The union bound therefore yields
\[
\limsup_n\Pp\!\left(\widehat\eta_n^{\,U}(\ell_n)<|\eta_{CR}|\right)
\le\gamma_\beta+\gamma_\lambda,
\]
without an independence condition.

Finally, $\widehat\eta_n^{\,U}(\ell_n)\le\eta^\dagger$ is equivalent, by direct rearrangement, to
\[
\ell_n\ge
\frac{t_n^{CR}+z_{1-\gamma_\beta}}
{t_n^{CR}+z_{1-\gamma_\beta}+\eta^\dagger}
=\lambda^\dagger_{\gamma_\beta,n}.
\]
If true asymptotic size exceeds $\alpha+\delta$, Theorem~\ref{thm-cluster} and the definition of $\eta^\dagger$ imply $|\eta_{CR}|>\eta^\dagger$; certification can then occur only when the upper bound misses, whose limiting probability is bounded above as stated.  If true size is admissible, false certification is impossible.  The i.i.d.\ proof is identical with the corresponding reported $t$-statistic.
\end{proof}

\section{Additional simulation designs}
\label{app-sim}

The first two designs supplement the core mapping validations (Designs~1 and~4) reported in Section~\ref{sec-sim}. The DGP is identical: two-way FE, $X_{it}=s\,\varepsilon_{it}$ with $s=\sqrt{\tau^2/n}$ (so $nQ_{K_n}=\tau^2$ exactly), $u_{it}\sim N(0,\sigma^2)$, $\nu_{it}\sim N(0,\sigma_\nu^2)$, $4{,}000$--$5{,}000$ replications per cell. Design~6 uses the same FE transformation but increases the replication count to $20{,}000$ to resolve a $5\%$ tail probability.

\paragraph{Design 2 (threshold accuracy).} For each cell we sweep $\tau^2$, locate (by interpolation of the empirical size curve) the $\tau^2$ at which size crosses $5\%+\delta$, and compare it to the corrected quadratic threshold~\eqref{eq-cv} and the exact-inversion threshold~\eqref{eq-cv-exact}. At $\delta=0.05$, over the six cells $(N,T)\in\{(200,5),(1000,5),(200,10)\}\times c^2\in\{2,5\}$, the ratio (empirical)/(quadratic) ranges from $0.98$ to $1.09$ (mean $1.03$) and (empirical)/(exact) from $0.95$ to $1.06$, with no systematic dependence on $(N,T,c^2)$: the corrected closed form locates the true size boundary to within Monte Carlo and grid accuracy. By contrast the discarded linear surrogate $\delta/(2\phi)$ over-demands the residual signal by a factor of $2.5$--$2.8$ across the same cells, and, being anti-conservative at larger $\delta$ (Remark~\ref{rem-exact-cv}), is not a safe approximation. We report the exact-inversion threshold~\eqref{eq-cv-exact} as headline, with the quadratic form~\eqref{eq-cv} as its accurate closed-form companion.

\paragraph{Design 3 (pilot sensitivity).} (a) On a genuinely-failing specification with strong attenuation ($\lambda_n=0.50$, true $|\eta|=0.92$, exact size $15\%$, empirical $16\%$), the naive plug-in of the attenuated coefficient $\hbeta^*_K=\lambda\beta_0$ reports $|\eta|=0.46$ and \emph{passes the point diagnostic}, while the within-reliability-corrected pilot $\hbeta_0^{\mathrm{corr}}=\hbeta^*_K/\hlambda_n$ recovers $|\eta|=0.92$ and correctly flags it. This is Proposition~\ref{prop-pilot}(i) on simulated data: the naive diagnostic passes a specification whose true size is $16\%$. (b) Misspecifying the noise pilot by a factor $k\in\{0.5,0.75,1.25,2\}$ maps monotonically and boundedly into the distortion estimate ($|\hat\eta|/|\eta| = 0.58, 0.80, 1.18, 1.63$ for $k = 0.5, 0.75, 1.25, 2$), motivating the breakdown-reliability band (Definition~\ref{def-breakdown}) over a single point verdict.

\paragraph{Design 6 (false-certification control).} Every DGP is placed just outside the admissible region: $|\eta|=\eta^\dagger+0.01=0.6624$, giving exact-normal size $0.1016>0.10$. Equivalently, true reliability lies $0.0014$--$0.0024$ below the oracle population counterpart of the point breakdown reliability. Every certificate issued is therefore false. Table~\ref{tab:supp-design6} reports the issuance frequency. With reliability known and $\gamma_\beta=0.05$, the rate is at most $0.0474$. The columns headed ``point'' add independent Gaussian pilot noise and then incorrectly treat the noisy estimate as known; at pilot standard error $0.020$, this shortcut certifies as often as $0.1935$. The valid alternative uses a one-sided $97.5\%$ lower reliability bound and $\gamma_\beta=\gamma_\lambda=0.025$, preserving the total $0.05$ budget by Proposition~\ref{prop-certificate}; every realized rate is below $0.016$. The conservative slack reflects both the Bonferroni split and the deliberately one-sided lower reliability bound.

\begin{table}[htbp]\centering
\caption{Design 6: false-certification frequencies. Each cell has $20{,}000$ replications and true exact-normal size $0.1016$, just above the admissible $0.10$. ``Known'' uses the true reliability with $\gamma_\beta=0.05$. ``Point'' treats an independently noisy reliability estimate as known and is not a valid certificate. ``Lower bound'' uses $\gamma_\beta=\gamma_\lambda=0.025$; the column header gives the reliability pilot's standard error.}\label{tab:supp-design6}
\begin{adjustbox}{max width=\textwidth}
\begin{tabular}{rrrrrccccc}
\toprule
$N$ & $T$ & $\rho_n$ & $\tau^2$ & $\lambda$ & known & point .005 & lower bound .005 & point .020 & lower bound .020 \\
\midrule
$120$ & $5$  & $0.207$ & $2$ & $0.80$ & $0.0474$ & $0.0479$ & $0.0154$ & $0.0598$ & $0.0061$ \\
$240$ & $5$  & $0.203$ & $8$ & $0.80$ & $0.0448$ & $0.0450$ & $0.0146$ & $0.0601$ & $0.0047$ \\
$120$ & $10$ & $0.108$ & $2$ & $0.80$ & $0.0441$ & $0.0456$ & $0.0157$ & $0.0610$ & $0.0047$ \\
$240$ & $10$ & $0.104$ & $8$ & $0.90$ & $0.0427$ & $0.0526$ & $0.0069$ & $0.1935$ & $0.0032$ \\
\bottomrule
\end{tabular}
\end{adjustbox}
\begin{minipage}{0.96\textwidth}\footnotesize
\emph{Notes:} The script \texttt{code/verify\_certificate.jl} writes the unrounded values to \texttt{results/design6\_certificate.csv}. The lower-bound rule does not require independence between coefficient and reliability uncertainty; independence is used in this simulation only to make the failure of the point-pilot shortcut transparent.
\end{minipage}
\end{table}

\end{document}